\documentclass[letterpaper,11pt]{article}

\usepackage[%dvips,
            CJKbookmarks=true,
            bookmarksnumbered=true,
            bookmarksopen=true,
            colorlinks=true,
            citecolor=red,
            linkcolor=blue,
            anchorcolor=red,
            urlcolor=blue,
            pdfauthor={Kiryung Lee, Sohail Bahmani, Yonina Eldar, and Justin Romberg},
            pdfstartview=FitH,
            ]{hyperref}

\usepackage{bookmark}

\usepackage{amsmath}
\usepackage{amsxtra,amssymb,amsthm,amsfonts}
\usepackage{mathtools}
\usepackage{mathrsfs}
\usepackage{dsfont}
\usepackage[usenames]{color}
\usepackage{umoline}
\usepackage[margin=1in]{geometry}
\usepackage{setspace}
\usepackage{bbm}
\usepackage{bm}

\usepackage{enumerate}
\usepackage{graphicx}
\usepackage{epstopdf}
\usepackage{subcaption}
\usepackage{amssymb}
\usepackage{array}
\usepackage{multirow}

\usepackage{cleveref}
\crefname{equation}{}{}

\usepackage{mathdots}

\usepackage{cite}

\newtheorem{lemma}{Lemma}[section]

\newtheorem{theorem}[lemma]{Theorem}

\newtheorem{rem}[lemma]{Remark}

\newcommand{\re}{\begin{rem}\rm}
\newcommand{\mar}{\end{rem}}

\newcommand{\fo}{\begin{eqnarray*}}
\newcommand{\mel}{\end{eqnarray*}}

\newcommand{\A}{{\mathcal A}}

\newcommand{\qd}{\end{proof}\vspace{0.5ex}}

\makeatletter
\newcommand{\opnorm}{\@ifstar\@opnorms\@opnorm}
\newcommand{\@opnorms}[1]{%
  \left|\mkern-1.5mu\left|\mkern-1.5mu\left|
   #1
  \right|\mkern-1.5mu\right|\mkern-1.5mu\right|
}
\newcommand{\@opnorm}[2][]{%
  \mathopen{#1|\mkern-1.5mu#1|\mkern-1.5mu#1|}
  #2
  \mathclose{#1|\mkern-1.5mu#1|\mkern-1.5mu#1|}
}
\makeatother

\newcommand{\norm}[1]{\Vert#1\Vert}

\newcommand{\transpose}{\top}

\newcommand{\vect}{\mathrm{vec}}

\newcommand{\R}{\mathbb{R}}
\newcommand{\E}{\operatorname{\mathbb{E}}}

\newcommand{\vct}[1]{\boldsymbol{#1}}
\newcommand{\mtx}[1]{\boldsymbol{#1}}
\newcommand{\<}{\langle}
\renewcommand{\>}{\rangle}
\DeclareMathOperator*{\minimize}{\text{minimize}}

\newcommand{\va}{\vct{a}}
\newcommand{\vb}{\vct{b}}

\newcommand{\ve}{\vct{e}}

\newcommand{\vg}{\vct{g}}
\newcommand{\vh}{\vct{h}}

\newcommand{\vu}{\vct{u}}
\newcommand{\vv}{\vct{v}}

\newcommand{\vx}{\vct{x}}
\newcommand{\vy}{\vct{y}}

\newcommand{\vlambda}{\vct{\lambda}}

\newcommand{\vphi}{\vct{\phi}}

\newcommand{\vzero}{\vct{0}}

\newcommand{\mA}{\mtx{A}}
\newcommand{\mB}{\mtx{B}}

\newcommand{\mD}{\mtx{D}}
\newcommand{\mE}{\mtx{E}}
\newcommand{\mF}{\mtx{F}}
\newcommand{\mG}{\mtx{G}}
\newcommand{\mH}{\mtx{H}}

\newcommand{\mM}{\mtx{M}}

\newcommand{\mP}{\mtx{P}}
\newcommand{\mQ}{\mtx{Q}}
\newcommand{\mR}{\mtx{R}}

\newcommand{\mT}{\mtx{T}}
\newcommand{\mU}{\mtx{U}}
\newcommand{\mV}{\mtx{V}}

\newcommand{\mX}{\mtx{X}}
\newcommand{\mY}{\mtx{Y}}
\newcommand{\mZ}{\mtx{Z}}

\newcommand{\mDelta}{\mtx{\Delta}}

\newcommand{\mLambda}{\mtx{\Lambda}}

\newcommand{\mSigma}{\mtx{\Sigma}}
\newcommand{\mUpsilon}{\mtx{\Upsilon}}
\newcommand{\mPhi}{\mtx{\Phi}}
\newcommand{\mPsi}{\mtx{\Psi}}

\newcommand{\mId}{{\bf I}}

\renewcommand{\hat}[1]{\widehat{#1}}
\renewcommand{\tilde}[1]{\widetilde{#1}}

\usepackage[displaymath, mathlines, pagewise]{lineno}
\begin{document}
\doublespacing

\title{Phase Retrieval of Low-Rank Matrices by Anchored Regression}

\author{
%%%% First author details
{\sc Kiryung Lee}$^*$\\[2pt]
Department of Electrical and Computer Engineering, \\[2pt]
The Ohio State University, Columbus, OH 43210, USA\\
$^*${Corresponding author: lee.8763@osu.edu}\\[6pt]
%%%%%%% Second author details
{\sc Sohail Bahmani}\\[2pt]
School of Electrical and Computer Engineering, \\[2pt]
Georgia Institute of Technology, Atlanta, GA 30332, USA\\
{sohail.bahmani@ece.gatech.edu}\\[2pt]
%%%%%%% Third author details
{\sc Yonina C. Eldar}\\[2pt]
Department of Computer Science and Applied Mathematics, \\[2pt]
Weizmann Institute of Science, Rehovot 7610001, Israel\\
{yonina.eldar@weizmann.ac.il}\\[2pt]
%%%%%%% Fourth author details
{\sc Justin Romberg}\\[2pt]
School of Electrical and Computer Engineering, \\[2pt]
Georgia Institute of Technology, Atlanta, GA 30332, USA\\
{jrom@ece.gatech.edu}\\[2pt]
}
\maketitle

\clearpage

\begin{abstract}
  We study the low-rank phase retrieval problem, where we try to recover a $d_1\times d_2$ low-rank matrix from a series of phaseless linear measurements.  This is a fourth-order inverse problem, as we are trying to recover factors of matrix that have been put through a quadratic nonlinearity after being multiplied together.

  We propose a solution to this problem using the recently introduced technique of anchored regression.  This approach uses two different types of convex relaxations: we replace the quadratic equality constraints for the phaseless measurements by a search over a polytope, and enforce the rank constraint through nuclear norm regularization.  The result is a convex program that works in the space of $d_1 \times d_2$ matrices.

  We analyze two specific scenarios.  In the first, the target matrix is rank-$1$, and the observations are structured to correspond to a phaseless blind deconvolution.  In the second, the target matrix has general rank, and we observe the magnitudes of the inner products against a series of independent Gaussian random matrices.  In each of these problems, we show that the anchored regression returns an accurate estimate from a near-optimal number of measurements given that we have access to an anchor matrix of sufficient quality.  We also show how to create such an anchor in the phaseless blind deconvolution problem, again from an optimal number of measurements, and present a partial result in this direction for the general rank problem.
\end{abstract}

\clearpage

\section{Introduction}

We consider the problem of recovering a low-rank matrix $\mX_\sharp$ from phaseless linear measurements of the form
\begin{equation}
	\label{eq:gen_mdl}
	y_m = |\<\mPhi_m,\mX_\sharp\>|^2 + \xi_m,~~m=1,\ldots,M.
\end{equation}
We refer to this inverse problem as \emph{low-rank phase retrieval} (LRPR).  LRPR is a combination of two problems that have received a lot of attention over the past decade.  The phase retrieval problem, where the goal is to recover a vector $\vx\in\R^d$ from $M$ quadratic measurements of the form $|\<\vx,\vphi_m\>|^2$, is known to be solvable when the $\vphi_m$ are generic and $M\gtrsim d$ (e.g., see \cite{jaganathan2016phase} and references therein).  There are tractable algorithms for solving the equations that use convex relaxations based on semi-definite programming \cite{candes13ph,waldspurger15ph,candes15ph} and polytope constraints \cite{bahmani2017phase,goldstein17co}. There also exist fast iterative algorithms for nonconvex programming (e.g., \cite{netrapalli13ph,candes2015phase,chen2015solving,wang2017solving,tan2018phase,sun2018geometric}).  The problem of recovering a $d_1\times d_2$ matrix of rank $r$ from $M$ linear measurements of the form $\<\mPhi_m,\mX\>$ has also been thoroughly analyzed in the literature for generic $\mPhi_m$ \cite{recht10gu,candes14so}, $\mPhi_m$ that return samples of the matrix \cite{candes09ex,candes10po,keshavan10ma,recht11si}, and $\mPhi_m$ with structured randomness \cite{gross11re,ahmed15co}; a survey of these results can be found in \cite{davenport2016overview}.

Our contribution in this paper is to show that for certain choices of the $\mPhi_m$, we can recover $\mX_\sharp$ from phaseless measurements \eqref{eq:gen_mdl} from far fewer than $d_1d_2$ measurements by taking advantage of the low-rank structure of $\mX_\sharp$. 
Our recovery algorithm uses the recently developed idea of \emph{anchored regression} \cite{bahmani2018solving,bahmani2017phase}.  The common approaches to estimate $\mX_\sharp$ from the nonlinear observations \eqref{eq:gen_mdl} lead to nonconvex programs.  The anchored regression, however, enables estimation by convex programming as follows.  The first step is effectively relaxing the nonlinear equations \eqref{eq:gen_mdl} to convex feasibility constraints. The second step, is to use an {\em anchor matrix} $\mX_0$, which serves as an initial guess for the solution, to formulate a simple convex program that finds a matrix that is feasible in the relaxed constraints and is best aligned with $\mX_0$. When the measurements are noiseless ($\xi_m=0$), we solve
\begin{equation}
\label{eq:convex_prog_noiseless}
\begin{array}{ll}
\displaystyle \minimize_{\mX} & - \mathrm{Re}\,\langle \mX_0, \mX \rangle + \lambda \norm{\mX}_*  \\
\mathrm{subject~to} &
\displaystyle |\langle \mPhi_m, \mX \rangle|^2 \leq y_m,\quad m=1,\ldots,M.
\end{array}
\end{equation}
This is a convex program over the space of $d_1\times d_2$ matrices.  Geometrically, each constraint $|\langle \mPhi_m, \mX \rangle|^2 \leq y_m$ is a convex set that has the target $\mX_\sharp$ on its surface.  The program finds an extreme point of the intersection of these convex sets by minimizing the linear functional $- \mathrm{Re}\,\langle \mX_0, \mX \rangle$ regularized by the nuclear norm $\|\mX\|_*$ to account for the low-rank structure of the solution.  The success of this program in recovering the target (to within a global phase ambiguity) depends on the behavior of the constraints around $\mX_\sharp$ and having an anchor $\mX_0$ sufficiently correlated with $\mX_\sharp$.

When there is noise, we relax the constraints in \eqref{eq:convex_prog_noiseless} and solve
\begin{equation}
\label{eq:convex_prog}
\begin{array}{ll}
\displaystyle \minimize_{\mX} & - \mathrm{Re}\,\langle \mX_0, \mX \rangle + \lambda \norm{\mX}_* \\
\mathrm{subject~to} &
\displaystyle \frac{1}{M} \sum_{m=1}^M (|\langle \mPhi_m, \mX \rangle|^2 - y_m)_+
\leq\eta\,,
\end{array}
\end{equation}
where $(\cdot)_+$ denotes the positive part function. This yields a stable solution in the sense that if the conditions for noise-free recovery are met, and we choose $\eta$ larger than the positive part of the perturbations, that is,
\[
	\eta = \frac{1}{M} \sum_{m=1}^M (-\xi_m)_+ + \epsilon, \quad \text{for some} ~ \epsilon \geq 0,
\]
then the solution $\hat\mX$ to \eqref{eq:convex_prog} obeys $\|\hat\mX- e^{\mathfrak{j}\theta} \mX_\sharp\|_\mathrm{F} \lesssim \eta$ for some $\theta \in [0, 2\pi)$. Here $\epsilon$ denotes an error in estimating the average of the positive part of perturbations by $\eta$.

We analyze two scenarios in detail.  In the first scenario, the target matrix $\mX_\sharp$ is of rank $r$, and the measurement matrices $\mPhi_m$ have independent real-valued Gaussian entries,
\begin{equation}
	\label{eq:Phigauss}
	\vect(\mPhi_m)\sim\mathcal{N}(\vzero,\mId),\quad m=1,\ldots,M.
\end{equation}
Theorem~\ref{thm:main_iidG} below shows that if we start with an anchor matrix that is sufficiently close to $\mX_\sharp$, exact recovery occurs when $M\gtrsim r(d_1+d_2)\log(d_1+d_2)$.  Lemma~\ref{lemma:init_iidG} shows that the anchor matrix can be computed from the data by a variation of the spectral initialization when the number of measurements $M$ satisfies $M \gtrsim r^3 \kappa^4(d_1+d_2)\log(d_1+d_2)$, where $\kappa$ denotes the condition number of $\mX_\sharp$.  We also show that the recovery procedure is stable in presence of noise.

In our second scenario, the target matrix has rank one, $\mX_\sharp=\sigma\vu\vv^*$ with $\vu\in\mathbb{C}^{d_1}, \vv\in\mathbb{C}^{d_2}, \|\vu\|_2=\|\vv\|_2=1$, as do the measurement matrices, $\mPhi_m=\va_m\vb_m^*$.  As we discuss below, this scenario is a model for the blind deconvolution of two signals from magnitude measurements in the frequency domain.  Our analysis in Theorem~\ref{thm:main_rank1} below takes the $\va_m$ and $\vb_m$ to be complex-valued independent Gaussian random vectors,
\begin{equation}
	\label{eq:ab_random}
	\va_m\sim\mathcal{CN}(\vzero,\mId),\quad \vb_m\sim\mathcal{CN}(\vzero,\mId),\quad m=1,\ldots,M.
\end{equation}
Under this model, we show that anchored regression produces a stable estimate of $(\vu,\vv)$ when $M$ is within a logarithmic factor of $d_1+d_2$.  Lemma~\ref{lemma:init} gives a computationally efficient technique for constructing the anchor in a commensurate number of measurements.

\section{Application: Blind deconvolution from Fourier magnitude observations}
\label{sec:phaselessdeconv}

Low-rank phase retrieval arises in a variation of the blind deconvolution problems. We consider estimating two unknown signals from the Fourier magnitudes of the convolution.  While blind deconvolution is itself an ill-posed, nonlinear problem, the absence of phase information in the Fourier measurements makes it even more challenging.
The type of phaseless blind deconvolution problem we describe below arises in various applications in communications and imaging.  In optical communications, high spectral efficiency and robustness against adversarial channel conditions for multiple-input multiple-output (MIMO) channels can be achieved using orthogonal frequency division multiplexing (OFDM).  Calibrating these communication channels involves solving a blind deconvolution problem.  This problem has to be solved from phaseless observations, as practical direct detection receivers work with intensity-only measurements \cite{arik2016direct} to provide robustness against synchronization errors, which has been one of the key issues in the OFDM systems \cite{schmidl1997robust,bouziane2015blind}.

A similar calibration problem arises in Fourier ptychography \cite{eckert2016algorithmic}.  In this application, an image is computed from phaseless Fourier domain measurements.  If there is uncertainty in the point spread function of the optical system, recovering the image becomes a phaseless blind deconvolution problem.

Blind deconvolution that identifies unknown signals $\vx,\vh \in \mathbb{C}^M$ (up to reciprocal scaling) from their circular convolution is in general ill-posed, but can be solved with a priori information on  $\vx$ and $\vh$. The circular convolution of $\vx$ and $\vh$ can be equivalently expressed in the Fourier domain as the element-wise product, namely
\begin{equation}
    \label{eq:bd_meas}
\mF (\vx \circledast \vh) = \sqrt{M} \mF \vx \odot \mF \vh,
\end{equation}
where $\mF \in \mathbb{C}^{M \times M}$ is the unitary discrete Fourier matrix of size $M$.

We will  impose subspace priors on $\vx$ and $\vh$, modeling $\vx\in\mathbb{C}^{M}$ as being in the low-dimensional columnspace of $\mD \in \mathbb{C}^{M \times d_1}$, and $\vh$ as being in the columnspace of $\mE \in \mathbb{C}^{M \times d_2}$.  Then $\vx$ and $\vh$ are represented as
\begin{equation}
    \label{eq:bd_spmdl}
	\vx = \mD \vu \quad \text{and} \quad \vh = \mE \overline{\vv},
\end{equation}
for some $\vu \in \mathbb{C}^{d_1}$ and $\vv \in \mathbb{C}^{d_2}$.  Here, $\overline{\vv}$ denotes the entry-wise complex conjugate of $\vv$.  Let $\va_m$ denote the $m$th column of $\mD^* \mF^*$ and $\vb_m$ denote the $m$th column of $\mE^\transpose \mF^\transpose$ for $m=1,\dots,M$.  Then the Fourier measurement of the convolution at frequency $m$ (after an appropriate normalization) is given as $\va_m^* \vu \vv^* \vb_m$. Under this subspace model, it suffices to recover $\vu$ and $\vv$.

In particular applications, the subspace model for $\vh$ might be introduced as a linear approximation of parametric models via principal component analysis. This technique is used for source localization and channel estimation in underwater acoustics \cite{mantzel14ro,tian2017multichannel
}. Some analysis in the context of dimensionality reduction of manifolds is provided in \cite{mantzel2015compressed,dirksen2015dimensionality}.

In the scenario where only noisy Fourier magnitudes of the convolution is observed, the corresponding quadratic measurements are given in the form of
\[
y_m = |\va_m^* \vu \vv^* \vb_m|^2 + \xi_m, \quad m=1,\dots,M,
\]
where $\xi_1,\dots,\xi_M$ denote additive noise.  Through the lifting reformulation \cite{ahmed2014blind} that substitutes $\vu \vv^*$ by a rank-$1$ matrix $\mX_\sharp$, the recovery reduces to a LRPR that estimates the unknown rank-$1$ matrix $\mX_\sharp$ from its noisy quadratic measurements:
\begin{equation}
\label{eq:meas}
y_m = |\langle \va_m \vb_m^*, \mX_\sharp \rangle|^2 + \xi_m, \quad m=1,\dots,M.
\end{equation}
This is a particular instance of LRPR and generates the quadratic measurements with rank-$1$ matrices $\va_1 \vb_1^*, \dots, \va_M \vb_M^*$.

In other words, the recovery combines blind deconvolution and phase retrieval; hence, it suffers from the ambiguities in both problems.  Similar to the phase retrieval, the absence of the phases in the measurements makes the reconstruction a nonconvex problem, even after it has been lifted.  By themselves, both phase retrieval and blind deconvolution amount to solving a system of quadratic equations. However, the phaseless blind deconvolution problem \eqref{eq:meas} is a systems of fourth-order equations.  Below, we will show that this system can indeed be tractably solved under certain randomness assumptions on the considered subspaces.

\section{Related Work}

Recovery of a structured signal from nonlinear measurements has received a significant amount of attention in the last decade, particularly in terms of theoretical analysis of various optimization formulations. A prominent example is the phase retrieval problem, which recovers an unknown signal from quadratic measurements. Unique identification of the solution and performance guarantees of optimization algorithms in the case where the unknown signal is sparse has been recently studied in \cite{li2013sparse,lecue2015,chen2015exact,eldar2015sparse,cai2016optimal,bahmani2018solving,jaganathan2017sparse}. 

Another example, discussed in the previous section, is the blind deconvolution problem, which amounts to solving a system of bilinear equations.
Although many approaches for blind deconvolution and its variations have been proposed in the communications, signal processing, and computational imaging literature, there has been significant progress in recent years in identifying provable performance guarantees.
These results offer theoretical guarantees on the number of measurements $M$ in \eqref{eq:bd_meas} as a function of the subspace dimensions $d_1,d_2$ (number of columns of $\mD,\mE$ in \eqref{eq:bd_spmdl}) needed to recover $\vu,\vv$.  Results that exhibit near-optimal scaling of $M$ versus $d_1,d_2$ are known both for convex relaxations of the problem, and for iterative algorithms that minimize a nonconvex loss \cite{ahmed2014blind,li2018rapid,huang2018blind}.
These results have also been extended  to sparsity (in place of subspace) models where the recovery is performed through alternating minimization \cite{lee2017blind};  however, the near optimal result in this work makes some technical, and perhaps too restrictive, assumptions on the success of projection steps.

The blind deconvolution problem can be made easier if we have the freedom to obtain diversified observations. Specifically, the identification of unknown channel impulse responses excited by an unknown source has been studied extensively in the communications literature since the 1990s (e.g., \cite{xu1995least,moulines1995subspace}). These classical results assumed that the channel responses had finite length and provided algebraic performance guarantees. In recent years, its generalization to the blind gain and phase calibration problem has been analyzed and robust optimization algorithms were proposed along with performance guarantees \cite{ling2017self,wang2016blind,lee2016fast,lee2018spectral,li2018blind,cosse2017note,ahmed2018leveraging}. There also exists further generalization to the off-the-grid sparsity models \cite{chi2016guaranteed,yang2016super}.

The nonlinear recovery problem considered in this paper is motivated to study a version of blind deconvolution where the convolution measurements are observed through certain nonlinearities.  Bendory et al. \cite{bendory2018blind} studied a similar problem arising in blind ptychgraphy and identified a set of conditions under which a signal can be identified uniquely from the magnitudes of a short-time Fourier transform taken with an unknown window. In this paper we are more interested in the recovery by a practical convex program from Fourier magnitudes.  The lifting reformulation renders the reconstruction problem into phase retrieval of a low-rank matrix.

The problem of recovering a low-rank matrix from phaseless linear measurements can also be interpreted as a generalization of classical subspace learning (i.e\ principal components analysis).  This connection was made explicit in \cite{chen2015exact}, where the problem of estimating a covariance matrix from compressed, streaming data was considered.  In a subsequent work, \cite{vaswani2016low} considered the quadratic subspace learning problem in a more general setting.  A regularized gradient descent method was proposed to solve the LRPR problem, and they provided an analysis for the accuracy of the initialization step under certain randomness assumptions on the measurement matrices.

Unlike the aforementioned works \cite{chen2015exact,vaswani2016low}, we take a different approach to solving the LRPR problem that uses the recently introduced  \emph{anchored regression} \cite{bahmani2017phase,goldstein17co} technique for relaxing nonlinear measurements.
Unlike lifting techniques, this method recasts phase retrieval problem as a convex program without increasing the number of optimization variables.  Unlike techniques based on nonconvex optimization, its analysis relies only on geometry rather than the trajectory of a certain sequence of iterates, which significantly simplifies the derivations.  The anchored regression formulation also makes it straightforward to incorporate structural priors on the data through the introduction of convex regularizers \cite{bahmani2018solving}.  Importantly we present performance guarantees for stable recovery of low-rank matrices from its random quadratic measurements, which implies exact reconstruction in the noiseless case. Previously, it was only shown that the initialization by a truncated spectral method provides an accurate approximation  \cite{vaswani2016low}. After an early version of this paper \cite{7thICCHA}, another approach to the same problem was independently studied in \cite{ahmed2018blind}.  Unlike their work, our approach is not restricted to the case of rank-$1$ measurement matrices and more importantly like the anchored regression provides flexibility that allows the nonlinearity in the measurement model beyond the quadratic function.

While a general theory for solving equations with convex nonlinearities has been developed, of which \eqref{eq:gen_mdl} is an example, it still remains to compute the key estimates that depend on the structure of the problem (the low-rankness in our case).  Furthermore, it is crucial to design an appropriate initialization scheme that provides a valid anchor matrix.  We propose a unified approach to the initialization that takes advantage of the separability of the unknown matrix.

It would be of independent interest to see various estimates on functions of random matrices by the noncommutative Rosenthal inequality \cite{junge2013noncommutative}. All of the matrix Bernstein inequalities \cite{tropp2012user,koltchinskii2011nuclear} and noncommutative Rosenthal inequality \cite{junge2013noncommutative} provide tail estimates of a sum of independent random matrices. In applying the matrix Bernstein inequalities, one has to verify that all summands have bounded spectral norms (deterministically or almost surely) or compute their Orlicz norms. On the contrary, the noncommutative Rosenthal inequality \cite{junge2013noncommutative} first computes moment bounds and then provides a tail estimate by the Markov inequality. Particularly when random matrices are given by a set of Gaussian random variables, the spectral norm is not bounded almost surely and computing the Orlicz norm of the spectral norm is not trivial. Therefore, it is desirable to derive relevant tail estimates by using the noncommutative Rosenthal inequality. Additionally, the expectations of high-order tensor products of Gaussian random vectors in the appendix might be useful in the study of other applications sharing similar tensor structures.

%---------------------------
\section{Main Results}

We have four main results.  The first two, presented in Section~\ref{sec:samplecomplexity}, give {\em sample complexity} bounds that relate the accuracy of the estimate returned by \eqref{eq:convex_prog} to the number of measurements $M$ observed as in \eqref{eq:gen_mdl}.  In both cases where the $\mPhi_m$s are rank-$1$ as in \eqref{eq:ab_random} and when they have i.i.d. entries as in \eqref{eq:Phigauss}, we achieve a sample complexity that scales nearly optimally with the size of the target matrix $\mX_\sharp$ and its rank.  These results assume that we have an anchor matrix $\mX_0$ that is sufficiently correlated with $\mX_\sharp$.

Our next two main results, presented in Section~\ref{sec:spectralinit}, show how such an anchor matrix can be created from the measurements using a spectral initialization.  For the random rank-$1$ measurements, we are able to construct a sufficiently accurate anchor from a number of observations $M$ that is proportional to the degrees of freedom in the model of $\mX_\sharp$ up to a logarithmic factor.  In the case of Gaussian measurements, we have only a partial result in general, and show that a very rough anchor can be bootstrapped into a more accurate one.  In the case where $\mX_\sharp$ is positive semi-definite or is rank-$1$, however, the results are near-optimal.

%------------------------------------------------------------
\subsection{Sample complexity}
\label{sec:samplecomplexity}

We begin by presenting theorems that give guarantees on the accuracy of the solution to the convex program \eqref{eq:convex_prog} in relation to the number of measurements $M$.  In both of the theorems below, we will assume that we have an anchor matrix $\mX_0$ that is roughly aligned with the target $\mX_\sharp$; we defer the construction of this anchor to Section~\ref{sec:spectralinit}.

We start with the case where $\mX_\sharp$ is rank-$1$, and the measurements are formed by taking the outer product of two random vectors, $\mPhi_m = \va_m \vb_m^*$.  As discussed in Section~\ref{sec:phaselessdeconv} above, this scenario is motivated by problems that involve blind deconvolution from quadratic measurements.  Since these applications typically involve the Fourier transform, we formulate our results using complex-valued vectors and matrices.

\begin{theorem}
\label{thm:main_rank1}
Let $\mX_\sharp=\sigma_\sharp\vu_\sharp\vv_\sharp^*$ be a complex rank-$1$ matrix observed as in \eqref{eq:meas} with $\mPhi_m = \va_m \vb_m^*$ for $m=1\dots,M$, where $\va_1,\dots,\va_M$ and $\vb_1,\dots,\vb_M$ are independent complex Gaussian random vectors as in \eqref{eq:ab_random}.  Suppose that $\mX_0 = \vu_0 \vv_0^*$ with $\norm{\vu_0}_2 = \norm{\vv_0}_2 = 1$ satisfies
\begin{equation}
	\label{eq:cond_anchor_rank1}
	\inf_{\theta \in [0,2\pi)} \norm{\vu_0 \vv_0^* - e^{\mathfrak{i} \theta} \vu_\sharp \vv_\sharp^*}_\mathrm{F} \leq \delta
\end{equation}
for $\delta \leq 0.2$.  Then one can set the regularization parameter in \eqref{eq:convex_prog} such that there exist numerical constants $C_1, C_2, C_3$ and a constant $C_\delta$ that depends only on $\delta$, for which the following holds.\footnote{As shown in the proof of Theorem~\ref{thm:main_rank1}, given $\delta$, one can choose $\lambda$ explicitly as $0.9-\delta$. For specific methods of constructing the anchor matrix, an appropriate value of $\delta$ can be determined.} If
\begin{equation}
	\label{eq:sampl_comp_rank1}
	\frac{M}{\log^2 M} \geq C_\delta (d_1 + d_2),
\end{equation}
then the solution $\hat{\mX}$ to \eqref{eq:convex_prog} satisfies
\begin{equation}
	\label{eq:thm:esterr_rank1}
	\inf_{\theta \in [0,2\pi)} \norm{\hat{\mX} - e^{\mathfrak{i} \theta} \mX_\sharp}_\mathrm{F}
	\leq \frac{C_1}{\norm{\mX_\sharp}_\mathrm{F}} \Big( \frac{1}{M} \sum_{m=1}^M |\xi_m| + \epsilon \Big)
\end{equation}
with probability at least $1-e^{-C_3 M}$.  Furthermore, the left and right singular vectors $\hat{\vu}$ and $\hat{\vv}$ of $\hat{\mX}$ satisfy
\begin{equation}
	\label{eq:thm:esterr_factors_rank1}
	\sin\angle(\hat{\vu},\vu_\sharp) \vee \sin\angle(\hat{\vv},\vv_\sharp)
	\leq \frac{C_2}{\norm{\mX_\sharp}_\mathrm{F}^2} \Big( \frac{1}{M} \sum_{m=1}^M |\xi_m| + \epsilon \Big).
\end{equation}
\end{theorem}

The sufficient number of measurements for stable recovery of $\mX_\sharp$ (and hence its factors $\vu_\sharp$ and $\vv_\sharp$) required by \eqref{eq:sampl_comp_rank1}, scales nearly optimally. That is, the sufficient number of samples is proportional to the degrees of freedom of the unknown rank-$1$ matrix, i.e., $d_1+d_2$. In Section~\ref{sec:spectralinit} below, we will see that we can also find $\vu_0,\vv_0$ that obey \eqref{eq:cond_anchor_rank1} from a comparable number of measurements.  Combining these results shows that we can recover a $d_1\times d_2$ rank-$1$ matrix from phaseless rank-$1$ measurements when $M$ equals to $d_1+d_2$ up to a logarithmic factor. 

Our second sample complexity result states a performance bound for \eqref{eq:convex_prog} when the measurements are unstructured Gaussian random matrices and the target is a $d_1\times d_2$ matrix of rank $r$.  This type of measurement model has served as a standard benchmark in the structured recovery literature, and indeed we do obtain a much tighter bound in this case if the target is well-conditioned.  To ease the derivation, we state the result for real-valued matrices, but it is straightforward to extend it to the complex-valued case at the cost of making the calculations slightly more involved.

\begin{theorem}
\label{thm:main_iidG}
Let $\mX_\sharp \in \mathbb{R}^{d_1\times d_2}$ be of rank $r$, $\mX_\sharp  = \mU_\sharp \mSigma_\sharp \mV_\sharp^\transpose$ denote the compact singular value decomposition of $\mX_\sharp$, and $\mPhi_1,\dots,\mPhi_M \in \mathbb{R}^{d_1 \times d_2}$ be Gaussian random matrices as in \eqref{eq:Phigauss}.  Suppose that we have an anchor matrix $\mX_0 = \mU_0 \mV_0^\transpose$, where $\mU_0^\transpose \mU_0 = \mV_0^\transpose \mV_0 = \mId_r$, that satisfies
\begin{equation}
	\label{eq:cond_anchor_rankr}
	\min\left(\norm{\mU_0 \mV_0^\transpose - \mU_\sharp \mV_\sharp^\transpose}_\mathrm{F},~
	\norm{\mU_0 \mV_0^\transpose + \mU_\sharp \mV_\sharp^\transpose}_\mathrm{F}\right)
	\leq \delta \norm{\mU_\sharp \mV_\sharp^\transpose}_\mathrm{F}
\end{equation}
for $\delta$ that obeys
\begin{equation}
\label{eq:cond_flat_rankr}
	\frac{\delta}{1-\lambda} \leq 0.45 \, (2.8 - \kappa),
\end{equation}
where $\kappa$ and $\lambda$ denote the condition number of $\mX_\sharp$ and the regularization parameter in \eqref{eq:convex_prog} respectively. Then there exist universal constants $C_1,C_2$ and a constant $C_{\delta}$ that only depends on $\delta$ for which the following holds. If
\begin{equation}
	\label{eq:sampl_comp_iidG}
	M \geq C_{\delta} r (d_1+d_2) \log(d_1+d_2),
\end{equation}
then the solution $\hat{\mX}$ to \eqref{eq:convex_prog} satisfies
\begin{equation*}
%	\label{eq:thm:esterr_iidG}
	\min\left(\norm{\hat{\mX} - \mX_\sharp}_\mathrm{F},~ \norm{\hat{\mX} + \mX_\sharp}_\mathrm{F}\right)	
	~\leq~
	\frac{C_1}{\norm{\mX_\sharp}_\mathrm{F}} \Big( \frac{1}{M} \sum_{m=1}^M |\xi_m| + \epsilon \Big),
\end{equation*}
with probability at least $1 - e^{-C_2 M}$.
\end{theorem}

Although to the authors' knowledge this is the first result of its kind in the literature, and the bound \eqref{eq:sampl_comp_iidG} scales in the rank $r$ and dimensions $d_1,d_2$ as well as one could hope, we point out a few ways this result could be improved.  First, the condition \eqref{eq:cond_flat_rankr} is very restrictive in the sense that it applies only to matrices with a small condition number.  Second, constructing $\mU_0 \mV_0^\transpose$ that obeys \eqref{eq:cond_anchor_rankr} is non-trivial; as we will see in Section~\ref{sec:spectralinit} below, we will only really be able to do this with confidence when $\mX_\sharp$ is positive semi-definite or is rank-$1$.

%------------------------------------------------------------
\subsection{Spectral initialization with partial trace}
\label{sec:spectralinit}

Our main results, presented as Theorems~\ref{thm:main_rank1}  and \ref{thm:main_iidG} above, give bounds on the number of equations $M$ that are needed to guarantee that the solution to \eqref{eq:convex_prog} has a certain accuracy.  This accuracy depends on the anchor matrix $\mX_0$ being sufficiently close to the unknown matrix $\mX_\sharp$.  In both cases, we use $\mX_0 = \mU_0\mV_0^*$ as an anchor, where $\mU_0 \mSigma_0 \mV_0^*$ is the compact singular value decomposition (SVD) of an approximation of $\mX_\sharp$;  we are after $\mU_0,\mV_0$, each with orthonormal columns, such that for some $\delta>0$ and a unit modulus $z$ we have
\begin{equation}
\label{eq:cond_anchor}
	\|\mU_0\mV_0^* - z\mU_\sharp\mV_\sharp^*\|_\mathrm{F} \leq\delta \|\mU_\sharp\mV_\sharp^*\|_\mathrm{F}.
\end{equation}
In each of the main theorems below, the bounds on $M$ scale like $\delta^{-2}$, and we will achieve the tightest results when we can take $\delta$ as a constant independent of the matrix dimensions and rank.
In this section, we describe a {\em data-driven} technique for constructing such an anchor matrix.

To understand the challenges in creating the anchor, let us first recall the now well-known spectral initialization for standard phase retrieval for vectors ($d_2=1$ in the formulation above).  In this case, we use the observations $y_m$ to form the $d_1\times d_1$ matrix
\begin{equation}
	\label{eq:vpranchor}
	\hat{\mR} = \frac{1}{M}\sum_{m=1}^M y_m\mPhi_m\mPhi_m^*,
\end{equation}
and then use the leading eigenvector of $\hat{\mR}$ as the anchor matrix $\mX_0$.  The idea is that when $y_m=|\<\mX_0,\mPhi_m\>|^2$ and the $\mPhi_m$ are random and drawn independent of one another, the expectation of $\E[y_m\mPhi_m\mPhi_m^*]$ has a leading eigenvector that is exactly $\mX_0$, and for $M$ large enough, the sum in \eqref{eq:vpranchor} provides a good approximation to this expectation.  In \cite{candes2015phase}, it is shown that \eqref{eq:cond_anchor_rank1} holds for constant $\delta$ when $M\gtrsim d_1\log d_1$.

We might consider using the same initialization when $\mX_\sharp$ and the $\mPhi_m$ are matrices.  Using a vectorized version of the above, we can form
\[
	\hat\mR = \frac{1}{M}\sum_{m=1}^M y_m \vect(\mPhi_m)\vect(\mPhi_m)^*,
\]
compute the leading eigenvector, then reshape into a $d_1\times d_2$ matrix.  We are now guaranteed a good anchor when $M\gtrsim d_1d_2\log(d_1d_2)$.  The problem, though, is that this bound is independent of the rank of $\mX_\sharp$; we are interested in recovery results that scale as closely as possible to the intrinsic number of degrees of freedom $r(d_1+d_2)$ in our matrix model.  Simply finding the largest eigenvector of $\hat\mR$ and then re-arranging into a $d_1\times d_2$ matrix will not, by itself, result in a matrix that is rank $r$, and there is no known algorithm with provable performance guarantees for finding a rank-constrained matrix that is maximally aligned with the columnspace of $\hat\mR$ (this is a variation on the ``Sparse PCA'' problem).

Our approach for estimating the anchor matrix will be to estimate the row and columnspaces of $\mX_\sharp$ individually.  We will find a $d_1\times r$ matrix $\mU_0$ whose columns are orthonormal and approximately span the columnspace, a $d_2\times r$ matrix $\mV_0$ whose columns are orthonormal and approximately span the rowspace, and then take
\[
	\mX_0 = \mU_0\mV_0^*.
\]
For the columnspace estimate $\mU_0$, we choose $d_2\times q$ {\em compression matrices} $\mPsi_m$ and form
\begin{equation}
	\label{eq:def_Upsilon}
	\mUpsilon = \frac{1}{M} \sum_{m=1}^M y_m \mPhi_m \mPsi_m \mPsi_m^* \mPhi_m^*,
\end{equation}
then take the $r$ leading eigenvectors of $\mUpsilon$ as $\mU_0$.  Similarly for the rowspace, we choose $d_1\times q$ $\mPsi'_m$, form
\begin{equation}
	\label{eq:def_Upsilon_prime}
	\mUpsilon' =
	\frac{1}{M} \sum_{m=1}^M y_m \mPhi_m^* \mPsi_m' \mPsi_m'^* \mPhi_m,
\end{equation}
and take the $r$ leading eigenvectors as $\mV_0$.

With the measurement matrix $\mPhi_m$ random, we want to choose the compression matrices $\mPsi_m$ in \eqref{eq:def_Upsilon} to meet two criteria:
\begin{enumerate}
	\item The expectation $\E[\mUpsilon]$ has leading eigenvectors that span the same $r$-dimensional space as the eigenvectors of $\mX_\sharp\mX_\sharp^*$.
	\item The spectral gap between the $r$th and $(r+1)$th eigenvalues of  $\E[\mUpsilon]$ is large enough so that  it upper bounds the perturbation error $\|\mUpsilon-\E\mUpsilon\|$ for relatively small $M$.  This allows us to use the classical Davis-Kahan theorem to show that the leading eigenvectors of $\mUpsilon$ are approximately aligned with the leading eigenvectors of $\E[\mUpsilon]$.
\end{enumerate}
Similar statements hold for the $\mPsi'_m$ in \eqref{eq:def_Upsilon_prime}.

For our blind deconvolution from phaseless measurements application, where $\mX_\sharp=\sigma\vu\vv^*$ and $\mPhi_m = \va_m\vb_m^*$, there is a clear way to meet these criteria.  If we take
\[
	\mPsi_m = \frac{\vb_m}{\norm{\vb_m}_2^2}
	\quad \text{and} \quad
	\mPsi_m' = \frac{\va_m}{\norm{\va_m}_2^2},
	\quad m=1,\dots,M,
\]
then
\begin{align}
	\label{eq:mUpsilon_rank1}
	\mUpsilon &= \frac{1}{M}\sum_{m=1}^My_m\va_m\va_m^*
	= \frac{\sigma^2}{M}\sum_{m=1}^M|\va_m^*\vu|^2|\vv^*\vb_m|^2\va_m\va_m^*+ \xi_m\va_m\va_m^*, \\
	\label{eq:mUpsilon'}
	\mUpsilon' &= \frac{1}{M}\sum_{m=1}^My_m\vb_m\vb_m^*
	= \frac{\sigma^2}{M}\sum_{m=1}^M|\va_m^*\vu|^2|\vv^*\vb_m|^2\vb_m\vb_m^*+ \xi_m\vb_m\vb_m^*.
\end{align}
For independent $\va_m,\vb_m$ that follow \eqref{eq:ab_random}, a simple calculation yields
\[
	\E\mUpsilon = \sigma^2\vu\vu^* + \left(\sigma^2 + \frac{1}{M}\sum_{m=1}^M\xi_m\right)\mId,
	\quad
	\E\mUpsilon' = \sigma^2\vv\vv^* + \left(\sigma^2 + \frac{1}{M}\sum_{m=1}^M\xi_m\right)\mId.
\]
The leading eigenvector for $\mUpsilon$ is the left singular vector $\vu$ for $\mX_\sharp$, the leading eigenvector of $\mUpsilon'$ is the right singular vector $\vv$, and the spectral gap in both cases is $\sigma^2$.  That $\mUpsilon-\E\mUpsilon$ and $\mUpsilon'-\E\mUpsilon'$ are small enough so that their leading eigenvectors are close to $\vu$ and $\vv$ when $M$ is withing a logarithmic factor of $(d_1+d_2)$ is essentially the content of the following lemma.

\begin{lemma}
	\label{lemma:init}
	Let $\mPhi_m=\va_m\vb_m^*$  be as in \eqref{eq:ab_random}.  Let $\vu_0 \in \mathbb{C}^{d_1}$ (resp. $\vv_0 \in \mathbb{C}^{d_2}$) be the leading eigenvector of $\mUpsilon$ in \eqref{eq:mUpsilon_rank1} (resp. $\mUpsilon'$ in \eqref{eq:mUpsilon'}) with measurements $y_m$ constructed as in \eqref{eq:gen_mdl}.  Let $\vu_\sharp$ and $\vv_\sharp$ denote the left and right singular vectors of the rank-$1$ matrix $\mX_\sharp$.  Let $\delta \in (0,1)$ and $\alpha \in \mathbb{N}$.  There exist numerical constants $C_1,C_2$ that only depend on $\alpha$, for which the following holds.  If
\begin{equation}
	\label{eq:sampl_comp_init}
	\frac{M}{\log^3 M} \geq C_1 \delta^{-2} (d_1 + d_2)
\end{equation}
and
\begin{equation}
	\label{eq:snr_cond_init}
	\max_{1\leq m\leq M} |\xi_m| \leq C_2 \norm{\mX_\sharp}^2 \log M,
\end{equation}
then \eqref{eq:cond_anchor_rank1} holds with probability at least $1 - M^{-\alpha}$.
\end{lemma}

\begin{rem}
\label{rem:snr}
The inequality \eqref{eq:snr_cond_init} requires that signal-to-noise-ratio is larger than the given threshold.  The proof of Lemma~\ref{lemma:init} presents a stronger result that holds by \eqref{eq:sampl_comp_init} and 
\begin{equation}
	\label{eq:snr_cond_init2}
	\frac{M}{\log M} \geq C_2 \alpha \left( \frac{\max_{1\leq m\leq M} |\xi_m|}{\norm{\mX_\sharp}^2} \vee \delta^{-1} \left(\frac{\max_{1\leq m\leq M} |\xi_m|}{\norm{\mX_\sharp}^2}\right)^2\right)  \delta^{-1} (d_1+d_2).
\end{equation}
Indeed, \eqref{eq:snr_cond_init} together with \eqref{eq:sampl_comp_init} implies \eqref{eq:snr_cond_init2}.  Even if \eqref{eq:snr_cond_init} is violated, \eqref{eq:cond_anchor_rank1} still holds with high probability whenever $M$ is large enough to satisfy \eqref{eq:snr_cond_init2} that naturally adapts to the signal-to-noise-ratio. To achieve the order of the logarithmic term in \eqref{eq:sampl_comp_init}, it is necessary to satisfy $M \lesssim e^{d_1+d_2}$.  Since this upper bound is rather trivial compared to \eqref{eq:sampl_comp_init}, we omit the condition in the statement of Lemma~\ref{lemma:init}.
\end{rem}

Lemma~\ref{lemma:init} along with Theorem~\ref{thm:main_rank1} give us a clean solution to the phaseless blind deconvolution problem.  For generic $\va_m,\vb_m$, the system
\[
	y_m = |\<\vu,\va_m\>\<\vb_m,\vv\>|^2 + \mathrm{noise},\quad m=1,\ldots,M,
\]
can be (stably) solved for $\vu,\vv$ when $M$ is within a logarithmic factor of $d_1+d_2$, the total number of unknowns.

For phaseless measurements of a $d_1\times d_2$ matrix of rank $r$, the story is unfortunately not as clean, even when the $\mPhi_m$ in \eqref{eq:gen_mdl} are i.i.d. Gaussian.  The following lemma gives us a partial result on our ability to create a data-driven anchor.  It shows that given an estimate of the rowspace, this estimate can be leveraged into an accurate estimate of the columnspace.

\begin{lemma}
	\label{lemma:init_iidG}
	Let $\mX_\sharp$ and $\mPhi_m$s be as in Theorem~\ref{thm:main_iidG}.  Let $\hat{\mV} \in \mathbb{R}^{d_2 \times r}$ satisfy $\hat{\mV}^\transpose \hat{\mV} = \mId_r$. Suppose that $\hat{\mV}$ is given a priori and provides an estimate of the rowspace of $\mX_\sharp$ so that
	\begin{equation}
		\label{eq:sin_estV}
		\norm{(\mId_{d_2} - \hat{\mV} \hat{\mV}^\transpose) \mV_\sharp \mV_\sharp^\transpose} \leq \delta_{\mathrm{in}}
	\end{equation}
	for some $\delta_{\mathrm{in}} < 1$.  Take $\mUpsilon$ as in \eqref{eq:def_Upsilon} with $\mPsi_m = \hat{\mV}$, and let the columns of $\mU_0$ be the eigenvectors of $\mUpsilon$ corresponding to the $r$-largest eigenvalues.  Fix $\delta_{\mathrm{out}} \in (0,1)$ and $\alpha \in \mathbb{N}$.  Then there exist numerical constants $C_1,C_2$ that only depend on $\alpha$, for which the following holds.  If
	\begin{equation}
		\label{eq:sampl_comp_init_iidG}
	\frac{M}{\log^3 M} \geq \frac{C_1 \alpha^3 \kappa^4 r^3 d_1}{\delta_{\mathrm{out}}^2 (1-\delta_{\mathrm{in}})^2}
	\end{equation}
	and
	\begin{equation}
		\label{eq:snr_cond_init_iidG}
	\frac{\max_{1\leq m\leq M} |\xi_m|}{\norm{\mX_\sharp}} \lesssim \sqrt{r} \log M \wedge \frac{\kappa^2 r^2 \log^2 M}{\delta_\mathrm{out}(1-\delta_{\mathrm{in}})},
	\end{equation}
	then
	\begin{equation}
		\label{eq:sin_estU}
		\norm{(\mId_{d_1} - \mU_0 \mU_0^\transpose) \mU_\sharp \mU_\sharp^\transpose} \leq \delta_{\mathrm{out}},
	\end{equation}
	holds with probability $1 - M^{-\alpha}$, where $\kappa$ denotes the condition number of $\mX_\sharp$.
\end{lemma}

\begin{rem}
If noise is weak enough to satisfy \eqref{eq:snr_cond_init_iidG}, then \eqref{eq:sampl_comp_init_iidG} implies 
\begin{equation}
		\label{eq:snr_cond_init2_iidG}
	\frac{M}{\log M} \geq C_2 \alpha \left( \frac{\kappa^2 \max_{1\leq m\leq M} |\xi_m|}{\delta_{\mathrm{out}} (1-\delta_{\mathrm{in}}) \norm{\mX_\sharp}^2} \vee r \left(\frac{\kappa^2 \max_{1\leq m\leq M} |\xi_m|}{\delta_{\mathrm{out}} (1-\delta_{\mathrm{in}}) \norm{\mX_\sharp}^2}\right)^2 \right)  r d_1.
\end{equation}
Similarly to Remark~\ref{rem:snr}, Lemma~\ref{lemma:init_iidG} can also be strengthened by substituting \eqref{eq:snr_cond_init_iidG} by \eqref{eq:snr_cond_init2_iidG}.  The signal-to-noise-ratio need not be larger than the threshold in \eqref{eq:snr_cond_init_iidG} whenever $M$ also satisfies \eqref{eq:snr_cond_init2_iidG}.  Indeed, this version of Lemma~\ref{lemma:init_iidG} is proved in Appendix~\ref{sec:proof:lemma:init_iidG}. 
\end{rem}

Lemma~\ref{lemma:init_iidG} shows that one obtains an estimate of the columnspace of accuracy $\delta_{\mathrm{out}}$ from a given estimate of the rowspace of accuracy $\delta_{\mathrm{in}}$.  Here the accuracy is measured by the sine of the principal angle between two subspaces.  The number of measurements $M$ in \eqref{eq:sampl_comp_init_iidG} that guarantees this result increases as one wishes for a more accurate estimate (smaller $\delta_{\mathrm{out}}$) or the input to the initialization method is less accurate (larger $\delta_{\mathrm{in}}$).

Furthermore, it is straightforward to exchange the roles of $\mU_\sharp$ and $\mV_\sharp$ above.  If we have an estimate $\hat\mU$ of $\mU_\sharp$, then we can form $\mUpsilon'$ as in \eqref{eq:def_Upsilon_prime} with $\mPhi_m = \hat\mU$, take its leading eigenvectors, and have (under analogous conditions as those in the theorem) an accurate estimate of $\mV_\sharp$.

The scaling of the number of measurements in \eqref{eq:sampl_comp_init_iidG} has suboptimal dependence on the rank, but its dependence on the side length of the matrix is linear.

Producing an estimate of $\mX_\sharp$ from matrices $\mU_0$ and $\mV_0$ whose ranges approximate its row and columnspaces is itself non-trivial.  It involves solving another phase retrieval problem, finding a diagonal $\mSigma$ so that
\[
	y_m\approx |\<\mU_0\mSigma\mV_0^\transpose,\mPhi_m\>|^2,\quad  m=1,\ldots,M.
\]
Although it might be possible to control the error propagation from the estimates $\mU_0,\mV_0$ to the solution of the problem above, this analysis appears to be extremely complicated.\footnote{An alternative approach is to estimate $\mSigma$ from $\mU_0$ and $\mV_0$ through extra independent random measurements. However this approach doubles the number of observations and may not be interesting in practice. Therefore, we pursue analysis in some special cases without extra observations.}  However,  there are two specific scenarios where we can upper-bound the error in estimating $\mU_\sharp \mV_\sharp^\transpose$ by the subspace estimation errors.
\begin{enumerate}

  \item rank-$1$ case: Let $\sigma_\sharp \vu_\sharp \vv_\sharp^\transpose$ be the SVD of $\mX_\sharp$.  Let $\phi := \angle(\vu_0,\vu_\sharp)$ and $\psi := \angle(\vv_0,\vv_\sharp)$.  Then
\begin{align*}
& \norm{\vu_0 \vv_0^\transpose - \vu_\sharp \vv_\sharp^\transpose}_\mathrm{F}^2
\wedge
\norm{\vu_0 \vv_0^\transpose + \vu_\sharp \vv_\sharp^\transpose}_\mathrm{F}^2
= 2 - 2 \cos\phi \cos\psi
\leq 2 - 2 \cos^2(\phi \vee \psi) \\
& = 2 \sin^2(\phi \vee \psi)
= \norm{(\mId_{d_1} - \vu_0 \vu_0^\transpose) \vu_\sharp}_2^2
\vee
\norm{(\mId_{d_2} - \vv_0 \vv_0^\transpose) \vv_\sharp}_2^2.
\end{align*}

  \item Positive semi-definite case: Let $\mU_\sharp \mLambda_\sharp \mU_\sharp^\transpose$ be the SVD of $\mX_\sharp$.  Then
\[
\norm{\mU_0 \mU_0^\transpose - \mU_\sharp \mU_\sharp^\transpose}_\mathrm{F}^2
= 2r - 2 \norm{\mU_0^\transpose \mU_\sharp}_\mathrm{F}^2
= 2 \norm{(\mId_{d_1} - \mU_0 \mU_0^\transpose) \mU_\sharp}_\mathrm{F}^2
\leq 2 r \norm{(\mId_{d_1} - \mU_0 \mU_0^\transpose) \mU_\sharp}^2.
\]

\end{enumerate}

For the above two cases, one can combine Theorem~\ref{thm:main_iidG} and Lemma~\ref{lemma:init_iidG} to get a complete analysis of the regularized anchored regression.  In the latter case, we still assume that an estimate of $\mU_\sharp$ is given a priori. Lemma~\ref{lemma:init_iidG} provides a refined estimate so that we can invoke Theorem~\ref{thm:main_iidG} with the resulting $\mU_0$.

%------------------------------------------------------------------------------------------------------------------------
\section{Proof of Main Results}

The convex program for phase retrieval of low-rank matrices in \eqref{eq:convex_prog} is variation to a special case of the anchored regression studied in \cite{bahmani2018solving} and the performance guarantees in this paper primarily follow from the main results in \cite{bahmani2018solving}.  
The theorems stated in the previous section are basically obtained by computing the key quantities that determine the sample complexity.

\subsection{Theoretical analysis of regularized anchored regression}
\label{subsec:general}
At the core of our analysis is an adaptation of the main result of \cite{bahmani2018solving}. 
%For the convenience of the readers, we provide a paraphrased version of the main result in \cite{bahmani2018solving} for the regularized anchored regression, where the regularizer is given by the nuclear norm and the measurements are quadratic.  
The main idea of \cite[Theorem~2.1]{bahmani2018solving} is to use the small-ball method to find a uniform lower bound for a certain empirical process that is determined by the independent random matrices $\mPhi_1,\dots,\mPhi_M$ and indexed by a \emph{deterministic} set $\mathcal{H} \subset \mathbb{C}^{d_1 \times d_2}$ containing $\mathbf{\varDelta} = \hat{\mX} - \mX_\sharp$. Then, this uniform lower bound implies an upper bound for the estimation error $\mathbf{\varDelta}$.

However, the original statement of \cite[Theorem 2.1]{bahmani2018solving} cannot be applied directly to the problem of interest in this paper because of two important differences. First, due to technical challenges in our specific problem, as elaborated in Section~\ref{sec:spectralinit}, we can only construct a weaker form of anchor compared to that considered originally in \cite{bahmani2018solving}. Second, we want to address the case of recovering complex and rank-$1$ matrices as considered in Theorem \ref{thm:main_rank1}. The results of \cite{bahmani2018solving}, however, only consider variables and operations in the real space. Therefore, we need to adapt the result of \cite{bahmani2018solving} with slight modifications so that it becomes compatible with our setting.

As discussed in Section \ref{sec:spectralinit}, instead of an anchor that approximates the ground truth $\mX_\sharp$, we require the anchor to approximate $\mU_\sharp\mV_\sharp^*$ up to a global phase. To be explicit, we only need to consider a complex phase ambiguity in the case of recovering a complex rank-$1$ target, where we have $\mU_\sharp =\vu_\sharp$ and $\mV_\sharp = \vv_\sharp$, and the anchor should basically approximate $\vu_\sharp\vv_\sharp^*$. In the case of recovering a real-valued low-rank matrix, the phase ambiguity simply reduces to a sign ambiguity.

With these consideration in mind, here and throughout, we assume that the global phase of the anchor $\mX_0$ is aligned with $\mU_\sharp\mV_\sharp^*$, namely
\begin{equation}
\label{eq:phasea0b0}
\mathrm{Re}\,\langle \mX_0, \mU_\sharp\mV_\sharp^* \rangle \geq 0,
\quad
\mathrm{Im}\,\langle \mX_0, \mU_\sharp\mV_\sharp^* \rangle = 0,
\end{equation}
which, if we operate entirely in the real domain, simply reduces to $\langle \mX_0, \mU_\sharp\mV_\sharp^\transpose \rangle \ge 0$.
The assumption \eqref{eq:phasea0b0} can be made without loss of generality because of the following \emph{equivariance} property. For any $\theta\in[0,2\pi)$, if we replace the anchor $\mX_0$ in \eqref{eq:convex_prog} by $e^{\mathfrak{i}\theta}\mX_0$, then the original solution $\hat{\mX}$ accordingly changes to $e^{\mathfrak{i}\theta}\hat{\mX}$. This property is due to fact the the nuclear norm as well as the constraints in \eqref{eq:convex_prog} are invariant under the mapping $\mX\mapsto e^{\mathfrak{i}\theta}\mX$. Since we define the accuracy as the distance to the orbit of $\mX_\sharp$, i.e., $\{e^{\mathfrak{i}\omega}\mX_\sharp\,:\,\omega\in[0,2\pi)\}$, the mentioned adjustment of the anchor will not affect the accuracy guarantees.
Indeed, under \eqref{eq:phasea0b0}, the assumption in \eqref{eq:cond_anchor} simplifies to 
\begin{equation}
\label{eq:proximityX0toXsharp2}
\norm{\mX_0 - \mU_\sharp \mV_\sharp^*}_\mathrm{F} \leq \delta \norm{\mU_\sharp \mV_\sharp^*}_\mathrm{F} = \delta \sqrt{r}.
\end{equation}

Since $\hat{\mX}$ is a minimizer to \eqref{eq:convex_prog} and $\mX_\sharp$ is within its feasible set, it naturally follows that $\mathbf{\varDelta}=\hat{\mX}-\mX_\sharp$ belongs to the set of all ascent directions of the objective function given by
\[
\mathcal{A} :=
\Big\{
\mH \in \mathbb{C}^{d_1 \times d_2} \,:\,
\inf_{\mG \in \lambda \partial \norm{\mX_\sharp}_*}
\mathrm{Re}\,\langle \mX_0 - \mG, \mH \rangle \geq 0
\Big\}.
\]
It is desirable to construct the anchor matrix $\mX_0$ from the available measurements and avoid \emph{sample splitting} schemes. However, for such constructions of the anchor matrix, the set $\mathcal{A}$ will also depend on the measurement matrices $\{\mPhi_m\}_{m=1}^M$ that complicates the analysis. To avoid these complications, similar to the approach of \cite{bahmani2018solving}, we relax $\mathcal{A}$ to some superset that is not dependent on the measurement matrices. Here we consider the superset $\mathcal{A}_\delta$ of $\mathcal{A}$, defined as
\begin{equation}
\label{eq:A_delta}
\mathcal{A}_\delta :=
\Big\{
\mH \in \mathbb{C}^{d_1 \times d_2} \,:\,
\inf_{\mG \in \lambda \partial \norm{\mX_\sharp}_*}
\sqrt{r} \delta \norm{\mH}_\mathrm{F} +
\mathrm{Re}\,\langle \mU_\sharp \mV_\sharp^* - \mG, \mH \rangle \geq 0
\Big\}\,,
\end{equation}
which is clearly independent of $\{\mPhi_m\}_{m=1}^M$.  
Inclusion of $\mathcal{A}$ in $\mathcal{A}_\delta$ follows from \eqref{eq:proximityX0toXsharp2}, the triangle inequality, and the Cauchy-Schwarz inequality.

To address a technical challenge that only arises when operating in the complex domain, for recovery of complex rank-$1$ matrices we need to make another modification compared to the original result of \cite{bahmani2018solving}. Specifically, similar to \cite{bahmani2017phase}, with $\mX_\sharp = \sigma_\sharp\vu_\sharp \vv_\sharp^*$ as the complex rank-$1$ ground truth, we introduce the set
\begin{equation}
\label{eq:R_delta}
\mathcal{R}_\delta
:=
\left\{
\mH \in \mathbb{C}^{d_1 \times d_2}
\,:\,
\left\|\mH - \frac{\mX_\sharp \langle \mX_\sharp, \mH\rangle}{\norm{\mX_\sharp}_\mathrm{F}^2}\right\|_\mathrm{F}
\geq \frac{\sqrt{1-\delta^2} \, \left|\mathrm{Im} \langle \mX_\sharp, \mH \rangle \right|}{\delta \norm{\mX_\sharp}_\mathrm{F}}
\right\}.
\end{equation}
Obviously, $\mathcal{R}_\delta$ is only important if we operate in the complex domain; in the real domain, $\mathcal{R}_\delta$ is the entire space and effectively can be ignored. The following lemma, proved in Appendix~\ref{sec:proof:lemma:imagpart}, the set $\mathcal{R}_\delta$ also contains $\mathbf{\varDelta}$ when $\mX_0$ and $\mX_\sharp$ are at most $\delta$-apart.
%The set $\mathcal{A}_\delta$ defined above only controls the real part of $\langle \vu_\sharp \vv_\sharp^*, \hat{\mX} - \mX_\sharp \rangle$.  We construct another set $\mathcal{R}_\delta$ in the following lemma so that it controls the magnitude of the imaginary part of $\langle \vu_\sharp \vv_\sharp^*, \hat{\mX} - \mX_\sharp \rangle$ as similarly done in \cite{bahmani2017phase}. (See Appendix~\ref{sec:proof:lemma:imagpart} for the proof.)

\begin{lemma}
\label{lemma:imagpart}
With $\mX_\sharp = \sigma_\sharp \vu_\sharp\vv_\sharp^*$, suppose that \eqref{eq:phasea0b0} and 
\begin{equation}
\label{eq:goodanchor}
\left\|\mX_0 - \frac{\mX_\sharp \langle \mX_\sharp, \mX_0\rangle}{\norm{\mX_\sharp}_\mathrm{F}^2}\right\|_\mathrm{F} \leq \delta \norm{\mX_0}_\mathrm{F}
\end{equation}
hold. Then $\hat{\mX} - \mX_\sharp \in \mathcal{R}_\delta$.
\end{lemma}

Finally, based on the arguments in \cite[Theorem~2.1]{bahmani2018solving}, our result depends on the following two key quantities defined with respect to the set $\mathcal{H} = \mathcal{A}_\delta \cap \mathcal{R}_\delta$. First, the Rademacher complexity of $\mathcal{H}$ is defined as
\begin{equation}
\label{eq:defRC}
\mathfrak{C}_M(\mathcal{H})
:= \E
\sup_{\mH \in \mathcal{H} \setminus \{\mathbf{0}\}} \frac{1}{\sqrt{M}}
\sum_{m=1}^M \frac{\epsilon_m
\mathrm{Re}(\langle \mX_\sharp, \mPhi_m \rangle \langle \mPhi_m, \mH \rangle)}{\norm{\mH}_\mathrm{F}},
\end{equation}
where $\epsilon_1,\dots,\epsilon_M$ are i.i.d. Rademacher random variables independent of everything else. Second, for $\tau > 0$, we also consider a variation of small-ball probability that is defined as
\begin{equation}
\label{eq:probtau}
P_\tau(\mathcal{H})
:=
\inf_{\mH \in \mathcal{H}}
\mathbb{P}(\mathrm{Re}(\langle \mX_\sharp, \mPhi_m \rangle \langle \mPhi_m, \mH \rangle) \geq \tau \norm{\mH}_\mathrm{F}).
\end{equation}
 Equipped with these notions, the following theorem provides the accuracy guarantees for the regularized anchored regression in the context of low-rank phase retrieval problem. 

\begin{theorem}[{An adaptation of \cite[Theorem~2.1]{bahmani2018solving} for low-rank phase retrieval}]
\label{thm:armtx}
Suppose that $\mPhi_1,\dots,\mPhi_M$ in \eqref{eq:convex_prog} are independent random matrices, and $\mX_0$ satisfies \eqref{eq:cond_anchor}, \eqref{eq:phasea0b0}, and \eqref{eq:goodanchor}, where $0 < \delta < 1$.  Recalling the definitions \cref{eq:A_delta,eq:R_delta,eq:defRC} and \cref{eq:probtau}, for any $t > 0$, if
\begin{equation}
\label{eq:lbM}
M \geq 4 \, \Big( \frac{\mathfrak{C}_M(\mathcal{A}_\delta \cap \mathcal{R}_\delta) + t \tau}{\tau P_\tau(\mathcal{A}_\delta \cap \mathcal{R}_\delta)} \Big)^2\,,
\end{equation} 
then the solution $\hat{\mX}$ to \eqref{eq:convex_prog} obeys
\begin{equation*}
%\label{eq:esterr}
\inf_{\theta \in [0,2\pi)} \norm{\hat{\mX} - e^{\mathfrak{i} \theta} \mX_\sharp}_\mathrm{F}
\leq \frac{2}{\tau P_\tau(\mathcal{A}_\delta \cap \mathcal{R}_\delta)} \Big( \frac{1}{M} \sum_{m=1}^M |\xi_m| + \epsilon \Big)\,
\end{equation*}
with probability at least $1-e^{-2t^2}$.
\end{theorem}

\begin{rem}
There are a few remarks on Theorem~\ref{thm:armtx} in order. 
\begin{enumerate}
    \item We emphasize again that the required conditions in \eqref{eq:phasea0b0} for the anchor, can be made without loss of generality due to the equivariance property discussed above.
    
	\item The original result in \cite[Theorem 2.1]{bahmani2018solving} considered the problem only in the real domain, where the condition \eqref{eq:phasea0b0} is reduced to the (implicit) assumption $\langle\mX_0,\mX_\sharp\rangle\ge 0$. As mentioned above, in this scenario the set $\mathcal{R}_\delta$ becomes trivial (i.e., $\mathcal{R}_\delta = \mathbb{R}^{d_1 \times d_2}$) as well.

	\item The additive noise $\xi_m$ to the quadratic measurement $|\langle \mPhi_m, \mX_\sharp\rangle|^2$ is arbitrary fixed. Specifically, we assume that $\xi_m$ does not depend on $\mPhi_1,\dots,\mPhi_M$.

    \end{enumerate}
\end{rem}

Theorems~\ref{thm:main_rank1} and \ref{thm:main_iidG} are then obtained from Theorem~\ref{thm:armtx} by specifying key estimates depending on the corresponding measurement matrices. 
For the convenience in computing these estimates, we provide a more explicit characterization of $\mathcal{A}_\delta$ as follows. 
The subdifferential of $\norm{\cdot}_*$ at $\mX_\sharp$, whose SVD is $\mU_\sharp \mSigma_\sharp \mV_\sharp^*$, is expressed as
\begin{equation}
\label{eq:subdiffatXstar}
\partial \norm{\mX_\sharp}_* =
\Big\{
\mZ \,:\, \mathcal{P}_T(\mZ) = \mU_\sharp \mV_\sharp^*, ~
\norm{\mathcal{P}_{T^\perp}(\mZ)} \leq 1
\Big\},
\end{equation}
where $\mathcal{P}_T: \mathbb{C}^{d_1 \times d_2} \to \mathbb{C}^{d_1 \times d_2}$ denotes the orthogonal projection onto the \emph{tangent space} $T$ of the manifold of rank-$r$ matrices at $\mX_\sharp$ given by
\[
T = \left\{ \mU_\sharp \tilde{\mV}^* + \tilde{\mU} \mV_\sharp^* \,:\, \tilde{\mV} \in \mathbb{C}^{d_2 \times r},~ \tilde{\mU} \in \mathbb{C}^{d_1 \times r} \right\}
\]
and $\mathcal{P}_{T^\perp}: \mathbb{C}^{d_1 \times d_2} \to \mathbb{C}^{d_1 \times d_2}$ denotes the projection onto $T^\perp$, the perpendicular subspace of $T$.
By plugging in the expression of the subdifferential in \eqref{eq:subdiffatXstar} to \eqref{eq:A_delta}, we obtain an alternative expression of $\mathcal{A}_\delta$ given by
\begin{equation}
\label{eq:A_delta2}
\mathcal{A}_\delta
=
\Big\{
\mH \in \mathbb{C}^{d_1 \times d_2} \,:\, \sqrt{r} \delta \norm{\mH}_\mathrm{F}
- \lambda \norm{\mathcal{P}_{T^\perp}(\mH)}_*
+ (1-\lambda) \mathrm{Re}\,\langle \mU_\sharp \mV_\sharp^*, \mH \rangle \geq 0
\Big\}.
\end{equation}

\subsection{Proof of Theorem~\ref{thm:main_iidG}}

All matrices and scalars are real-valued in Theorem~\ref{thm:main_iidG}. 
Thus $\mathcal{R}_\delta$ becomes trivial and it suffices to compute estimates of $P_\tau(\mathcal{H})$ and $\mathfrak{C}_M(\mathcal{H})$ for $\mathcal{H} = \mathcal{A}_\delta$.
%Recall that for the simplicity of analysis, we restrict Theorem~\ref{thm:main_iidG} to the case where all measurement functionals and singular vectors of $\mX_\sharp$ are real-valued.  
%The extension to the general complex-valued case requires a bit more computation but the sample complexity remains within the same order.  
The following lemmas respectively provide estimates of $P_\tau(\mathcal{A}_\delta)$ and $\mathfrak{C}_M(\mathcal{A}_\delta)$ whose proofs are deferred to Appendices~\ref{sec:proof:lemma:lbp_iidG} and \ref{sec:proof:lemma:ubrc_iidG}.

\begin{lemma}
\label{lemma:lbp_iidG}
Suppose the hypotheses in Theorem~\ref{thm:main_iidG} hold.  
Then, for any $\tau' > 0$,
\begin{equation}
\label{eq:res:lemma:lbp_iidG}
\inf_{\mH \in \mathcal{A}_\delta}
\mathbb{P}\Big(
\mathrm{Re}(\langle \mX_\sharp, \mPhi_m \rangle \langle \mPhi_m, \mH \rangle)
\geq \tau' \norm{\mX_\sharp}_\mathrm{F} \norm{\mH}_\mathrm{F}
\Big)
\geq \frac{\exp(-20 \tau')}{10}.
\end{equation}
\end{lemma}

\begin{lemma}
\label{lemma:ubrc_iidG}
Suppose the hypotheses in Theorem~\ref{thm:main_iidG} hold. Then
\begin{equation}
\label{eq:res:lemma:ubrc_iidG}
\mathfrak{C}_M(\mathcal{A}_\delta)
\leq \frac{C(1-\lambda+\delta) \norm{\mX_\sharp}_\mathrm{F} \sqrt{r(d_1+d_2) \log(d_1+d_2)}}{\lambda}
\end{equation}
for a numerical constant $C$.
\end{lemma}

To prove Theorem~\ref{thm:main_iidG}, we only need to apply the above estimates in Theorem~\ref{thm:armtx}. 
We first show that the assumptions of Theorem~\ref{thm:main_iidG} are sufficient to invoke Theorem~\ref{thm:armtx}. 
Following the discussion in Section \ref{subsec:general}, the condition \eqref{eq:phasea0b0} can be satisfied without loss of generality by flipping the sign of $\mX_0$ if necessary.
Fix $\tau'$ to a positive constant (e.g., $\tau' = 0.1$). Let $\tau = \tau' \norm{\mX_\sharp}_\mathrm{F}$. 
Then Lemma~\ref{lemma:lbp_iidG} implies that $\tau P_\tau(\mathcal{A}_\delta) \geq c \norm{\mX_\sharp}_\mathrm{F}$ for a numerical constant $c > 0$.  
Choosing $\lambda = 0.9 - \delta$ makes the right-hand side of \eqref{eq:res:lemma:ubrc_iidG} an increasing function of $\delta$.  
Then, by Lemma~\ref{lemma:ubrc_iidG}, the Rademacher complexity $\mathfrak{C}_M(\mathcal{A}_\delta)$ is upper-bounded by $\sqrt{r(d_1+d_2)\log(d_1+d_2)}$ up to a constant solely determined by $\delta$.  
Therefore, \eqref{eq:sampl_comp_iidG} implies that \eqref{eq:lbM} holds whenever $t\tau'$ is dominated by $\sqrt{M}$. 
We can choose $t$ so that the probability of failure is at most $e^{-2t^2} = e^{-c M}$, for some numerical constant $c>0$.
\qed

\subsection{Proof of Theorem~\ref{thm:main_rank1}}

Theorem~\ref{thm:main_rank1} considers recovery of complex-valued rank-$1$ matrices.  
We apply Theorem~\ref{thm:armtx} for $\mathcal{H} = \mathcal{A}_\delta \cap \mathcal{R}_\delta$ to prove Theorem~\ref{thm:main_rank1}.  
The following lemmas, proved in Appendix~\ref{sec:proof:lemma:lbp_rank1} and Appendix~\ref{sec:proof:lemma:ubrc_rank1}, respectively provide a lower bound on $P_\tau(\mathcal{H})$ and an upper bound on $\mathfrak{C}_M(\mathcal{H})$.

\begin{lemma}
\label{lemma:lbp_rank1}
Suppose the hypotheses in Theorem~\ref{thm:main_rank1} hold.  Suppose that $\delta + \lambda < 1$ and $\delta \leq 0.2$.  
Then there exists a numerical constant $\tau' > 0$ such that
\[
\inf_{\mH \in \mathcal{A}_\delta \cap \mathcal{R}_\delta}
\mathbb{P}\Big(
\mathrm{Re}(\vb^* \mX_\sharp^* \va \va^* \mH \vb)
\geq \tau' \norm{\mX_\sharp}_\mathrm{F} \norm{\mH}_\mathrm{F}
\Big)
\geq C_{\tau'},
\]
where $C_{\tau'}$ is a positive numerical constant that only depends on $\tau'$.
\end{lemma}

\begin{lemma}
\label{lemma:ubrc_rank1}
Suppose the hypotheses in Theorem~\ref{thm:main_rank1} hold. Then
\begin{equation}
\label{eq:res:lemma:ubrc_rank1}
\mathfrak{C}_M(\mathcal{A}_\delta)
\leq \frac{C (1-\lambda+\delta) \norm{\mX_\sharp}_\mathrm{F} \sqrt{d_1+d_2} \, \log M}{\lambda}
\end{equation}
for a numerical constant $C$.
\end{lemma}

The error bound in \eqref{eq:thm:esterr_rank1} then follows from Theorem~\ref{thm:armtx} with the above estimates given by Lemmas~\ref{lemma:lbp_rank1} and \ref{lemma:ubrc_rank1}.  
To apply Lemma~\ref{lemma:lbp_rank1}, we choose $\lambda = 0.9-\delta$.  
Then, similar to the proof of Theorem~\ref{thm:main_iidG}, the factor $(1-\lambda+\delta)/\lambda$ becomes an increasing function in $\delta$.  
The constant $C_\delta$ is given by this function of $\delta$ together with the result of Lemma~\ref{lemma:lbp_rank1}.

Finally, the error bound for the estimation of $\vu$ and $\vv$ in \eqref{eq:thm:esterr_factors_rank1} follows immediately from the Davis-Kahan Theorem (Theorem~\ref{thm:sintheta}). \qed

\section{Numerical Results}
We have conducted a Monte Carlo simulation to study the empirical performance of the proposed convex programs. Specifically, we considered the optimization problem in \eqref{eq:convex_prog_noiseless} in the noiseless case where the measurement matrices are given as the outer product of two Gaussian random vectors and the unknown rank-$1$ matrix is a square matrix ($d_1 = d_2 = d$). To solve \eqref{eq:convex_prog_noiseless}, we used the software package \texttt{TFOCS} \cite{becker2011templates} that uses a smoothed conic dual formulation.

\begin{figure}
	\centering
    \includegraphics[scale=0.75]{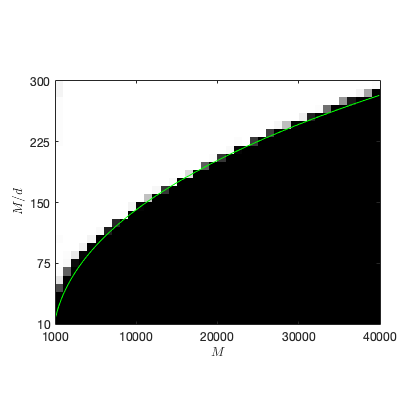}
    \vspace{-15mm}
    \caption{Empirical phase transition in the noiseless case with rank-$1$ measurements. The success rate out of 100 trials is plotted in a gray sclae (white: all success, black: all failure).}
    \label{fig:empPT}
\end{figure}

Figure~\ref{fig:empPT} illustrates the empirical phase transition. For a fixed number of measurements $M$, we vary the matrix size $d$ where the ratio $M/d$ belongs to a given interval. In Figure~\ref{fig:empPT}, the convex program provides the exact recovery when $d$ is below a certain threshold determined by $M$. The sample complexity result by Theorem~\ref{thm:main_rank1} and Lemma~\ref{lemma:init} quantifies this threshold as $C M / \log^\alpha M$ for some constants $C, \alpha > 0$. Alternatively, if the oversampling rate $M/d$ exceeds a polylog factor of $M$, then the convex program provides the exact recovery. The empirical phase transition occurs at $M/d \approx 0.14 \log^5 M$ or $d \approx 7.3 M/\log^5 M$ indicated by the green curve in the figure. Although, the requirements for the constants $C$ and $\alpha$ in our proofs seem conservative, our theory is consistent with the empirical performance up to the choice of these constants.

\section{Discussions}

We proposed a simple initial estimation using partial traces. The regularized anchored regression with the nuclear norm given by this initial estimate provides a stable estimate for LRPR. Performance guarantees were derived for several random measurement models.

The anchored regression was originally proposed for the plain phase retrieval problem and later modified to the regularized version to accommodate a geometric prior on the solution. 
There also exist alternative methods for phase retrieval and their modification with prior signal models. It would be possible to adapt the Wirtinger flow \cite{candes2015phase} and its variation for the sparsity prior \cite{cai2016optimal} to the low-rankness model. To fully convexify LRPR without requiring any initial estimate, one may apply the lifting-reformulation twice, which will provide a linear inverse problem where the solution is rearranged as a 4-way tensor of rank-$1$. While the rank-$1$ prior of the tensor can be promoted by a convex regularizer by the tensor nuclear norm, it is proven NP-hard to compute the tensor nuclear norm \cite{hillar2013most}. 

\section*{Acknowledgements}
This work was supported in part by NSF CCF-1718771, by C-BRIC, one of six centers in JUMP, a Semiconductor Research Corporation (SRC) program sponsored by DARPA, and by the EU Horizon 2020 research and innovation program under 646804-ERC-COG-BNYQ.  
The authors thank the anonymous reviewers for their careful reading of the manuscript and their many insightful comments and suggestions.

\appendix

\section{Expectations of symmetric Gaussian tensors}

We repeatedly use the expectation of various tensor products of an i.i.d. Gaussian vector, which are summarized below.  First we consider the expectation of the fourth-order tensor product.

\begin{lemma}
\label{lemma:exp_gauss_quad}
Let $\vg \sim \mathcal{N}(\vzero,\mId_d)$.  Then
\begin{align*}
\mathbb{E} \, \vg \otimes \vg \otimes \vg \otimes \vg
&= \sum_{j,k=1}^d (
\ve_j \otimes \ve_j \otimes \ve_k \otimes \ve_k
+ \ve_j \otimes \ve_k \otimes \ve_j \otimes \ve_k
+ \ve_j \otimes \ve_k \otimes \ve_k \otimes \ve_j
),
\end{align*}
where $\ve_j$ denotes the $j$th column of $\mId_d$ for $j=1,\dots,d$.
\end{lemma}

\begin{proof}[Proof of Lemma~\ref{lemma:exp_gauss_quad}]
The expectation of $\vg \otimes \vg \otimes \vg \otimes \vg$ is written as
\begin{align*}
& \mathbb{E} \sum_{i_1,i_2,i_3,i_4 = 1}^d
(\ve_{i_1} \ve_{i_1}^\transpose \otimes \ve_{i_2} \ve_{i_2}^\transpose \otimes \ve_{i_3} \ve_{i_3}^\transpose \otimes \ve_{i_4} \ve_{i_4}^\transpose)
(\vg \otimes \vg \otimes \vg \otimes \vg) \\
&= \sum_{i_1,i_2,i_3,i_4 = 1}^d
\mathbb{E} g_{i_1} g_{i_2} g_{i_3} g_{i_4}
(\ve_{i_1} \otimes \ve_{i_2} \otimes \ve_{i_3} \otimes \ve_{i_4}),
\end{align*}
where $g_i$ denotes the $i$th entry of $\vg$ for $i=1,\dots,d$.  The proof completes by noting that all odd moments of a standard normal variable vanish.
\end{proof}

The following lemma is a direct consequence of Lemma~\ref{lemma:exp_gauss_quad}.

\begin{lemma}
\label{lemma:exp_gauss_quad_ip}
Let $\vx, \vy \in \mathbb{R}^d$ and $\vg \sim \mathcal{N}(\vzero,\mId_d)$.  Then
\[
\mathbb{E} (\vx^\transpose \vg \vg^\transpose \vy) \vg \vg^\transpose
= (\vx^\transpose \vy) \mId_d + \vx \vy^\transpose + \vy \vx^\transpose.
\]
\end{lemma}

Next we consider the expectation of an 8-way tensor product applying to a fourth-order tensor product of a unit vector.

\begin{lemma}
\label{lemma:exp_gauss_octa}
Let $\vx \in \mathbb{S}^{d-1}$ and $\vg \sim \mathcal{N}(\vzero,\mId_d)$.  Then
\begin{equation}
\begin{aligned}
{} & \mathbb{E} (\vx^\transpose \vg)^4 (\vg \otimes \vg \otimes \vg \otimes \vg)
%\\
%{} &
= 24 (\vx \otimes \vx \otimes \vx \otimes \vx) \\
{} & \qquad + 12 \sum_{l=1}^d (
\vx \otimes \vx \otimes \ve_l \otimes \ve_l
+ \vx \otimes \ve_l \otimes \vx \otimes \ve_l
+ \vx \otimes \ve_l \otimes \ve_l \otimes \vx \\
{} & \qquad \quad \quad \quad \quad + \ve_l \otimes \vx \otimes \vx \otimes \ve_l
+ \ve_l \otimes \vx \otimes \ve_l \otimes \vx
+ \ve_l \otimes \ve_l \otimes \vx \otimes \vx) \\
{} & \qquad + 3 \sum_{j,k=1}^d (
\ve_j \otimes \ve_j \otimes \ve_k \otimes \ve_k
+ \ve_j \otimes \ve_k \otimes \ve_j \otimes \ve_k
+ \ve_j \otimes \ve_k \otimes \ve_k \otimes \ve_j
),
\end{aligned}
\label{res:lemma:exp_gauss_octa}
\end{equation}
where $\ve_l$ denotes the $j$th column of $\mId_d$ for $l=1,\dots,d$.
\end{lemma}

\begin{proof}[Proof of Lemma~\ref{lemma:exp_gauss_octa}]

The expectation $\mathbb{E} (\vx^\transpose \vg)^4 (\vg \otimes \vg \otimes \vg \otimes \vg)$ is rewritten as
\begin{align}
& \mathbb{E} (\vg \otimes \vg \otimes \vg \otimes \vg) (\vg \otimes \vg \otimes \vg \otimes \vg)^\transpose (\vx \otimes \vx \otimes \vx \otimes \vx) \nonumber \\
&= \sum_{\mB_1,\mB_2,\mB_3,\mB_4 \in \{\mP_{\vx},\mP_{\vx^\perp}\}}
\mathbb{E} (\mB_1 \otimes \mB_2 \otimes \mB_3 \otimes \mB_4) (\vg \otimes \vg \otimes \vg \otimes \vg) (\vg \otimes \vg \otimes \vg \otimes \vg)^\transpose (\vx \otimes \vx \otimes \vx \otimes \vx), \label{eq1:proof:lemma:exp_gauss_octa}
\end{align}
where $\mP_{\vx}$ and $\mP_{\vx^\perp}$ denote the orthogonal projection operators onto the subspace spanned by $\vx$ and its orthogonal complement, respectively.

If any of $\mB_1,\mB_2,\mB_3,\mB_4$ is different from the other three matrices, then the corresponding summand in \eqref{eq1:proof:lemma:exp_gauss_octa} becomes zero since it has a factor that is an odd moment of $\bm{x}^\transpose \bm{g} \sim \mathcal{N}(0,1)$.  Therefore, it suffices to consider the following three cases.

\noindent\textbf{Case 1:} $\mB_1 = \mB_2 = \mB_3 = \mB_4 = \mP_{\vx}$.

Since
\[
\mP_{\vx} \otimes \mP_{\vx} \otimes \mP_{\vx} \otimes \mP_{\vx}
= (\vx \vx^\transpose \otimes \vx \vx^\transpose \otimes \vx \vx^\transpose \otimes \vx \vx^\transpose)
= (\vx \otimes \vx \otimes \vx \otimes \vx) (\vx \otimes \vx \otimes \vx \otimes \vx)^\transpose,
\]
it follows that the corresponding summand is written as
\begin{equation}
\begin{aligned}
& (\vx \otimes \vx \otimes \vx \otimes \vx)
\mathbb{E} (\vx \otimes \vx \otimes \vx \otimes \vx)^\transpose (\vg \otimes \vg \otimes \vg \otimes \vg) (\vg \otimes \vg \otimes \vg \otimes \vg)^\transpose (\vx \otimes \vx \otimes \vx \otimes \vx) \\
&= \mathbb{E} (\vx^\transpose \vg)^8 (\vx \otimes \vx \otimes \vx \otimes \vx)
= 105 (\vx \otimes \vx \otimes \vx \otimes \vx).
\end{aligned}
\label{eq2:proof:lemma:exp_gauss_octa}
\end{equation}

\noindent\textbf{Case 2:} Two of $\mB_1, \mB_2, \mB_3, \mB_4$ are $\mP_{\vx}$ and the other two matrices are $\mP_{\vx^\perp}$.

First we consider the sub-case where $\mB_1 = \mB_2 = \mP_{\vx}$ and $\mB_3 = \mB_4 = \mP_{\vx^\perp}$.  Since $\mP_{\vx^\perp} \vg$ and $\vx^\transpose \vg$ are independent, we can replace $\vx^\transpose \vg$ by $\vx^\transpose \vg'$ where $\vg'$ is an independent copy of $\vg$.  Then the corresponding summand is written as
\begin{align*}
& \mathbb{E}_{\vg'} (\vg'^\transpose \vx)^6 \, \mathbb{E}_{\vg} \mP_{\vx^\perp} \vg \otimes \mP_{\vx^\perp} \vg \otimes \vx \otimes \vx
%\\
%&
= 15 (\mP_{\vx^\perp} \otimes \mP_{\vx^\perp}) (\mathbb{E}_{\vg} \vg \otimes \vg) \otimes \vx \otimes \vx \\
& = 15 (\mP_{\vx^\perp} \otimes \mP_{\vx^\perp}) \mathrm{vec} (\mathbb{E}_{\vg} \vg \vg^\transpose) \otimes \vx \otimes \vx
%\\
%&
= 15 (\mP_{\vx^\perp} \otimes \mP_{\vx^\perp}) \mathrm{vec} (\mId_d) \otimes \vx \otimes \vx \\
& = 15 \, \mathrm{vec} (\mP_{\vx^\perp} \mId_d \mP_{\vx^\perp}) \otimes \vx \otimes \vx %\\
%&
= 15 \, \mathrm{vec} (\mId_d - \mP_{\vx}) \otimes \vx \otimes \vx \\
& = 15 \Big(- \vx \otimes \vx \otimes \vx \otimes \vx + \sum_{l=1}^d \ve_l \otimes \ve_l \otimes \vx \otimes \vx\Big).
\end{align*}

The summands corresponding to the other sub-cases of Case~2 are calculated similarly, and the partial summation of \eqref{eq1:proof:lemma:exp_gauss_octa} for Case~2 is written as
\begin{equation}
\begin{aligned}
{} & - 90 \, \vx \otimes \vx \otimes \vx \otimes \vx + 15 \sum_{l=1}^d (
\vx \otimes \vx \otimes \ve_l \otimes \ve_l
+ \vx \otimes \ve_l \otimes \vx \otimes \ve_l
+ \vx \otimes \ve_l \otimes \ve_l \otimes \vx \\
{} & \quad \quad \quad \quad + \ve_l \otimes \vx \otimes \vx \otimes \ve_l
+ \ve_l \otimes \vx \otimes \ve_l \otimes \vx
+ \ve_l \otimes \ve_l \otimes \vx \otimes \vx).
\end{aligned}
\label{eq3:proof:lemma:exp_gauss_octa}
\end{equation}

\noindent\textbf{Case 3:} $\mB_1 = \mB_2 = \mB_3 = \mB_4 = \mP_{\vx^\perp}$.

Again by the independence between $\mP_{\vx^\perp} \vg$ and $\vx^\transpose \vg$, the corresponding summand is written as
\begin{equation}
\begin{aligned}
& \mathbb{E}_{\vg'} (\vg'^\transpose \vx)^4 \, \mathbb{E}_{\vg} (\mP_{\vx^\perp} \otimes \mP_{\vx^\perp} \otimes \mP_{\vx^\perp} \otimes \mP_{\vx^\perp}) (\vg \otimes \vg \otimes \vg \otimes \vg) \\
& = 3 (\mP_{\vx^\perp} \otimes \mP_{\vx^\perp} \otimes \mP_{\vx^\perp} \otimes \mP_{\vx^\perp}) \mathbb{E}_{\vg} (\vg \otimes \vg \otimes \vg \otimes \vg).
\end{aligned}
\label{eq4a:proof:lemma:exp_gauss_octa}
\end{equation}
By plugging in the expression of $\mathbb{E} \, \vg \otimes \vg \otimes \vg \otimes \vg$ in Lemma~\ref{lemma:exp_gauss_quad}, the right-hand side of \eqref{eq4a:proof:lemma:exp_gauss_octa} is written as
\begin{equation}
\begin{aligned}
& 3 \underbrace{
(\mP_{\vx^\perp} \otimes \mP_{\vx^\perp} \otimes \mP_{\vx^\perp} \otimes \mP_{\vx^\perp}) \sum_{j,k=1}^d \ve_j \otimes \ve_j \otimes \ve_k \otimes \ve_k
}_{\text{($\S$)}} \\
& \quad + 3 \underbrace{
(\mP_{\vx^\perp} \otimes \mP_{\vx^\perp} \otimes \mP_{\vx^\perp} \otimes \mP_{\vx^\perp}) \sum_{j,k=1}^d \ve_j \otimes \ve_k \otimes \ve_j \otimes \ve_k
}_{\text{($\S\S$)}} \\
& \quad + 3 \underbrace{
(\mP_{\vx^\perp} \otimes \mP_{\vx^\perp} \otimes \mP_{\vx^\perp} \otimes \mP_{\vx^\perp}) \sum_{j,k=1}^d \ve_j \otimes \ve_k \otimes \ve_k \otimes \ve_j
}_{\text{($\S\S\S$)}}.
\end{aligned}
\label{eq4:proof:lemma:exp_gauss_octa}
\end{equation}
The first term ($\S$) in \eqref{eq4:proof:lemma:exp_gauss_octa} is rewritten as
\begin{align*}
\text{($\S$)} & = (\mP_{\vx^\perp} \otimes \mP_{\vx^\perp} \otimes \mP_{\vx^\perp} \otimes \mP_{\vx^\perp}) [\mathrm{vec}(\mId_d) \otimes \mathrm{vec}(\mId_d)] \\
& = [(\mP_{\vx^\perp} \otimes \mP_{\vx^\perp}) \mathrm{vec}(\mId_d)] \otimes [(\mP_{\vx^\perp} \otimes \mP_{\vx^\perp}) \mathrm{vec}(\mId_d)] \\
& = \mathrm{vec}(\mP_{\vx^\perp}) \otimes \mathrm{vec}(\mP_{\vx^\perp})
= \mathrm{vec}(\mId_d - \mP_{\vx}) \otimes \mathrm{vec}(\mId_d - \mP_{\vx}) \\
& = \mathrm{vec}(\mP_{\vx}) \otimes \mathrm{vec}(\mP_{\vx}) + \mathrm{vec}(\mId_d) \otimes \mathrm{vec}(\mId_d)
- \mathrm{vec}(\mId_d) \otimes \mathrm{vec}(\mP_{\vx})
- \mathrm{vec}(\mP_{\vx}) \otimes \mathrm{vec}(\mId_d) \\
& = \vx \otimes \vx \otimes \vx \otimes \vx + \sum_{j,k=1}^d \ve_j \otimes \ve_j \otimes \ve_k \otimes \ve_k
- \sum_{l=1}^d
(\vx \otimes \vx \otimes \ve_l \otimes \ve_l + \ve_l \otimes \ve_l \otimes \vx \otimes \vx).
\end{align*}
Similarly ($\S\S$) and ($\S\S\S$) are written as the sum of rank-$1$ tensors.  Then applying these results to \eqref{eq4:proof:lemma:exp_gauss_octa} provides
\begin{equation}
\begin{aligned}
& \mathbb{E}_{\vg'} (\vg'^\transpose \vx)^4 \, \mathbb{E}_{\vg} (\mP_{\vx^\perp} \otimes \mP_{\vx^\perp} \otimes \mP_{\vx^\perp} \otimes \mP_{\vx^\perp}) (\vg \otimes \vg \otimes \vg \otimes \vg) \\
& = 9 \, \vx \otimes \vx \otimes \vx \otimes \vx
- 3 \sum_{l=1}^d (
\vx \otimes \vx \otimes \ve_l \otimes \ve_l
+ \vx \otimes \ve_l \otimes \vx \otimes \ve_l
+ \vx \otimes \ve_l \otimes \ve_l \otimes \vx \\
{} & \quad \quad \quad \quad + \ve_l \otimes \vx \otimes \vx \otimes \ve_l
+ \ve_l \otimes \vx \otimes \ve_l \otimes \vx
+ \ve_l \otimes \ve_l \otimes \vx \otimes \vx) \\
{} & + 3 \sum_{j,k=1}^d (
\ve_j \otimes \ve_j \otimes \ve_k \otimes \ve_k
+ \ve_j \otimes \ve_k \otimes \ve_j \otimes \ve_k
+ \ve_j \otimes \ve_k \otimes \ve_k \otimes \ve_j
).
\end{aligned}
\label{eq6:proof:lemma:exp_gauss_octa}
\end{equation}
The identity in \eqref{res:lemma:exp_gauss_octa} is then obtained by combining \eqref{eq2:proof:lemma:exp_gauss_octa}, \eqref{eq3:proof:lemma:exp_gauss_octa}, and \eqref{eq6:proof:lemma:exp_gauss_octa} through \eqref{eq1:proof:lemma:exp_gauss_octa}.
\end{proof}

\section{Moment and tail bounds of random matrices}

The following lemma, which provides a central moment bound on a standard normal variable, is a direct consequence of the Khintchine inequality (e.g., \cite[Corollary~5.12]{vershynin2012introduction}).

\begin{lemma}
\label{lemma:gaussianLP}
Let $g \sim \mathcal{N}(0,1)$.  Then there exists a numerical constant $C$ such that
\[
(\E |g|^p)^{1/p} \leq C \sqrt{p}, \quad \forall p \in \mathbb{N}.
\]
\end{lemma}

We also use moment and tail bounds of random matrices in the spectral norm given by the noncommutative Rosenthal inequality \cite[Theorem~0.4]{junge2013noncommutative}.

\begin{theorem}[{Noncommutative Rosenthal inequality \cite[Theorem~0.4]{junge2013noncommutative}}]
\label{thm:ncrosenthal}
Let $\mY_1,\dots,\mY_M$ be independent random matrices with zero-mean.  Then there exists a numerical constant $C$ such that
\[
\Big(\E \Big\| \sum_{m=1}^M \mY_m \Big\|^p\Big)^{1/p}
\leq C \Big[
\sqrt{p} \Big( \Big\| \sum_{m=1}^M \E \mY_m \mY_m^* \Big\|^{1/2}
\vee
\Big\| \sum_{m=1}^M \E \mY_m^* \mY_m \Big\|^{1/2}
\Big)
\vee
p \Big( \sum_{m=1}^M \E \norm{\mY_m}^p \Big)^{1/p}
\Big]
\]
for all $1 \leq p < \infty$.
\end{theorem}
Then the following lemma follows immediately from Theorem~\ref{thm:ncrosenthal}.

\begin{lemma}
\label{lemma:wgaussouterproduct}
Let $\vg_1,\dots,\vg_M \in \mathbb{R}^d$ be independent copies of $\vg \sim \mathcal{N}(\vzero,\mId_d)$, $\vlambda = [\lambda_1,\dots,\lambda_M]^\transpose \in \mathbb{R}^M$, and $\nu \in (0,1)$.  Then there exist numerical constants $C_1,C_2 > 0$ such that
\begin{equation}
\label{eq:wgaussouterproduct_moment}
\Big( \E \Big\| \frac{1}{M} \sum_{m=1}^M \lambda_m (\vg_m \vg_m^\transpose - \mId_d) \Big\|^p \Big)^{1/p}
\leq C_1 \norm{\vlambda}_\infty \Big[ M^{-1/2} \sqrt{p d} + M^{1/p-1} p (d + p) \Big]
\end{equation}
for all $p \in \mathbb{N}$, and
\[
\Big\| \frac{1}{M} \sum_{m=1}^M \lambda_m (\vg_m \vg_m^\transpose - \mId_d) \Big\|
\leq \delta
\]
holds with probability $1 - \nu$ provided
\begin{equation}
\label{eq:wgaussouterproduct_M}
M \geq C_2 \left(\delta^{-1} \norm{\vlambda}_\infty \vee \delta^{-2} \norm{\vlambda}_\infty^2\right) \left(d \log(M/\nu) \vee \log^2(M/\nu)\right).
\end{equation}
\end{lemma}

\begin{proof}[Proof of Lemma~\ref{lemma:wgaussouterproduct}]
We apply Theorem~\ref{thm:ncrosenthal} for $\mY_m = \lambda_m (\vg_m \vg_m^\transpose - \mId_d)$ for $m=1,\dots,M$. By the traingle inequality, we have 
\[
(\mathbb{E} \norm{\mY_m}^p)^{1/p} 
\leq \lambda_m + \lambda_m (\mathbb{E} \norm{\vg_m \vg_m^\transpose}^p)^{1/p}
= \lambda_m + \lambda_m (\mathbb{E} \norm{\vg_m}_2^{2p})^{1/p}
\leq C_1 \lambda_m (d + p).
\]
Here the last step follows since 
\[
\|\norm{\vg}_2 - \sqrt{d}\|_{L_{2p}} \leq C \sqrt{2p} \left\| \norm{\vg}_2 - \sqrt{d} \, \right\|_{\psi_2} \leq C' \sqrt{p},
\]
where $\norm{\cdot}_{\psi_2}$ denotes the subgaussian norm.
Therefore we obtain 
\begin{equation}
\label{eq:wgaussouterproduct_moment:eq1}
\Big( \sum_{m=1}^M \E \norm{\mY_m}^p \Big)^{1/p} 
\leq C_3 \|\vlambda\|_\infty M^{1/p} (d + p).
\end{equation}

Furthermore, the expectation of $\mY_m^2 = \lambda_m^2( \vg_m \vg_m^\transpose \vg_m \vg_m^\transpose - 2 \vg_m \vg_m^\transpose + \mId_d)$ is computed by using Lemma~\ref{lemma:exp_gauss_quad_ip} as $\mathbb{E} \mY_m^2 = \lambda_m^2 (d+1) \mId_d$. Therefore it follows that
\begin{equation}
\label{eq:wgaussouterproduct_moment:eq2}
\Big\| \sum_{m=1}^M \E \mY_m^2 \Big\|^{1/2} = \sqrt{d+1} \|\vlambda\|_2 \leq C_4 \sqrt{Md} \|\vlambda\|_\infty. 
\end{equation}
Then \eqref{eq:wgaussouterproduct_moment} is obtained by plugging in \eqref{eq:wgaussouterproduct_moment:eq1} and \eqref{eq:wgaussouterproduct_moment:eq2} to Theorem~\ref{thm:ncrosenthal}. 

Next, by the Markov inequality, we have
\begin{align}
\mathbb{P}
\Big(
\Big\| \frac{1}{M} \sum_{m=1}^M \mY_m \Big\|
> \delta
\Big)
& \leq
\delta^{-p} \E \Big\| \frac{1}{M} \sum_{m=1}^M \mY_m \Big\|^p \nonumber \\
& \leq C_1 \delta^{-p} \norm{\vlambda}_\infty^p \Big[ M^{-1/2} \sqrt{p d} + M^{1/p-1} p (d + p) \Big]^p. \label{eq:wgaussouterproduct_moment:eq3}
%& \leq C_1 \delta^{-p} \norm{\vlambda}_\infty \Big[ M^{-1/2} \sqrt{p d} + M^{1/p-1} p (d + p^2 d^{1/p}) \Big]. \label{eq:wgaussouterproduct_moment:eq3}
\end{align}
Let $p = \log (M/\nu)$. Then \eqref{eq:wgaussouterproduct_M} implies that the right-hand side of \eqref{eq:wgaussouterproduct_moment:eq3} is upper-bounded by $\nu$. This completes the proof.
\end{proof}

\section{Proof of Lemma~\ref{lemma:init}}
%\label{sec:proof:lemma:init_rank1}

Let $\phi := \angle(\vu_0,\vu_\sharp)$ and $\psi := \angle(\vv_0,\vv_\sharp)$.
Then
\[
\inf_{\theta \in [0,2\pi)} \norm{\vu_0 \vv_0^\transpose - e^{\mathfrak{i} \theta} \vu_\sharp \vv_\sharp^\transpose}_\mathrm{F}^2
= 2 - 2 \cos\phi \cos\psi
\leq 2 - 2 \cos^2(\phi \vee \psi) = 2 \sin^2(\phi \vee \psi).
\]
Therefore, it suffices to show
\[
\sin (\phi \vee \psi) = \sin\phi \vee \sin\psi \leq \frac{\delta}{\sqrt{2}}.
\]

We will only show $\sin\phi \leq \sqrt{\delta/2}$.  The derivation of the other part is essentially the same due to symmetry.  Without loss of generality, we assume $\norm{\mX_\sharp}_\mathrm{F} = 1$ (or equivalently $\sigma_\sharp = 1$).

Since $\mX_\sharp$ is a scalar multiple of the most dominant eigenvector of $\E \mUpsilon$, we use the Davis-Kahan theorem \cite{davis1970rotation} to bound the error in estimating $\vu_\sharp$ as the dominant eigenvector of $\mUpsilon$.  
Among variations of the Davis-Kahan theorem, we use the version given in terms of the principal angle between two subspaces.  
The following theorem states this result and is obtained by combining the argument of \cite[Corollary~7.2.6]{golub2012matrix} and the $\sin\theta$ theorem for any unitarily invariant norm \cite{wedin1972perturbation}.

\begin{theorem}[{Davis-Kahan $\sin\theta$ theorem}]
	\label{thm:sintheta}
	Let $\mA, \mDelta \in \mathbb{C}^{n \times n}$ satisfy that $\mA$ and $\mA+\mDelta$ are positive semidefinite. Let $\mQ \in \mathbb{C}^{n \times r}$ (resp. $\hat{\mQ} \in \mathbb{C}^{n \times r}$) denote the matrix whose columns are the eigenvectors of $\mA$ (resp. $\mA+\mDelta$) corresponding to the $r$-largest eigenvalues.  Suppose that $\lambda_r(\mA) > \lambda_{r+1}(\mA)$.  If
	\begin{equation*}
%		\label{eq:smallperturb}
		\norm{\mDelta} \leq \frac{\lambda_r(\mA)-\lambda_{r+1}(\mA)}{5},
	\end{equation*}
	then
	\begin{equation*}
%		\label{eq:errbnd}
		\sin\angle(\mathrm{span}(\mQ),\mathrm{span}(\hat{\mQ})) \leq \frac{4 \norm{\mDelta}}{\lambda_r(\mA)-\lambda_{r+1}(\mA)}.	
	\end{equation*}
\end{theorem}

To prove Lemma~\ref{lemma:init}, we apply Theorem~\ref{thm:sintheta} to $\mA = \E \mUpsilon$ and $\mDelta = \mUpsilon - \E \mUpsilon$ with $r=1$.
Since
\[
\mA = \vu_\sharp \vu_\sharp^* + \Big(1 + \frac{1}{M} \sum_{m=1}^M \xi_m\Big) \mId_{d_1},
\]
it follows that
\[
\lambda_k(\mA) = \lambda_k(\vu_\sharp \vu_\sharp^*) + 1 + \frac{1}{M} \sum_{m=1}^M \xi_m.
\]
Therefore we obtain
\[
\lambda_1(\mA) - \lambda_2(\mA) = \lambda_1(\vu_\sharp \vu_\sharp^*) - \lambda_2(\vu_\sharp \vu_\sharp^*) = 1.
\]

It remains to show
\begin{equation}
\label{eq:small_perturb_rank1}
\norm{\mDelta} \leq \frac{\delta}{4\sqrt{2}}.
\end{equation}

Let us first decompose $\mDelta$ into its noise-free portion and the remainder as
\[
\mDelta
= \frac{1}{M} \sum_{m=1}^M |\vb_m^* \vv_\sharp|^2 |\va_m^* \vu_\sharp|^2 \va_m \va_m^* - \E |\vb_m^* \vv_\sharp|^2 |\va_m^* \vu_\sharp|^2 \va_m \va_m^*
+ \frac{1}{M} \sum_{m=1}^M \xi_m (\va_m \va_m^* - \mId_{d_1}).
\]
Then \eqref{eq:small_perturb_rank1} is implied by
\begin{equation}
\Big\| \frac{1}{M} \sum_{m=1}^M |\vb_m^* \vv_\sharp|^2 |\va_m^* \vu_\sharp|^2 \va_m \va_m^* - \E |\vb_m^* \vv_\sharp|^2 |\va_m^* \vu_\sharp|^2 \va_m \va_m^* \Big\|
\leq \frac{\delta}{8\sqrt{2}}
\label{eq:small_perturb_rank1_part1}
\end{equation}
and
\begin{equation}
\Big\| \frac{1}{M} \sum_{m=1}^M \xi_m (\va_m \va_m^* - \mId_{d_1}) \Big\|\leq \frac{\delta}{8\sqrt{2}}.
\label{eq:small_perturb_rank1_part2}
\end{equation}
Indeed, by Lemma~\ref{lemma:wgaussouterproduct}, \eqref{eq:snr_cond_init2} implies that \eqref{eq:small_perturb_rank1_part2} holds with probability $1 - \nu/2$ where $\nu = M^{-\alpha}$.

In the remainder of the proof, we show \eqref{eq:sampl_comp_init} implies \eqref{eq:small_perturb_rank1_part1} with probability $1-\nu/2$.
Let
\[
\mY_m = \mZ_m - \E \mZ_m, \quad m=1,\dots,M,
\]
where
\begin{equation}
\label{eq:def_Zm}
\mZ_m = |\vb_m^* \vv_\sharp|^2 |\va_m^* \vu_\sharp|^2 \va_m \va_m^*.
\end{equation}
Then \eqref{eq:small_perturb_rank1_part1} is written as
\begin{equation}
\Big\| \frac{1}{M} \sum_{m=1}^M \mY_m \Big\|
\leq \frac{\delta}{8\sqrt{2}}.
\label{eq:small_perturb_rank1_part1a}
\end{equation}

To show \eqref{eq:small_perturb_rank1_part1a}, we use the noncommutative Rosenthal inequality in Theorem \ref{thm:ncrosenthal}.
By direct calculation, we obtain
\[
\E \mZ_m = \vu_\sharp \vu_\sharp^* + \mId_{d_1}.
\]
Next, by plugging in \eqref{eq:def_Zm} into $\E \mZ_m^* \mZ_m$, we obtain
\begin{equation}
\label{eq:init_var}
\E \mZ_m^* \mZ_m = \E |\vb_m^* \vv_\sharp|^4 |\va_m^* \vu_\sharp|^4 \va_m \va_m^* \va_m \va_m^*.
\end{equation}
By decomposing the right-hand side of \eqref{eq:init_var} with $\mP_{\vu} + \mP_{\vu^\perp} = \mId_{d_1}$, $\E \mZ_m^* \mZ_m$ is rewritten as
\begin{subequations}
\label{eq:init_var_part1a}
\begin{align}
\E \mZ_m^* \mZ_m
&= \E |\vb_m^* \vv_\sharp|^4 |\va_m^* \vu_\sharp|^4 \mP_{\vu_\sharp} \va_m \va_m^* \mP_{\vu_\sharp} \va_m \va_m^* \mP_{\vu_\sharp} \label{eq:init_var_part1} \\
&+ \E |\vb_m^* \vv_\sharp|^4 |\va_m^* \vu_\sharp|^4 \mP_{\vu_\sharp} \va_m \va_m^* \mP_{\vu_\sharp^\perp} \va_m \va_m^* \mP_{\vu_\sharp} \label{eq:init_var_part2} \\
&+ \E |\vb_m^* \vv_\sharp|^4 |\va_m^* \vu_\sharp|^4 \mP_{\vu_\sharp^\perp} \va_m \va_m^* \mP_{\vu_\sharp} \va_m \va_m^* \mP_{\vu_\sharp^\perp} \label{eq:init_var_part3a} \\
&+ \E |\vb_m^* \vv_\sharp|^4 |\va_m^* \vu_\sharp|^4 \mP_{\vu_\sharp^\perp} \va_m \va_m^* \mP_{\vu_\sharp^\perp} \va_m \va_m^* \mP_{\vu_\sharp^\perp} \label{eq:init_var_part3b} \\
&+ \E |\vb_m^* \vv_\sharp|^4 |\va_m^* \vu_\sharp|^4 \mP_{\vu_\sharp} \va_m \va_m^* \mP_{\vu_\sharp^\perp} \va_m \va_m^* \mP_{\vu_\sharp^\perp} \label{eq:init_var_part4} \\
&+ \E |\vb_m^* \vv_\sharp|^4 |\va_m^* \vu_\sharp|^4 \mP_{\vu_\sharp^\perp} \va_m \va_m^* \mP_{\vu_\sharp^\perp} \va_m \va_m^* \mP_{\vu_\sharp} \label{eq:init_var_part5} \\
&+ \E |\vb_m^* \vv_\sharp|^4 |\va_m^* \vu_\sharp|^4 \mP_{\vu_\sharp} \va_m \va_m^* \mP_{\vu_\sharp} \va_m \va_m^* \mP_{\vu_\sharp^\perp} \label{eq:init_var_part6} \\
&+ \E |\vb_m^* \vv_\sharp|^4 |\va_m^* \vu_\sharp|^4 \mP_{\vu_\sharp^\perp} \va_m \va_m^* \mP_{\vu_\sharp} \va_m \va_m^* \mP_{\vu_\sharp}. \label{eq:init_var_part7}
\end{align}
\end{subequations}

Since $\vu_\sharp^* \va_m$ and $\mP_{\vu_\sharp^\perp} \va_m$ are independent, which follows from $\va_m \sim \mathcal{CN}(\vzero,\mId_{d_1})$, we can substitute $\mP_{\vu_\sharp^\perp} \va_m$ by $\mP_{\vu_\sharp^\perp} \breve{\va}_m$, where $\breve{\va}_m$ is an independent copy of $\va_m$. For a standard complex Gaussian random variable $\breve{g} \sim \mathcal{CN}(0,1)$, we have
\[
\E |\breve{g}|^2 = 1, ~ \E |\breve{g}|^4 = 2,
~ \E |\breve{g}|^6 = 6, ~ \E |\breve{g}|^8 = 24.
\]
Therefore, by using these even-order moments of $\mathcal{CN}(0,1)$ together with the independence between $\va_m$ and $\breve{\va}_m$, we can compute \cref{eq:init_var_part1,eq:init_var_part2,eq:init_var_part3a,eq:init_var_part3b} as follows:
\begin{align*}
\E |\vb_m^* \vv_\sharp|^4 |\va_m^* \vu_\sharp|^4 \mP_{\vu_\sharp} \va_m \va_m^* \mP_{\vu_\sharp} \va_m \va_m^* \mP_{\vu_\sharp} &= 48 \mP_{\vu_\sharp}, \\
\E |\vb_m^* \vv_\sharp|^4 |\va_m^* \vu_\sharp|^4 \mP_{\vu_\sharp} \va_m \breve{\va}_m^* \mP_{\vu_\sharp^\perp} \breve{\va}_m \va_m^* \mP_{\vu_\sharp} &= 12(d_1-1) \mP_{\vu_\sharp}, \\
\E |\vb_m^* \vv_\sharp|^4 |\va_m^* \vu_\sharp|^4 \mP_{\vu_\sharp^\perp} \breve{\va}_m \va_m^* \mP_{\vu_\sharp} \va_m \breve{\va}_m^* \mP_{\vu_\sharp^\perp} &= 12 \mP_{\vu_\sharp^\perp} \\
\E |\vb_m^* \vv_\sharp|^4 |\va_m^* \vu_\sharp|^4 \mP_{\vu_\sharp^\perp} \breve{\va}_m \breve{\va}_m^* \mP_{\vu_\sharp^\perp} \breve{\va}_m \breve{\va}_m^* \mP_{\vu_\sharp^\perp} &= 4(d_1+1) \mP_{\vu_\sharp^\perp}.
\end{align*}
Furthermore, each of the remaining summands \cref{eq:init_var_part4,eq:init_var_part5,eq:init_var_part6,eq:init_var_part7} vanishies since it has a factor given as a central Gaussian moments of an odd order.

Applying the above results to \eqref{eq:init_var_part1a} provides
\[
\E \mZ_m^* \mZ_m = (12d_1+36) \mP_{\vu_\sharp} + (4d_1+16) \mP_{\vu_\sharp^\perp}.
\]
Then, by the definition of $\mY_m$, we have
\begin{align*}
\E \mY_m^* \mY_m
= \E \mZ_m^* \mZ_m - (\E \mZ_m)^* (\E \mZ_m)
= (12d_1+32) \mP_{\vu_\sharp} + (4d_1+15) \mP_{\vu_\sharp^\perp}.
\end{align*}
Therefore, for $d_1 \geq 3$, we have
\begin{equation}
\label{eq:init_ncr_sigma}
\Big\| \sum_{m=1}^M \E \mY_m \mY_m^* \Big\|^{1/2}
\vee
\Big\| \sum_{m=1}^M \E \mY_m^* \mY_m \Big\|^{1/2}
\leq C_1 \sqrt{M d_1}.
\end{equation}

Next we compute the $p$th moment of the spectral norm.  The $p$th moment is considered as the norm in $L_p$.  Then by the triangle inequality in $L_p$ we obtain
\begin{equation}
\label{eq:msn_eq1}
\left(\E \norm{\mY_m}^p\right)^{1/p}
\leq \left(\E \norm{\mZ_m}^p\right)^{1/p} + \norm{\E \mZ_m}
\leq \left(\E \norm{\mZ_m}^p\right)^{1/p} + 2.
\end{equation}
Again by the triangle inequality we obtain
\begin{equation}
\label{eq:msn_eq2}
\begin{aligned}
(\E \norm{\mZ_m}^p)^{1/p}
& = \left[\E \left(
|\vb_m^* \vv_\sharp|^2 |\va_m^* \vu_\sharp|^2 \norm{\mP_{\vu_\sharp}\va_m}_2^2
+ |\vb_m^* \vv_\sharp|^2 |\va_m^* \vu_\sharp|^2 \norm{\mP_{\vu_\sharp^\perp}\va_m}_2^2\right)^p\right]^{1/p} \\
& \leq
\left(\E |\vb_m^* \vv_\sharp|^{2p} |\va_m^* \vu_\sharp|^{4p}\right)^{1/p}
+ \left(\E |\vb_m^* \vv_\sharp|^{2p} |\va_m^* \vu_\sharp|^{2p} \norm{\mP_{\vu_\sharp^\perp} \breve{\va}_m}_2^{2p}\right)^{1/p} \\
& \leq
\left(\E |\vb_m^* \vv_\sharp|^{2p}\right)^{1/p} \left(\E |\va_m^* \vu_\sharp|^{4p}\right)^{1/p}
+
\left(\E |\vb_m^* \vv_\sharp|^{2p}\right)^{1/p}
\left(\E |\va_m^* \vu_\sharp|^{2p}\right)^{1/p}
\left(\E \norm{\breve{\va}_m}_2^{2p}\right)^{1/p}.
\end{aligned}
\end{equation}
Since $\va_m^* \vu_\sharp \sim \mathcal{CN}(0,1)$ and $\vb_m^* \vv_\sharp \sim \mathcal{CN}(0,1)$, by Lemma~\ref{lemma:gaussianLP}, there exists a numerical constant $C_2$ such that
\[
(\E |\va_m^* \vu_\sharp|^p)^{1/p} = (\E |\vb_m^* \vv_\sharp|^p)^{1/p} \leq C_2 \sqrt{p}.
\]
Since $2 \norm{\breve{\va}_m}_2^2$ is a chi-square random variable of the degree-of-freedom $2 d_1$, we obtain
\begin{equation*}
%\label{eq:chisquare_moment}
\left(\E \norm{\breve{\va}_m}_2^{2p}\right)^{1/p} 
\leq C_3 (d_1 + p),
\quad \forall p \geq 2.
\end{equation*}
Applying these upper estimates of the moments to \eqref{eq:msn_eq2} then to \eqref{eq:msn_eq1} provides
\begin{align*}
(\E \norm{\mY_m}^p)^{1/p} & \leq C_4 (p^2 d_1 + p^3),
\end{align*}
which implies
\begin{equation}
\label{eq:init_ncr_R}
p \Big( \sum_{m=1}^M \E \norm{\mY_m}^p \Big)^{1/p}
\leq C_4 M^{1/p} (p^3 d_1 + p^4).
\end{equation}

By applying \eqref{eq:init_ncr_sigma} and \eqref{eq:init_ncr_R} to Theorem~\ref{thm:ncrosenthal}, we obtain
\begin{equation}
\label{eq:init_moment_eq1}
\begin{aligned}
\Big( \E \Big\| \frac{1}{M} \sum_{m=1}^M \mY_m \Big\|^p \Big)^{1/p}
%\\ & \quad
\leq C_5
\Bigg[
\sqrt{\frac{p d_1}{M}}
+ \frac{M^{1/p} (p^3 d_1 + p^4)}{M}
\Bigg]
\end{aligned}
\end{equation}
for all $p \geq 2$ and $d_1 \geq 3$.

Finally, similar to \cite[Proposition~7.11]{foucart2013mathematical}, we derive a tail bound from moment bounds.  It follows from the Markov inequality that
\begin{equation}
\label{eq:init_moment_eq2}
\mathbb{P}
\Big(
\Big\| \frac{1}{M} \sum_{m=1}^M \mY_m \Big\|
> \frac{\delta}{8\sqrt{2}}
\Big)
\leq
\Big(\frac{8\sqrt{2}}{\delta}\Big)^p
\E \Big\| \frac{1}{M} \sum_{m=1}^M \mY_m \Big\|^p.
\end{equation}
By plugging in \eqref{eq:init_moment_eq1} to \eqref{eq:init_moment_eq2}, it follows that \eqref{eq:small_perturb_rank1_part1a} holds with probability $\nu$ provided that
\begin{equation*}
%\label{eq:init_moment_eq3}
C_6
\Bigg[
\sqrt{\frac{p d_1}{M}}
+ \frac{M^{1/p} (p^3 d_1 + p^4)}{M}
\Bigg]
\leq \delta \nu^{1/p}.
\end{equation*}
Then we set $p = \log (M/\nu)$ so that \eqref{eq:sampl_comp_init} implies that \eqref{eq:sampl_comp_init} holds with probability $1-\nu/2$.  Therefore, the probability for violating \eqref{eq:small_perturb_rank1_part1a} becomes $\nu = M^{-\alpha}$.  This completes the proof.

\section{Proof of Lemma~\ref{lemma:init_iidG}}
\label{sec:proof:lemma:init_iidG}
%\begin{proof}[Proof of Lemma~\ref{lemma:init_iidG}]
To simplify notation, let
\begin{align*}
\mZ_m :=
\langle \mPhi_m, \mX_\sharp \rangle^2
\mPhi_m \hat{\mV} \hat{\mV}^\transpose \mPhi_m^\transpose,
\quad m=1,\dots,M.
\end{align*}
Then $\mUpsilon$ is written as
\begin{equation}
\label{eq:exp_Ups2}
\mUpsilon = \frac{1}{M} \sum_{m=1}^M (
\mZ_m + \underbrace{\xi_m \mPhi_m \hat{\mV} \hat{\mV}^\transpose \mPhi_m^\transpose}_{\text{($\flat$)}}
).
\end{equation}

We derive the expectation of $\mUpsilon$ in the following steps:  First the expectation of the noise part ($\flat$) in \eqref{eq:exp_Ups2} is computed as
\begin{equation}
\label{eq:exp_Ups_noise}
\E \xi_m \mPhi_m \hat{\mV} \hat{\mV}^\transpose \mPhi_m^\transpose
= \xi_m \mathrm{tr}(\hat{\mV} \hat{\mV}^\transpose) \mId_{d_1}
= r \xi_m \mId_{d_1}.
\end{equation}
Next we compute $\E \mZ_m$ by using Lemma~\ref{lemma:exp_gauss_quad}.  Let $\vx_\sharp = \mathrm{vec}(\mX_\sharp)$ and $\vphi_m = \mathrm{vec}(\mPhi_m)$ for $m=1,\dots,M$.  Then $\mZ_m$ is rewritten as
\begin{align*}
\mZ_m
{} & =
(\mathrm{tr} \otimes \mId_{d_1})
\Big[
(\hat{\mV}^\transpose \otimes \mId_{d_1})
\langle \vphi_m, \vx_\sharp \rangle^2 \vphi_m \vphi_m^\transpose
(\hat{\mV} \otimes \mId_{d_1})
\Big].
\end{align*}
Since the partial trace operator is linear, the expectation of $\mZ_m$ is written as
\begin{equation}
\begin{aligned}
\mathbb{E} \mZ_m
{} & =
(\mathrm{tr} \otimes \mId_{d_1})
\Big[
(\hat{\mV}^\transpose \otimes \mId_{d_1})
\mathbb{E} \langle \vphi_m, \vx_\sharp \rangle^2 \vphi_m \vphi_m^\transpose
(\hat{\mV} \otimes \mId_{d_1})
\Big] \\
{} & = (\mathrm{tr} \otimes \mId_{d_1})
\Big[
(\hat{\mV}^\transpose \otimes \mId_{d_1})
(2 \vx_\sharp \vx_\sharp^\transpose + \norm{\mX_\sharp}_\mathrm{F}^2 \mId_{d_1 d_2})
(\hat{\mV} \otimes \mId_{d_1})
\Big] \\
{} & = 2 \mX_\sharp \hat{\mV} \hat{\mV}^\transpose \mX_\sharp^\transpose
+ r \norm{\mX_\sharp}_\mathrm{F}^2 \mId_{d_1},
\end{aligned}
\label{eq:exp_Z_m}
\end{equation}
where the second identity follows from Lemma~\ref{lemma:exp_gauss_quad}.
Then by combining \eqref{eq:exp_Ups_noise} and \eqref{eq:exp_Z_m}, the expectation of $\mUpsilon$ is written as
\begin{equation}
\label{eq:exp_mUps}
\E \mUpsilon
= 2 \mX_\sharp \hat{\mV} \hat{\mV}^\transpose \mX_\sharp^\transpose
+ \Big(r \norm{\mX_\sharp}_\mathrm{F}^2 + \frac{r}{M} \sum_{m=1}^M \xi_m \Big) \mId_{d_1}.
\end{equation}

It follows from \eqref{eq:sin_estV} that $\mX_\sharp \hat{\mV} \hat{\mV}^\transpose \mX_\sharp^\transpose$ in the right-hand side of \eqref{eq:exp_mUps} has rank-$r$ and its invariant space coincides with that of $\mX_\sharp \mX_\sharp^\transpose = \mU_\sharp \mSigma_\sharp^2 \mU_\sharp^\transpose$.  The inclusion of the former subspace to the latter is obvious from the construction.  Furthermore, the rank of $\mX_\sharp \hat{\mV} \hat{\mV}^\transpose \mX_\sharp^\transpose$ is at most $r$.  Indeed, the $r$th largest singular value of $\mX_\sharp \hat{\mV} \hat{\mV}^\transpose \mX_\sharp^\transpose$ satisfies
\begin{align*}
& \sigma_r(\mX_\sharp \hat{\mV} \hat{\mV}^\transpose \mX_\sharp^\transpose)
\geq \sigma_r(\mX_\sharp)^2 \sigma_r(\hat{\mV} \hat{\mV}^\transpose \mV_\sharp \mV_\sharp^\transpose) \\
& \geq \sigma_r(\mX_\sharp)^2 \left(\sigma_r(\mV_\sharp \mV_\sharp^\transpose) - \norm{(\mId_{d_2} - \hat{\mV} \hat{\mV}^\transpose) \mV_\sharp \mV_\sharp^\transpose}\right) 
\geq (1-\delta_{\mathrm{in}}) \sigma_r(\mX_\sharp)^2,
\end{align*}
where the last step follows from \eqref{eq:sin_estV}.  Therefore, we deduce that $\mX_\sharp \hat{\mV} \hat{\mV}^\transpose \mX_\sharp^\transpose$ and $\mX_\sharp \mX_\sharp^\transpose$ have the same invariant subspace.

Recall that the columns of $\mU_0$ are the eigenvectors of $\mUpsilon$ corresponding to the $r$-largest eigenvalues.  Furthermore the subspace spanned by the top $r$ eigenvectors of $\E \mUpsilon$, is the same to the columnspace of $\mU_\sharp$. Therefore, the Davis-Kahan theorem (Theorem~\ref{thm:sintheta}) provides an upper bound for the estimation error measured by the principal angle between subspaces (the left-hand side of \eqref{eq:sin_estU}).  To this end, we apply Theorem~\ref{thm:sintheta} to $\mA = \E \mUpsilon$ and $\mDelta = \mUpsilon - \E \mUpsilon$ as shown below.

Since the spectral gap in $\mA$ satisfies
\[
\lambda_r(\mA) - \lambda_{r+1}(\mA)
= \lambda_r(\E \mUpsilon) - \lambda_{r+1}(\E \mUpsilon)
= \lambda_r(2 \mX_\sharp \hat{\mV} \hat{\mV}^\transpose \mX_\sharp^\transpose)
\geq 2 (1-\delta_{\mathrm{in}}) [\sigma_r(\mX_\sharp)]^2,
\]
the error bound in \eqref{eq:sin_estU} is obtained by Theorem~\ref{thm:sintheta} provided that
\begin{equation}
\label{eq:ub_spn_mDelta_Gaussian}
\norm{\mDelta} = \norm{\mUpsilon - \E \mUpsilon} \leq \frac{(1-\delta_{\mathrm{in}}) \delta_{\mathrm{out}} \sigma_r(\mX_\sharp)^2}{2}.
\end{equation}
By the triangle inequality, we obtain a sufficient condition for \eqref{eq:ub_spn_mDelta_Gaussian} given by
\begin{equation}
\label{eq:ub_spn_mDelta_Gaussian_part1}
\Big\| \frac{1}{M} \sum_{m=1}^M (
\mZ_m - \E \mZ_m) \Big\|
\leq \frac{(1-\delta_{\mathrm{in}}) \delta_{\mathrm{out}} \sigma_r(\mX_\sharp)^2}{4}
\end{equation}
and
\begin{equation}
\label{eq:ub_spn_mDelta_Gaussian_part2}
\Big\|
\frac{1}{M} \sum_{m=1}^M
\xi_m (\mPhi_m \hat{\mV} \hat{\mV}^\transpose \mPhi_m^\transpose - r \mId_{d_1})
\Big\|
\leq \frac{(1-\delta_{\mathrm{in}}) \delta_{\mathrm{out}} \sigma_r(\mX_\sharp)^2}{4}.
\end{equation}

In the remainder, we show that \eqref{eq:ub_spn_mDelta_Gaussian_part1} and \eqref{eq:ub_spn_mDelta_Gaussian_part1} hold with high probability when the conditions in \eqref{eq:sampl_comp_init_iidG} and \eqref{eq:snr_cond_init2_iidG} are satisfied.  First, by Lemma~\ref{lemma:wgaussouterproduct}, it follows from \eqref{eq:snr_cond_init2_iidG} that \eqref{eq:ub_spn_mDelta_Gaussian_part2} holds with probability $1-M^{-\alpha}/2$. Then it remains to show that \eqref{eq:ub_spn_mDelta_Gaussian_part1} holds with probability $1 - M^{-\alpha}/2$ when \eqref{eq:sampl_comp_init_iidG} is satisfied.
By the Markov inequality,
\begin{align*}
& \mathbb{P}\Big(
\Big\| \frac{1}{M} \sum_{m=1}^M (
\mZ_m - \E \mZ_m) \Big\|
> \frac{(1-\delta_{\mathrm{in}}) \delta_{\mathrm{out}} \sigma_r(\mX_\sharp)^2}{4}
\Big) 
\\
& \leq
\Big(\frac{4}{(1-\delta_{\mathrm{in}}) \delta_{\mathrm{out}} \sigma_r(\mX_\sharp)^2}\Big)^p
\cdot
\mathbb{E} \Big\| \frac{1}{M} \sum_{m=1}^M (\mZ_m - \mathbb{E} \mZ_m) \Big\|^p
\end{align*}
for any $p > 0$.
Therefore, \eqref{eq:ub_spn_mDelta_Gaussian_part1} holds with probability $1-M^{-\alpha}/2$ if
\begin{equation}
\label{eq:ub_spn_mDelta_Gaussian_part1a}
\underbrace{
\Big( \mathbb{E} \Big\| \frac{1}{M} \sum_{m=1}^M (\mZ_m - \mathbb{E} \mZ_m) \Big\|^p \Big)^{1/p}
}_{\text{($\ddag$)}}
\leq \frac{(1-\delta_{\mathrm{in}}) \delta_{\mathrm{out}} \sigma_r(\mX_\sharp)^2 M^{-\alpha/p}}{4}.
\end{equation}

To get an upper estimate of ($\ddag$) in \eqref{eq:ub_spn_mDelta_Gaussian_part1a}, we apply the noncommutative Rosenthal inequality (Theorem~\ref{thm:ncrosenthal}) to $\mY_m = \mZ_m - \mathbb{E} \mZ_m$ for $m=1,\dots,M$.  The first step is to compute the expectation of $\mY_m^2$ as follows:  Let $\mQ_1,\mQ_2,\mQ_3,\mQ_4 \in \mathbb{R}^{d_1 \times d_2}$.  Note that each entry of $\mQ_1 \mQ_2^\transpose \mQ_3 \mQ_4^\transpose$ is given as a linear combination of the entries of $\mathrm{vec}(\mQ_1) \otimes \mathrm{vec}(\mQ_2) \otimes \mathrm{vec}(\mQ_3) \otimes \mathrm{vec}(\mQ_4)$.  Therefore, there exists a linear map $\mathcal{R}: \mathbb{R}^{(d_1 d_2)^4} \to \mathbb{R}^{d_1 \times d_1}$ that satisfies
\[
\mathcal{R}[\mathrm{vec}(\mQ_1) \otimes \mathrm{vec}(\mQ_2) \otimes \mathrm{vec}(\mQ_3) \otimes \mathrm{vec}(\mQ_4)]
= \mQ_1 \mQ_2^\transpose \mQ_3 \mQ_4^\transpose.
\]
We also define
\[
\mT_m := \langle \mPhi_m, \mX_\sharp \rangle^4
[\mathrm{vec}(\mPhi_m \hat{\mV}) \otimes \mathrm{vec}(\mPhi_m \hat{\mV}) \otimes \mathrm{vec}(\mPhi_m \hat{\mV}) \otimes \mathrm{vec}(\mPhi_m \hat{\mV})].
\]
Then $\mZ_m^2$ is written as $\mZ_m^2 = \mathcal{R}(\mT_m)$.
Since $\mathrm{vec}(\mPhi_m \hat{\mV}) = (\hat{\mV}^\transpose \otimes \mId_{d_1}) \mathrm{vec}(\mPhi_m)$, it follows that $\mathbb{E} \mT_m$ is written as
\begin{align*}
\mathbb{E} \mT_m
&=
[ (\hat{\mV}^\transpose \otimes \mId_{d_1}) \otimes (\hat{\mV}^\transpose \otimes \mId_{d_1}) \otimes (\hat{\mV}^\transpose \otimes \mId_{d_1}) \otimes (\hat{\mV}^\transpose \otimes \mId_{d_1}) ] 
%\\
%&\quad \cdot 
\mathbb{E} (\vphi_m^\transpose \mathrm{vec}(\mX_\sharp))^4
(\vphi_m \otimes \vphi_m \otimes \vphi_m \otimes \vphi_m) \\
& = 24 \, [ \mathrm{vec}(\mX_\sharp \hat{\mV}) \otimes \mathrm{vec}(\mX_\sharp \hat{\mV}) \otimes \mathrm{vec}(\mX_\sharp \hat{\mV}) \otimes \mathrm{vec}(\mX_\sharp \hat{\mV}) ] \\
& \quad + 12 \sum_{k_1=1}^d \sum_{k_2=1}^d \Big[
\mathrm{vec}(\mX_\sharp \hat{\mV}) \otimes \mathrm{vec}(\mX_\sharp \hat{\mV}) \otimes \mathrm{vec}(\ve_{k_1} \tilde{\ve}_{k_2}^\transpose \hat{\mV}) \otimes \mathrm{vec}(\ve_{k_1} \tilde{\ve}_{k_2}^\transpose \hat{\mV}) \\
& \qquad\qquad\qquad + \mathrm{vec}(\mX_\sharp \hat{\mV}) \otimes \mathrm{vec}(\ve_{k_1} \tilde{\ve}_{k_2}^\transpose \hat{\mV}) \otimes \mathrm{vec}(\mX_\sharp \hat{\mV}) \otimes \mathrm{vec}(\ve_{k_1} \tilde{\ve}_{k_2}^\transpose \hat{\mV}) \\
& \qquad\qquad\qquad + \mathrm{vec}(\mX_\sharp \hat{\mV}) \otimes \mathrm{vec}(\ve_{k_1} \tilde{\ve}_{k_2}^\transpose \hat{\mV}) \otimes \mathrm{vec}(\ve_{k_1} \tilde{\ve}_{k_2}^\transpose \hat{\mV}) \otimes \mathrm{vec}(\mX_\sharp \hat{\mV}) \\
& \qquad\qquad\qquad + \mathrm{vec}(\ve_{k_1} \tilde{\ve}_{k_2}^\transpose \hat{\mV}) \otimes \mathrm{vec}(\mX_\sharp \hat{\mV}) \otimes \mathrm{vec}(\mX_\sharp \hat{\mV}) \otimes \mathrm{vec}(\ve_{k_1} \tilde{\ve}_{k_2}^\transpose \hat{\mV}) \\
& \qquad\qquad\qquad + \mathrm{vec}(\ve_{k_1} \tilde{\ve}_{k_2}^\transpose \hat{\mV}) \otimes \mathrm{vec}(\mX_\sharp \hat{\mV}) \otimes \mathrm{vec}(\ve_{k_1} \tilde{\ve}_{k_2}^\transpose \hat{\mV}) \otimes \mathrm{vec}(\mX_\sharp \hat{\mV}) \\
& \qquad\qquad\qquad + \mathrm{vec}(\ve_{k_1} \tilde{\ve}_{k_2}^\transpose \hat{\mV}) \otimes \mathrm{vec}(\ve_{k_1} \tilde{\ve}_{k_2}^\transpose \hat{\mV}) \otimes \mathrm{vec}(\mX_\sharp \hat{\mV}) \otimes \mathrm{vec}(\mX_\sharp \hat{\mV}) \Big] \\
& \quad + 3 \sum_{j_1,k_1=1}^{d_1} \sum_{j_2,k_2=1}^{d_2} \Big[
\mathrm{vec}(\ve_{j_1} \tilde{\ve}_{j_2}^\transpose \hat{\mV}) \otimes \mathrm{vec}(\ve_{j_1} \tilde{\ve}_{j_2}^\transpose \hat{\mV}) \otimes \mathrm{vec}(\ve_{k_1} \tilde{\ve}_{k_2}^\transpose \hat{\mV}) \otimes \mathrm{vec}(\ve_{k_1} \tilde{\ve}_{k_2}^\transpose \hat{\mV}) \\
& \qquad\qquad\qquad + \mathrm{vec}(\ve_{j_1} \tilde{\ve}_{j_2}^\transpose \hat{\mV}) \otimes \mathrm{vec}(\ve_{k_1} \tilde{\ve}_{k_2}^\transpose \hat{\mV}) \otimes \mathrm{vec}(\ve_{j_1} \tilde{\ve}_{j_2}^\transpose \hat{\mV}) \otimes \mathrm{vec}(\ve_{k_1} \tilde{\ve}_{k_2}^\transpose \hat{\mV}) \\
& \qquad\qquad\qquad + \mathrm{vec}(\ve_{j_1} \tilde{\ve}_{j_2}^\transpose \hat{\mV}) \otimes \mathrm{vec}(\ve_{k_1} \tilde{\ve}_{k_2}^\transpose \hat{\mV}) \otimes \mathrm{vec}(\ve_{k_1} \tilde{\ve}_{k_2}^\transpose \hat{\mV}) \otimes \mathrm{vec}(\ve_{j_1} \tilde{\ve}_{j_2}^\transpose \hat{\mV}) \Big],
\end{align*}
where Lemma~\ref{lemma:exp_gauss_octa} is used to compute $\mathbb{E} \mT_m$ in the second step.
Also by the linearity of the map $\mathcal{R}$, it follows that
\begin{align*}
\mathbb{E} \mZ_m^2
&= \mathcal{R}(\mathbb{E} \mT_m) \\
& = 24 \mX_\sharp \hat{\mV} \hat{\mV}^\transpose \mX_\sharp^\transpose \mX_\sharp \hat{\mV} \hat{\mV}^\transpose \mX_\sharp^\transpose \\
& \quad + 12 \norm{\mX_\sharp}_\mathrm{F}^2 \sum_{l_1=1}^d \sum_{l_2=1}^d \Big[
\mX_\sharp \hat{\mV}
\hat{\mV}^\transpose \mX_\sharp^\transpose
\ve_{l_1} \tilde{\ve}_{l_2}^\transpose \hat{\mV}
\hat{\mV}^\transpose \tilde{\ve}_{l_2} \ve_{l_1}^\transpose %\\
%& \qquad\qquad\qquad
+
\mX_\sharp \hat{\mV}
\hat{\mV}^\transpose \tilde{\ve}_{l_2} \ve_{l_1}^\transpose
\mX_\sharp \hat{\mV}
\hat{\mV}^\transpose \tilde{\ve}_{l_2} \ve_{l_1}^\transpose \\
& \qquad\qquad\qquad +
\mX_\sharp \hat{\mV}
\hat{\mV}^\transpose \tilde{\ve}_{l_2} \ve_{l_1}^\transpose
\ve_{l_1} \tilde{\ve}_{l_2}^\transpose \hat{\mV}
\hat{\mV}^\transpose \mX_\sharp^\transpose %\\
%& \qquad\qquad\qquad
+
\ve_{l_1} \tilde{\ve}_{l_2}^\transpose \hat{\mV}
\hat{\mV}^\transpose \mX_\sharp^\transpose
\mX_\sharp \hat{\mV}
\hat{\mV}^\transpose \tilde{\ve}_{l_2} \ve_{l_1}^\transpose \\
& \qquad\qquad\qquad +
\ve_{l_1} \tilde{\ve}_{l_2}^\transpose \hat{\mV}
\hat{\mV}^\transpose \mX_\sharp^\transpose
\ve_{l_1} \tilde{\ve}_{l_2}^\transpose \hat{\mV}
\hat{\mV}^\transpose \mX_\sharp^\transpose %\\
%& \qquad\qquad\qquad
+
\ve_{l_1} \tilde{\ve}_{l_2}^\transpose \hat{\mV}
\hat{\mV}^\transpose \tilde{\ve}_{l_2} \ve_{l_1}^\transpose
\mX_\sharp \hat{\mV}
\hat{\mV}^\transpose \mX_\sharp^\transpose
\Big] \\
& \quad + 3 \norm{\mX_\sharp}_\mathrm{F}^4 \sum_{j_1,k_1=1}^{d_1} \sum_{j_2,k_2=1}^{d_2} \Big[
\ve_{j_1} \tilde{\ve}_{j_2}^\transpose \hat{\mV}
\hat{\mV}^\transpose \tilde{\ve}_{j_2} \ve_{j_1}^\transpose
\ve_{k_1} \tilde{\ve}_{k_2}^\transpose \hat{\mV}
\hat{\mV}^\transpose \tilde{\ve}_{k_2} \ve_{k_1}^\transpose \\
& \qquad\qquad\qquad
+
\ve_{j_1} \tilde{\ve}_{j_2}^\transpose \hat{\mV}
\hat{\mV}^\transpose \tilde{\ve}_{k_2} \ve_{k_1}^\transpose
\ve_{j_1} \tilde{\ve}_{j_2}^\transpose \hat{\mV}
\hat{\mV}^\transpose \tilde{\ve}_{k_2} \ve_{k_1}^\transpose %\\
%& \qquad\qquad\qquad
+
\ve_{j_1} \tilde{\ve}_{j_2}^\transpose \hat{\mV}
\hat{\mV}^\transpose \tilde{\ve}_{k_2} \ve_{k_1}^\transpose
\ve_{k_1} \tilde{\ve}_{k_2}^\transpose \hat{\mV}
\hat{\mV}^\transpose \tilde{\ve}_{j_2} \ve_{j_1}^\transpose \Big].
\end{align*}
After direct calculation, the above expression for $\mathbb{E} \mZ_m^2 $ simplifies to
\begin{equation}
\begin{aligned}
\mathbb{E} \mZ_m^2
{} & = 24 \mX_\sharp \hat{\mV} \hat{\mV}^\transpose \mX_\sharp^\transpose \mX_\sharp \hat{\mV} \hat{\mV}^\transpose \mX_\sharp^\transpose \\
{} & + 12 (2r+d_1+2) \norm{\mX_\sharp}_\mathrm{F}^2 \mX_\sharp \hat{\mV} \hat{\mV}^\transpose \mX_\sharp^\transpose + 12 \norm{\mX_\sharp}_\mathrm{F}^2 \norm{\mX_\sharp \hat{\mV}}_\mathrm{F}^2 \mId_{d_1} \\
{} & + 3 \norm{\mX_\sharp}_\mathrm{F}^4 r(r+d_1+1) \mId_{d_1}.
\end{aligned}
\label{eq:exp_Z_m_sq}
\end{equation}
Then, by combining \eqref{eq:exp_Z_m} and \eqref{eq:exp_Z_m_sq}, we obtain
\begin{align*}
\mathbb{E} \mY_m^2
{} & = \mathbb{E} \mZ_m^2 - (\mathbb{E} \mZ_m)^2 \\
{} & = 20 \mX_\sharp \hat{\mV} \hat{\mV}^\transpose \mX_\sharp^\transpose \mX_\sharp \hat{\mV} \hat{\mV}^\transpose \mX_\sharp^\transpose 
%{} & 
+ 4 (5 r + 3 d_1 + 6) \norm{\mX_\sharp}_\mathrm{F}^2 \mX_\sharp \hat{\mV} \hat{\mV}^\transpose \mX_\sharp^\transpose \\
{} & + \left( 12 \norm{\mX_\sharp}_\mathrm{F}^2 \norm{\mX_\sharp \hat{\mV}}_\mathrm{F}^2 + \norm{\mX_\sharp}_\mathrm{F}^4 r(2r + 3d_1 + 3) \right)  \mId_{d_1}.
\end{align*}
Therefore, the spectral norm of $\mathbb{E} \mY_m^2$ is upper-bounded by
\begin{align*}
\norm{\mathbb{E} \mY_m^2}
{} & \leq 20 \norm{\mX_\sharp}^4
+ 4 (5r+3d_1+6) \norm{\mX_\sharp}_\mathrm{F}^2 \norm{\mX_\sharp}^2 
%\\
%{} & 
+ 12 \norm{\mX_\sharp}_\mathrm{F}^2 \norm{\mX_\sharp \hat{\mV}}_\mathrm{F}^2 
+ r (2r+3d_1+3) \norm{\mX_\sharp}_\mathrm{F}^4.
\end{align*}
Collecting the results for $m=1,\dots,M$ gives
\begin{equation}
\label{eq:ncRosenthal_part1_gaussian}
\Big\| \sum_{m=1}^M \mathbb{E} \mY_m^2 \Big\|^{1/2}
\leq C r^{3/2} \sqrt{M d_1} \norm{\mX_\sharp}^2.
\end{equation}
Moreover, by applying the triangle inequality in $L_p$ twice to \eqref{eq:exp_Z_m}, we obtain
\begin{equation}
\label{eq:exp_Y_m}
\begin{aligned}
(\E \norm{\mY_m}^p)^{1/p}
{} & \leq
\underbrace{
\Big[
\E \Big(
\langle \mPhi_m, \mX_\sharp \rangle^2
\norm{\mPhi_m \hat{\mV} \hat{\mV}^\transpose \mPhi_m^\transpose
- r \mId_{d_1}}
\Big)^p
\Big]^{1/p}
}_{\text{($\natural$)}} \\
{} & +
\underbrace{
r \Big[ \E \Big( \langle \mPhi_m, \mX_\sharp \rangle^2 - \norm{\mX_\sharp}_\mathrm{F}^2 \Big)^p \Big]^{1/p}
}_{\text{($\natural\natural$)}}
+ \underbrace{
2 \norm{ \mX_\sharp \hat{\mV} \hat{\mV}^\transpose \mX_\sharp^\transpose }
}_{\text{($\natural\natural\natural$)}}.
\end{aligned}
\end{equation}
By the Cauchy-Schwarz inequality in $L_2$, the first term ($\natural$) on the right-hand side of \eqref{eq:exp_Y_m} is upper-bounded by
\begin{align*}
\text{($\natural$)}
\leq \Big(
\E \langle \mPhi_m, \mX_\sharp \rangle^{4p}
\Big)^{1/2p}
\cdot \Big(
\E \norm{\mPhi_m \hat{\mV} \hat{\mV}^\transpose \mPhi_m^\transpose
- r \mId_{d_1}}^{2p}
\Big)^{1/2p}.
\end{align*}
Since $\langle \mX_\sharp, \mPhi_m \rangle \sim \mathcal{N}(0, \norm{\mX_\sharp}_\mathrm{F}^2)$, by Lemma~\ref{lemma:gaussianLP}, we have
\[
\Big(
\mathbb{E} \langle \mPhi_m, \mX_\sharp \rangle^{4p}
\Big)^{1/2p}
\leq C p \norm{\mX_\sharp}_\mathrm{F}^2.
\]
Then it follows from $\mathrm{vec}(\mPhi_m \hat{\mV}) = (\hat{\mV}^\transpose \otimes \mId_{d_1}) \vphi_m$ that
\[
\mathbb{E} \mathrm{vec}(\mPhi_m \hat{\mV}) \mathrm{vec}(\mPhi_m \hat{\mV})^\transpose
= \mathbb{E} (\hat{\mV}^\transpose \otimes \mId_{d_1}) \vphi_m \vphi_m^\transpose (\hat{\mV} \otimes \mId_{d_1}) \vphi_m
= \hat{\mV}^\transpose \hat{\mV} \otimes \mId_{d_1} = \mId_{d_2 d_1}, 
\]
which implies that $\mPhi_1 \hat{\mV}, \dots, \mPhi_M \hat{\mV} \in \mathbb{R}^{d_1 \times r}$ are independent copies of a standard i.i.d. Gaussian matrix.  Thus Lemma~\ref{lemma:wgaussouterproduct} implies
\begin{align*}
\Big( \mathbb{E} \|\mPhi_m \hat{\mV} \hat{\mV}^\transpose \mPhi_m^\transpose
- r \mId_{d_1}\|^{2p} \Big)^{1/2p}
\leq C
\left[
\sqrt{rpd_1}
+ r^{1/2p} p \left(d_1+p\right)
\right].
\end{align*}
Then, by the triangle inequality in $L_p$ and Lemma~\ref{lemma:gaussianLP}, ($\natural\natural$) is upper-bounded as 
\[
\text{($\natural\natural$)}
\leq
r \Big( \E \langle \mPhi_m, \mX_\sharp \rangle^{2p} \Big)^{1/p} + r \norm{\mX_\sharp}_\mathrm{F}^2
\leq C' r p \norm{\mX_\sharp}_\mathrm{F}^2.
\]
The last term is trivially upper-bounded by $\text{($\natural\natural\natural$)} \leq \norm{\mX_\sharp}_\mathrm{F}^2$.

By collecting the above results, we obtain that the $L_p$-norm of $\norm{\mY_m}$ is upper-bounded by
\begin{equation}
\label{eq:ncRosenthal_part2_gaussian}
(\E \norm{\mY_m}^p)^{1/p}
\leq
C_1 \norm{\mX_\sharp}_\mathrm{F}^2 \, p
\left( r + \sqrt{rpd_1}+r^{1/2p} p \left(d_1+p\right) \right).
\end{equation}
Then, by applying \eqref{eq:ncRosenthal_part1_gaussian} and \eqref{eq:ncRosenthal_part2_gaussian} to Theorem~\ref{thm:ncrosenthal}, we obtain that ($\ddag$) in \eqref{eq:ub_spn_mDelta_Gaussian_part1a} is upper-bounded by
\begin{equation*}
%\label{eq:ncRosenthal_gaussian}
\begin{aligned}
{} & \Big(\mathbb{E} \Big\| \frac{1}{M} \sum_{m=1}^M (\mZ_m - \mathbb{E} \mZ_m) \Big\|^p\Big)^{1/p} \\
{} & \quad
\leq
C_3 \norm{\mX_\sharp}^2 \Big(
r^{3/2} M^{-1/2} \sqrt{pd_1}
+ M^{1/p-1} 
p^2 r \left( r + \sqrt{rpd_1}+r^{1/2p} p \left(d_1+p\right) \right)  \Big).
\end{aligned}
\end{equation*}
Finally, we choose $p = \log(M/M^{-\alpha}) = (\alpha+1)\log M$.  Then \eqref{eq:sampl_comp_init_iidG} implies \eqref{eq:ub_spn_mDelta_Gaussian_part1a}.
This completes the proof.
%\end{proof}

\section{Proof of Lemma~\ref{lemma:imagpart}}
\label{sec:proof:lemma:imagpart}

Since $\hat{\mX}$ is a minimizer to \eqref{eq:convex_prog}, it satisfies
\begin{equation}
\label{eq:phaseXhat}
\mathrm{Im}\,\langle \mX_0, \hat{\mX} \rangle = 0.
\end{equation}
Then by \eqref{eq:phaseXhat} and \eqref{eq:phasea0b0} together with the fact that $\mathrm{rank}(\mX_\sharp) = 1$, we have
\[
\mathrm{Im}\,\langle \mX_0, \mH \rangle = 0,
\]
where $\mH = \hat{\mX} - \mX_\sharp$. Let $\mathcal{P}_{\mX_\sharp}$ denote the orthogonal projection onto $\mathbb{C} \mX_\sharp$,  that is
\[
\mathcal{P}_{\mX_\sharp}: \mM \mapsto \frac{\mX_\sharp \langle \mX_\sharp, \mM\rangle}{\norm{\mX_\sharp}_\mathrm{F}^2}.
\]
Then it follows that
\begin{align*}
0 = |\mathrm{Im}\,\langle \mX_0, \mH \rangle|
\geq
|\mathrm{Im}\,\langle \mathcal{P}_{\mX_\sharp}(\mX_0), \mH \rangle|
- |\mathrm{Im}\,\langle \mX_0 - \mathcal{P}_{\mX_\sharp}(\mX_0), \mH \rangle|,
\end{align*}
which is rearranged as
\begin{equation}
\label{eq:ubIm1}
|\mathrm{Im}\,\langle \mathcal{P}_{\mX_\sharp}(\mX_0), \mH \rangle|
\leq |\mathrm{Im}\,\langle \mX_0 - \mathcal{P}_{\mX_\sharp}(\mX_0), \mH \rangle|.
\end{equation}
By \eqref{eq:phasea0b0}, the left-hand side of \eqref{eq:ubIm1} is bounded from below as
\begin{align*}
& |\mathrm{Im}\,\langle \mathcal{P}_{\mX_\sharp}(\mX_0), \mH \rangle|
= \frac{\left|\mathrm{Im} (\langle \mX_0, \mX_\sharp\rangle \langle \mX_\sharp, \mH\rangle) \right|}{\norm{\mX_\sharp}_\mathrm{F}^2} \\
&= \frac{\langle \mX_0, \mX_\sharp\rangle}{\norm{\mX_\sharp}_\mathrm{F}} \cdot
\frac{\left|\mathrm{Im} \langle \mX_\sharp, \mH\rangle \right|}{\norm{\mX_\sharp}_\mathrm{F}}
\geq \sqrt{1-\delta^2} \cdot \norm{\mX_0}_\mathrm{F} \cdot \frac{\left|\mathrm{Im} \langle \mX_\sharp, \mH\rangle \right|}{\norm{\mX_\sharp}_\mathrm{F}}.
\end{align*}
Since the linear operator $\iota: \mM \mapsto \mM - \mathcal{P}_{\mX_\sharp}(\mM)$ is self-adjoint and idempotent, the right-hand side of \eqref{eq:ubIm1} is bounded from above as
\begin{align*}
& |\mathrm{Im}\,\langle \mX_0 - \mathcal{P}_{\mX_\sharp}(\mX_0), \mH \rangle|
= |\mathrm{Im}\,\langle \mX_0 - \mathcal{P}_{\mX_\sharp}(\mX_0), \mH - \mathcal{P}_{\mX_\sharp}(\mH) \rangle| \\
& \quad \leq \norm{\mX_0 - \mathcal{P}_{\mX_\sharp}(\mX_0)}_\mathrm{F} \cdot \norm{\mH - \mathcal{P}_{\mX_\sharp}(\mH)}_\mathrm{F}
\leq \delta \,\norm{\mX_0}_\mathrm{F} \cdot \norm{\mH - \mathcal{P}_{\mX_\sharp}(\mH)}_\mathrm{F}.
\end{align*}
Applying the above bounds to \eqref{eq:ubIm1} completes the proof.

\section{Proof of Lemma~\ref{lemma:lbp_iidG}}
\label{sec:proof:lemma:lbp_iidG}

The following lemma provides a tail probability of the product of two jointly Gaussian variables.

\begin{lemma}[{A variation of \cite[Lemma~5]{bahmani2017phase}}]
\label{lemma:gaussprod}
Let $g_1,g_2$ be random variables that satisfy
\[
\begin{bmatrix} g_1 \\ g_2 \end{bmatrix} \sim \mathcal{N}\left(\vzero,\begin{bmatrix} 1 & \rho \\ \rho & 1\end{bmatrix}\right).
\]
Then for all $t > 0$
\begin{equation}
\label{eq:tail_gauss_prod}
\mathbb{P}(g_1 g_2 > t) 
\geq \frac{2}{\pi} \cos^{-1}\left(\frac{\sqrt{3-\rho}}{2}\right) \exp\left(-\frac{2t}{1+\rho}\right).
\end{equation}
\end{lemma}

\begin{proof}[Proof of Lemma~\ref{lemma:gaussprod}]
Let $w_1$ and $w_2$ be independent copies of a standard normal random variable following $\mathcal{N}(0,1)$. Then $g_1$ and $g_2$ are written as
\begin{align*}
g_1 = \sqrt{\frac{1+\rho}{2}} w_1 + \sqrt{\frac{1-\rho}{2}} w_2 \quad \text{and} \quad
g_2 = \sqrt{\frac{1-\rho}{2}} w_1 - \sqrt{\frac{1-\rho}{2}} w_2.
\end{align*}
With this representation, we have
\begin{align*}
\mathbb{P}(g_1 g_2 > t) 
= \mathbb{P}\left( \frac{1+\rho}{2} \, w_1^2 - \frac{1-\rho}{2} \, w_2^2 > t \right)
= \mathbb{P}\left( \frac{\rho-1}{2} + \frac{w_1^2}{w_1^2+w_2^2} > \frac{t}{w_1^2+w_2^2} \right).
\end{align*}
Since $w_1^2/(w_1^2+w_2^2)$ and $1/(w_1^2+w_2^2)$ respectively depend only on the direction and the $\ell_2$ norm of the standard normal random vector $[w_1,w_2]^\top$, they are mutually independent. Furthermore, $R = w_1^2+w_2^2$ follows the exponential distribution with mean $1/2$ and $w_1/\sqrt{w_1^2+w_2^2}$ is written as $\cos\theta$ where $\theta$ is a uniform random variable on $[0,2\pi)$. Then it follows that
\begin{align}
\mathbb{P}\left( \frac{\rho-1}{2} + \frac{w_1^2}{w_1^2+w_2^2} > \frac{t}{w_1^2+w_2^2} \right)
& \geq \mathbb{P}\left(\frac{w_1^2}{w_1^2+w_2^2} \geq \frac{3-\rho}{4} \,\text{and}\, \frac{1+\rho}{4} > \frac{t}{w_1^2+w_2^2} \right) \nonumber \\
& = \mathbb{P}\left(\frac{w_1^2}{w_1^2+w_2^2} \geq \frac{3-\rho}{4} \right) 
\mathbb{P}\left(w_1^2+w_2^2 > \frac{4t}{1+\rho} \right) \nonumber \\
& = \mathbb{P}\left(\cos^2\theta \geq \frac{3-\rho}{4} \right) 
\mathbb{P}\left( R > \frac{4t}{1+\rho} \right). \label{eq:pf_tail_gauss_prod}
\end{align} 
The lower bound in \eqref{eq:tail_gauss_prod} is obtained by computing the probabilities in \eqref{eq:pf_tail_gauss_prod}. 
\end{proof}

We apply Lemma~\ref{lemma:gaussprod} for $g_1 = \langle \mX_\sharp, \mPhi \rangle / \norm{\mX_\sharp}_\mathrm{F}$, $g_2 = \langle \mPhi, \mH \rangle / \norm{\mH}_\mathrm{F}$, and $t = \tau'$.  Since the probability in \eqref{eq:tail_gauss_prod} is a monotone increasing function in $\rho$, to get a lower bound on the tail probability, it suffices to compute a lower estimate of $\rho$.

Let $\mX_\sharp = \mU_\sharp \mSigma_\sharp \mV_\sharp^\transpose$ denote the singular value decomposition of $\mX_\sharp$.  Let $\sigma_1,\dots,\sigma_r$ denote the singular values of $\mX_\sharp$ in the non-increasing order.  Then $\norm{\mX_\sharp}_* = \sum_{k=1}^r \sigma_k$.
By the triangle inequality, we have
\begin{equation}
\label{eq:lb_rho}
\rho = \frac{\langle \mX_\sharp, \mH \rangle}{\norm{\mX_\sharp}_\mathrm{F} \norm{\mH}_\mathrm{F}}
\geq
\underbrace{
\frac{\langle \norm{\mX_\sharp}_* \mU_\sharp \mV_\sharp^\transpose, \mH \rangle}{r \norm{\mX_\sharp}_\mathrm{F} \norm{\mH}_\mathrm{F}}
}_{\text{($\flat\flat$)}}
- 
\underbrace{
\frac{|\langle r \mX_\sharp - \norm{\mX_\sharp}_* \mU_\sharp \mV_\sharp^\transpose, \mH \rangle|}{r \norm{\mX_\sharp}_\mathrm{F} \norm{\mH}_\mathrm{F}}
}_{\text{($\flat\flat\flat$)}}.
\end{equation}
Note that, for all $\mH \in \mathcal{A}_\delta$, the first summand ($\flat\flat$) is further bounded from below by
\[
\frac{\langle \norm{\mX_\sharp}_* \mU_\sharp \mV_\sharp^\transpose, \mH \rangle}{r \norm{\mX_\sharp}_\mathrm{F} \norm{\mH}_\mathrm{F}}
\geq - \frac{\norm{\mX_\sharp}_*}{r \norm{\mX_\sharp}_\mathrm{F}} \cdot \frac{\sqrt{r} \delta}{1-\lambda}.
\]
The second term ($\flat\flat\flat$) can be upper-bounded by the Cauchy-Schwarz inequality with 
\[
\Bigg\|\mX_\sharp - \frac{\norm{\mX_\sharp}_* \mU_\sharp \mV_\sharp^\transpose}{r} \Bigg\|_\mathrm{F}
\leq \frac{\sqrt{r} (\sigma_1 - \sigma_r)}{2}.
\]
By plugging in the above estimates to \eqref{eq:lb_rho}, we obtain a sufficient condition for $\rho \geq -0.9$ given by
\begin{equation}
\label{eq:sufficient_rho_away_negative1}
\frac{\delta}{1-\lambda} \leq 
\frac{\sqrt{r} \norm{\mX_\sharp}_\mathrm{F}}{\norm{\mX_\sharp}_*}
\cdot \Big(0.9 - \frac{\sqrt{r}(\sigma_1(\mX_\sharp) - \sigma_r(\mX_\sharp))}{2\norm{\mX_\sharp}_\mathrm{F}}\Big).
\end{equation}
Here the right-hand side of \eqref{eq:sufficient_rho_away_negative1} is no larger than $(2.8-\kappa)/2$. Therefore, \eqref{eq:cond_flat_rankr} implies that $1+\rho \geq 0.1$. Then Lemma~\ref{lemma:gaussprod} provides the lower bound in \eqref{eq:res:lemma:lbp_iidG}.  This completes the proof.

\section{Proof of Lemma~\ref{lemma:ubrc_iidG}}
\label{sec:proof:lemma:ubrc_iidG}

Without loss of generality, we may assume $\norm{\mX_\sharp}_\mathrm{F} = \norm{\mH}_\mathrm{F} = 1$.
Since $\mathcal{P}_T$ and $\mathcal{P}_{T^\perp}$ are orthogonal projection operators onto corresponding subspaces, they are self-adjoint and idempotent linear operators.  Therefore, it follows that
\[
\langle \mPhi_m, \mH \rangle
= \langle \mathcal{P}_T(\mPhi_m), \mathcal{P}_T(\mH) \rangle
+ \langle \mathcal{P}_{T^\perp}(\mPhi_m), \mathcal{P}_{T^\perp}(\mH) \rangle.
\]
Then by H\"older's inequality, we obtain
\begin{align}
\mathfrak{C}_M(\mathcal{A}_\delta)
{} & = \mathbb{E} \sup_{\mH \in \mathcal{A}_\delta} \frac{1}{\sqrt{M}} \sum_{m=1}^M \epsilon_m \langle \mX_\sharp, \mPhi_m \rangle \langle \mPhi_m, \mH \rangle \nonumber \\
{} & \leq \mathbb{E} \, \Big\| \frac{1}{\sqrt{M}} \sum_{m=1}^M \epsilon_m \mathcal{P}_T(\mPhi_m) \langle \mPhi_m, \mX_\sharp \rangle \Big\|_\mathrm{F} \cdot \sup_{\mH \in \mathcal{A}_\delta} \norm{\mathcal{P}_T(\mH)}_\mathrm{F} \nonumber \\
{} & + \mathbb{E} \, \Big\| \frac{1}{\sqrt{M}} \sum_{m=1}^M \epsilon_m \mathcal{P}_{T^\perp}(\mPhi_m) \langle \mPhi_m, \mX_\sharp \rangle \Big\| \cdot \sup_{\mH \in \mathcal{A}_\delta} \norm{\mathcal{P}_{T^\perp}(\mH)}_* \nonumber \\
{} & \leq \mathbb{E} \, \Big\| \frac{1}{\sqrt{M}} \sum_{m=1}^M \epsilon_m \mathcal{P}_T(\mPhi_m) \langle \mPhi_m, \mX_\sharp \rangle \Big\|_\mathrm{F} \label{eq:up_rc_part1} \\
{} & + \mathbb{E} \, \Big\| \frac{1}{\sqrt{M}} \sum_{m=1}^M \epsilon_m \mathcal{P}_{T^\perp}(\mPhi_m) \langle \mPhi_m, \mX_\sharp \rangle \Big\| \cdot \frac{\sqrt{r}(1-\lambda+\delta)}{\lambda}, \label{eq:up_rc_part2}
\end{align}
where the last step follows from the expression of $\mathcal{A}_\delta$ in \eqref{eq:A_delta2}.

The part in \eqref{eq:up_rc_part1} is upper-bounded by
\begin{align*}
{} & \mathbb{E} \, \Big\| \frac{1}{\sqrt{M}} \sum_{m=1}^M \epsilon_m \mathcal{P}_T(\mPhi_m) \langle \mPhi_m, \mX_\sharp \rangle \Big\|_\mathrm{F}
%\\
%{} & \quad
\leq \sqrt{\mathbb{E} \, \Big\| \frac{1}{\sqrt{M}} \sum_{m=1}^M \epsilon_m \mathcal{P}_T(\mPhi_m) \langle \mPhi_m, \mX_\sharp \rangle \Big\|_\mathrm{F}^2} \\
{} & \quad = \sqrt{\frac{1}{M} \sum_{m=1}^M \mathbb{E} \norm{\mathcal{P}_T(\mPhi_m) \langle \mPhi_m, \mX_\sharp \rangle}_\mathrm{F}^2}
%\\
%{} & \quad
= \sqrt{\mathbb{E} \norm{\mathcal{P}_T(\mPhi) \langle \mPhi, \mX_\sharp \rangle}_\mathrm{F}^2},
\end{align*}
where the first step follows from Jensen's inequality; the second step holds since $(\epsilon_m)_{m=1}^M$ is a Rademacher sequence; the last step holds since $\mPhi_1,\dots,\mPhi_M$ are independent copies of $\mPhi$.
Indeed, since
\[
\mathcal{P}_T(\mPhi)
= \mU_\sharp \mU_\sharp^* \mPhi + (\mId_{d_1}-\mU_\sharp\mU_\sharp^*)\mPhi \mV_\sharp\mV_\sharp^*,
\]
it follows that
\begin{equation}
\label{eq:exp_PTPhiPhiX}
\mathbb{E} \norm{\mathcal{P}_T(\mPhi) \langle \mPhi, \mX_\sharp \rangle}_\mathrm{F}^2
= \mathbb{E}
\norm{\mU_\sharp\mU_\sharp^*\mPhi}_\mathrm{F}^2
\langle \mPhi, \mX_\sharp \rangle^2
+ \norm{(\mId_{d_1}-\mU_\sharp\mU_\sharp^*)\mPhi \mV_\sharp\mV_\sharp^*}_\mathrm{F}^2
\langle \mPhi, \mX_\sharp \rangle^2.
\end{equation}
The first summand in the right-hand side of \eqref{eq:exp_PTPhiPhiX} is computed as
\begin{align*}
\mathbb{E}
\norm{\mU_\sharp^\transpose \mPhi}_\mathrm{F}^2
\langle \mU_\sharp^\transpose \mPhi, \mU_\sharp^\transpose \mX_\sharp \rangle^2
&= \mathrm{tr}
\Big[
\mathbb{E}
\langle \mathrm{vec}(\mU_\sharp^\transpose \mPhi), \mathrm{vec}(\mU_\sharp^\transpose \mX_\sharp) \rangle^2
\mathrm{vec}(\mU_\sharp^\transpose \mPhi) \mathrm{vec}(\mU_\sharp^\transpose \mPhi)^\transpose
\Big] \\
&= \mathrm{tr}\left(2 \, \mathrm{vec}(\mU_\sharp^\transpose \mX_\sharp) \mathrm{vec}(\mU_\sharp^\transpose \mX_\sharp)^\transpose + \mId_{rd_2}\right)  = 2 + rd_2,
\end{align*}
where the second step follows from Lemma~\ref{lemma:exp_gauss_quad_ip} since $\mathrm{vec}(\mU_\sharp^\transpose \mPhi) \sim \mathcal{N}(\vzero,\mId_{rd_2})$.

Let $\mPhi'$ be an independent copy of $\mPhi$.  Since $(\mId_{d_1}-\mU_\sharp\mU_\sharp^*)\mPhi$ is independent of $\mU_\sharp\mU_\sharp^*\mPhi$, the second summand in the right-hand side of \eqref{eq:exp_PTPhiPhiX} is written as
\begin{align*}
& \mathbb{E} \norm{(\mId_{d_1}-\mU_\sharp\mU_\sharp^*) \mPhi \mV_\sharp\mV_\sharp^*}_\mathrm{F}^2
\langle \mU_\sharp\mU_\sharp^* \mPhi', \mX_\sharp \rangle^2 \\
&= \mathbb{E}_{\mPhi} \left\|\left(\mV_\sharp\mV_\sharp^* \otimes (\mId_{d_1}-\mU_\sharp\mU_\sharp^*)\right)  \mathrm{vec}(\mPhi)\right\|_2^2 \, \mathbb{E}_{\mPhi'} \langle \mPhi', \mX_\sharp \rangle^2 \\
&= \mathrm{tr}\left(\mV_\sharp\mV_\sharp^* \otimes (\mId_{d_1}-\mU_\sharp\mU_\sharp^*)\right)  = r(d_1-r).
\end{align*}
Therefore, we obtain
\[
\mathbb{E} \norm{\mathcal{P}_T(\mPhi) \langle \mPhi, \mX_\sharp \rangle}_\mathrm{F}^2
= r(d_1+d_2-r)+2.
\]
By Jensen's inequality, the expectation in \eqref{eq:up_rc_part2} is upper-bounded by
\[
\mathbb{E} \, \Big\| \sum_{m=1}^M \epsilon_m \mathcal{P}_{T^\perp}(\mPhi_m) \langle \mPhi_m, \mX_\sharp \rangle \Big\|
\leq \Big( \mathbb{E} \, \Big\| \sum_{m=1}^M \epsilon_m \mathcal{P}_{T^\perp}(\mPhi_m) \langle \mPhi_m, \mX_\sharp \rangle \Big\|^{2p} \Big)^{1/2p}
\]
for all $p \in \mathbb{N}$.  Then we apply the noncommutative Rosenthal inequality (Theorem~\ref{thm:ncrosenthal}) for
\[
\mY_m = \epsilon_m \mathcal{P}_{T^\perp}(\mPhi_m) \langle \mPhi_m, \mX_\sharp \rangle, \quad m=1,\dots,M.
\]
Since $\mathcal{P}_{T}(\mX_\sharp) = \mX_\sharp$ and $\mathcal{P}_{T}(\mPhi_m)$ is independent from $\mathcal{P}_{T^\perp}(\mPhi_m)$, it follows
\[
\mY_m = \epsilon_m \mathcal{P}_{T^\perp}(\mPhi_m) \langle \mathcal{P}_T(\mPhi_m'), \mX_\sharp \rangle, \quad m=1,\dots,M,
\]
where $\mPhi_1',\dots,\mPhi_M'$ are independent copies of $\mPhi_1,\dots,\mPhi_M$. Furthermore, we have $\E \mY_m = \vzero$ for $m=1,\dots,M$.
By direct computation with Lemma~\ref{lemma:exp_gauss_quad_ip}, we obtain
\[
\mathbb{E} \mY_m \mY_m^\transpose = \mathrm{tr}(\mP_{\mV_\sharp^\perp}) \mP_{\mU_\sharp^\perp}
\quad \text{and} \quad
\mathbb{E} \mY_m^\transpose \mY_m = \mathrm{tr}(\mP_{\mU_\sharp^\perp}) \mP_{\mV_\sharp^\perp},
\quad m=1,\dots,M.
\]
Therefore, we obtain
\[
\Big\| \sum_{m=1}^M \E \mY_m \mY_m^\transpose \Big\|^{1/2}
\vee
\Big\| \sum_{m=1}^M \E \mY_m^\transpose \mY_m \Big\|^{1/2}
\leq \sqrt{M(d_1+d_2)}.
\]

Next we derive an upper bound for $\left(\sum_{m=1}^M \mathbb{E} \norm{\mY_m}^{2p}\right)^{1/2p}$, which coincides with $M^{1/2p} \left(\mathbb{E} \norm{\mY_m}^{2p}\right)^{1/2p}$ for any $m \in \{1,\dots,M\}$. 
Since $\mathcal{P}_T(\mPhi_m)$ and $\mathcal{P}_{T^\perp}(\mPhi_m)$ are independent, it follows that the spectral norm of $\mY_m$ satisfies
\begin{align*}
\mathbb{E} \norm{\mY_m}^{2p}
\leq \Big(\mathbb{E} \norm{\mathcal{P}_{T^\perp}(\mPhi_m)}^{2p}\Big) \cdot \Big(\mathbb{E}|\langle \mathcal{P}_T(\mPhi_m), \mX_\sharp \rangle|^{2p}\Big)
%\leq \norm{\mX_\sharp}_\mathrm{F}^2 \cdot \mathbb{E} \norm{\mPhi_m}^2.
\leq (C \sqrt{p})^{2p} \, \mathbb{E} \norm{\mPhi_m}^{2p},
\end{align*}
where the last inequality follows from the fact that $\langle \mathcal{P}_T(\mPhi_m), \mX_\sharp \rangle \sim \mathcal{N}(0,1)$ satisfies 
\[
\mathbb{E} |\langle \mathcal{P}_T(\mPhi_m), \mX_\sharp \rangle|^{2p} \leq (C \sqrt{p})^{2p}.
\]
It remains to get an upper bound on $\mathbb{E} \norm{\mPhi_m}^{2p}$.
Note that $\norm{\mPhi_m}^2 = \norm{\mPhi_m^\transpose \mPhi_m}$ where $\mPhi_m^\transpose \mPhi_m$ follows the Wishart distribution.  Without loss of generality, we may assume $d_1 \leq d_2$ (otherwise we consider $\mPhi_m \mPhi_m^\transpose$ instead of $\mPhi_m^\transpose \mPhi_m$).  Then Lemma~\ref{lemma:wgaussouterproduct} implies
\begin{equation}
\label{eq:wishart_moment_bnd}
(\mathbb{E} \norm{\mPhi_m^\transpose \mPhi_m}^p)^{1/p}
\leq d_1 + C_1 \left(\sqrt{p d_1 d_2} + p d_1^{1/p} \left(d_2 + p\right)\right) .
\end{equation}
Indeed, \eqref{eq:wishart_moment_bnd} is obtained by Lemma~\ref{lemma:wgaussouterproduct} and the triangle inequality in the Banach space of random variables $L_p(\Omega,\mu)$. Note that $\mPhi_m^\transpose \mPhi_m$ is written as
\[
\mPhi_m^\transpose \mPhi_m = \sum_{k=1}^{d_1} \vg_k \vg_k^\transpose,
\]
where $\vg_1,\dots,\vg_{d_1}$ are independent copies of $\vg \sim \mathcal{N}(\bm{0},\mId_{d_2})$. Then it follows by Lemma~\ref{lemma:wgaussouterproduct} that
\[
\left(\mathbb{E} \left\|\frac{1}{d_1} \mPhi_m^\transpose \mPhi_m - \mId_{d_2}\right\|^p\right)^{1/p} 
\leq C_1 \left(d_1^{-1/2} \sqrt{p d_2} + d_1^{1/p-1} p \left(d_2 + p\right)\right) , 
\]
which, together with the triangle inequality and the homogeneity of $L_p$-norm, implies \eqref{eq:wishart_moment_bnd}. 
Then taking the square root on both sides of \eqref{eq:wishart_moment_bnd} gives 
\[
(\mathbb{E} \norm{\mPhi_m}^{2p})^{1/2p}
\leq C_2 \left(\sqrt{d_1} + (p d_1 d_2)^{1/4} + \sqrt{p} d_1^{1/2p} \sqrt{d_2} + p d_1^{1/2p}\right) .
\]

By collecting the above estimates, Theorem~\ref{thm:ncrosenthal} implies
\begin{align*}
{} & \Big( \mathbb{E} \, \Big\| \frac{1}{\sqrt{M}} \sum_{m=1}^M \epsilon_m \mathcal{P}_{T^\perp}(\mPhi_m) \langle \mPhi_m, \mX_\sharp \rangle \Big\|^{2p} \Big)^{1/2p} \\
{} & \leq C_3
\Big( \sqrt{p(d_1+d_2)} + M^{1/2p-1/2} p^{3/2} \left(\sqrt{d_1} + (p d_1 d_2)^{1/4} + \sqrt{p} d_1^{1/2p} \sqrt{d_2} + p d_1^{1/2p}\right)  \Big).
\end{align*}
For the brevity, let $d = d_1+d_2$. Let us choose $p = 1 \vee \log d$.  Then $1/2p-1/2 \leq 0$. Since $M \geq d$, we obtain
\[
M^{1/2p-1/2} \leq d^{1/2p-1/2} \leq \frac{d^{1/2\log d}}{\sqrt{d}} \leq C_4 d^{-1/2}.
\]
Furthermore, we have
\begin{align*}
\sqrt{d_1} + (pd_1d_2)^{1/4} + \sqrt{p} d_1^{1/2p} \sqrt{d_2} + p^{3/2} d_1^{1/2p} d_2^{1/2p} 
\leq C_5 \left( \sqrt{d \log d} + \log^{3/2} d \right).
\end{align*}
Combining the above estimates provides
\begin{align*}
\mathbb{E} \, \Big\| \frac{1}{\sqrt{M}} \sum_{m=1}^M \epsilon_m \mathcal{P}_{T^\perp}(\mPhi_m) \langle \mPhi_m, \mX_\sharp \rangle \Big\|
& \leq 
\Big( \mathbb{E} \, \Big\| \frac{1}{\sqrt{M}} \sum_{m=1}^M \epsilon_m \mathcal{P}_{T^\perp}(\mPhi_m) \langle \mPhi_m, \mX_\sharp \rangle \Big\|^{2p} \Big)^{1/2p} \\
& \leq C_6 \sqrt{(d_1+d_2) \log (d_1+d_2)}.
\end{align*}
Then the upper bound in \eqref{eq:res:lemma:ubrc_iidG} is obtained by applying the above estimates to \eqref{eq:up_rc_part1} and \eqref{eq:up_rc_part2}.

\section{Proof of Lemma~\ref{lemma:lbp_rank1}}
%\begin{proof}[Proof of Lemma~\ref{lemma:lbp_rank1}]
\label{sec:proof:lemma:lbp_rank1}

The event is determined by an 1-homogeneous equation in $\mH$ and $\mX_\sharp$.  Therefore, without loss of generality, we may assume that $\norm{\mH}_\mathrm{F} = \norm{\mX_\sharp}_\mathrm{F} = 1$.  Then $\mX_\sharp$ is written as $\vu_\sharp \vv_\sharp^*$ with $\norm{\vu_\sharp}_2 = \norm{\vv_\sharp}_2 = 1$.

First we decompose $\mathrm{Re}(\vb^* \vv_\sharp \vu_\sharp^* \va \va^* \mH \vb)$ as
\begin{align*}
\mathrm{Re}(\vb^* \vv_\sharp \vu_\sharp^* \va \va^* \mH \vb)
&= \mathrm{Re}(\vb^* \vv_\sharp \vu_\sharp^* \va \va^* \mP_{\vu_\sharp} \mH \mP_{\vv_\sharp} \vb)
+ \mathrm{Re}(\vb^* \vv_\sharp \vu_\sharp^* \va \va^* \mP_{\vu_\sharp^\perp} \mH \mP_{\vv_\sharp} \vb) \\
&+ \mathrm{Re}(\vb^* \vv_\sharp \vu_\sharp^* \va \va^* \mP_{\vu_\sharp} \mH \mP_{\vv_\sharp^\perp} \vb)
+ \mathrm{Re}(\vb^* \vv_\sharp \vu_\sharp^* \va \va^* \mP_{\vu_\sharp^\perp} \mH \mP_{\vv_\sharp^\perp} \vb).
\end{align*}
By plugging in $\mP_{\vu_\sharp} = \vu_\sharp\vu_\sharp^*$ and $\mP_{\vv_\sharp} = \vv_\sharp\vv_\sharp^*$ to the above identity, we rewrite $\mathrm{Re}(\vb^* \vv_\sharp \vu_\sharp^* \va \va^* \mH \vb)$ as
\begin{align*}
&\mathrm{Re}(\vb^* \vv_\sharp \vu_\sharp^* \va \va^* \mH \vb)
= |\vv_\sharp^* \vb|^2 |\vu_\sharp^* \va|^2 \, \mathrm{Re}(\vu_\sharp^* \mH \vv_\sharp)
+ |\vv_\sharp^* \vb|^2 |\vu_\sharp^* \va| \, \mathrm{Re}\Big(\frac{\vu_\sharp^* \va}{|\vu_\sharp^* \va|} \cdot \va^* \mP_{\vu_\sharp^\perp} \mH \vv_\sharp\Big) \\
&\qquad + |\vu_\sharp^* \va|^2 |\vv_\sharp^* \vb| \, \mathrm{Re}\Big(\frac{\overline{\vv_\sharp^* \vb}}{|\vv_\sharp^* \vb|} \cdot \vu_\sharp^* \mH \mP_{\vv_\sharp^\perp} \vb\Big)
+ |\vu_\sharp^* \va| |\vv_\sharp^* \vb| \, \mathrm{Re}\Big(\frac{\vu_\sharp^* \va}{|\vu_\sharp^* \va|} \cdot \frac{\overline{\vv_\sharp^* \vb}}{|\vv_\sharp^* \vb|} \cdot \va^* \mP_{\vu_\sharp^\perp} \mH \mP_{\vv_\sharp^\perp} \vb\Big).
\end{align*}
The following facts follow from the assumption that $\va \sim \mathcal{CN}(\vzero,\mId_{d_1})$ and $\vb \sim \mathcal{CN}(\vzero,\mId_{d_2})$ are mutually independent:
\begin{enumerate}
  \item $|\vu_\sharp^* \va|$, $|\vv_\sharp^* \vb|$, $\overline{\vu_\sharp^* \va}/|\vu_\sharp^* \va|$, $\overline{\vv_\sharp^* \vb}/|\vv_\sharp^* \vb|$, $\mP_{\vu_\sharp^\perp} \va$, and $\mP_{\vv_\sharp^\perp} \vb$ are independent random variables.
  \item $|\vu_\sharp^* \va|$ and $|\vv_\sharp^* \vb|$ follow the Rayleigh distribution with scale parameter 1.
  \item $\overline{\vu_\sharp^* \va}/|\vu_\sharp^* \va|$ and $\overline{\vv_\sharp^* \vb}/|\vv_\sharp^* \vb|$ follow the uniform distribution on the set of complex number of the unit modulus.
\end{enumerate}
Furthermore, due to the rotation invariance of the Gaussian distribution, $(\overline{\vu_\sharp^* \va}/|\vu_\sharp^* \va|) \mP_{\vu_\sharp^\perp} \va$ has the same distribution with $\mP_{\vu_\sharp^\perp} \va$.  Similarly, $(\overline{\vv_\sharp^* \vb}/|\vv_\sharp^* \vb|) \mP_{\vv_\sharp^\perp} \vb$ and $\mP_{\vv_\sharp^\perp} \vb$ have the same distribution.

Combining the above facts, we obtain that $\mathrm{Re}(\vb^* \vv_\sharp \vu_\sharp^* \va \va^* \mH \vb)$ has the same distribution with
\begin{align*}
x &:= r_1^2 r_2^2 \, \mathrm{Re}(\vu_\sharp^* \mH \vv_\sharp)
+ r_1 r_2^2 \, \mathrm{Re}(\va^* \mP_{\vu_\sharp^\perp} \mH \vv_\sharp)
+ r_1^2 r_2 \, \mathrm{Re}(\vu_\sharp^* \mH \mP_{\vv_\sharp^\perp} \vb)
+ r_1 r_2 \, \mathrm{Re}(\va^* \mP_{\vu_\sharp^\perp} \mH \mP_{\vv_\sharp^\perp} \vb),
\end{align*}
where $r_1,r_2,\va,\vb$ are independent and $r_1,r_2 \sim \text{Rayleigh(1)}$.

Now it suffices to compute the probability of the event $\mathcal{E}$ defined by
\[
\mathcal{E} := \{ x \geq \tau' \}.
\]
For positive constants $\alpha,\beta$, we define another event $\mathcal{E}_0$ by
\[
\mathcal{E}_0 := \{ \alpha \leq r_1 \leq \beta, ~ \alpha \leq r_2 \leq \beta \}.
\]
For example, we may set $\alpha = 0.9$ and $\beta = 1.1$.  Then $\mathbb{P}(\mathcal{E}_0) \geq 0.12$.

Let $z_1,z_2,z_3$ be random variables defined by
\begin{align*}
z_1 := \mathrm{Re}(\va^* \mP_{\vu_\sharp^\perp} \mH \vv_\sharp), \quad
z_2 := \mathrm{Re}(\vu_\sharp^* \mH \mP_{\vv_\sharp^\perp} \vb), \quad \text{and} \quad
z_3 := \mathrm{Re}(\va^* \mP_{\vu_\sharp^\perp} \mH \mP_{\vv_\sharp^\perp} \vb).
\end{align*}
Since $\va \sim \mathcal{CN}(\vzero,\mId_{d_1})$, if $\mP_{\vu_\sharp^\perp} \mH \vv_\sharp \neq \vzero$, then it follows that $\va^* \mP_{\vu_\sharp^\perp} \mH \vv_\sharp \sim \mathcal{CN}(0,\norm{\mP_{\vu_\sharp^\perp} \mH \mP_{\vv_\sharp}}_\mathrm{F}^2)$ and its real part $z_1$ follows $\mathcal{N}(0,\norm{\mP_{\vu_\sharp^\perp} \mH \mP_{\vv_\sharp}}_\mathrm{F}^2/2)$.  Otherwise, $\mP_{\vu_\sharp^\perp} \mH \vv_\sharp = \vzero$ implies $z_1 = 0$.  Similarly, $z_2 \sim \mathcal{N}(0,\norm{\mP_{\vv_\sharp^\perp} \mH^* \mP_{\vu_\sharp}}_\mathrm{F}^2/2)$ if $\mP_{\vv_\sharp^\perp} \mH^* \vu_\sharp \neq \vzero$; $z_2 = 0$ otherwise.  By the independence between $\va$ and $\vb$, it follows that $z_1+z_2 \sim \mathcal{N}(0,\norm{\mP_{\vu_\sharp^\perp} \mH \mP_{\vv_\sharp} + \mP_{\vu_\sharp} \mH \mP_{\vv_\sharp^\perp}}_\mathrm{F}^2/2)$ if $\norm{\mP_{\vu_\sharp^\perp} \mH \mP_{\vv_\sharp}}_\mathrm{F}^2 + \norm{\mP_{\vu_\sharp} \mH \mP_{\vv_\sharp^\perp}}_\mathrm{F}^2 > 0$; $z_1 + z_2 = 0$ otherwise.  In both cases, $z_1 + z_2$ has a symmetric distribution, that is $z_1+z_2$ is equivalent to $-(z_1+z_2)$ in distribution.

Furthermore, we can rewrite $z_3$ as a Gaussian bilinear form, i.e.,
\[
z_3 = \tilde{\va}^\transpose \mQ \tilde{\vb}
\]
for
\begin{align*}
\mQ &= \frac{1}{2}
\begin{bmatrix}
\mathrm{Re}(\mP_{\vu_\sharp^\perp} \mH \mP_{\vv_\sharp^\perp}) & - \mathrm{Im}(\mP_{\vu_\sharp^\perp} \mH \mP_{\vv_\sharp^\perp}) \\
\mathrm{Im}(\mP_{\vu_\sharp^\perp} \mH \mP_{\vv_\sharp^\perp}) & \mathrm{Re}(\mP_{\vu_\sharp^\perp} \mH \mP_{\vv_\sharp^\perp})
\end{bmatrix},
\end{align*}
where
\begin{align*}
\tilde{\va} &=
\sqrt{2}
\begin{bmatrix}
\mathrm{Re}(\va) \\
\mathrm{Im}(\va)
\end{bmatrix}
\sim \mathcal{N}(\vzero,\mId_{2d_1}),
\qquad
\tilde{\vb} =
\sqrt{2}
\begin{bmatrix}
\mathrm{Re}(\vb) \\
\mathrm{Im}(\vb)
\end{bmatrix}
\sim \mathcal{N}(\vzero,\mId_{2d_2}).
\end{align*}
It follows that $z_3$ has a symmetric distribution.  Furthermore, since $z_3$ is a Gaussian bilinear form, it has a mixed subexponential-subgaussian tail given by
\begin{equation}
\label{eq:ubtailH-W}
\mathbb{P}(|z_3| \geq t)
\leq C \exp\left[
- \frac{1}{C}
\left(
\frac{t^2}{\norm{\mQ}_\mathrm{F}^2}
\wedge
\frac{t}{\norm{\mQ}}
\right)
\right], \quad \forall t > 0
\end{equation}
for a numerical constant $C$.  
Lata{\l}a \cite{latala2006estimates} showed that this tail bound is tight with an analogous lower bound given by
\begin{equation*}
%\label{eq:lbtailH-W}
\mathbb{P}(|z_3| \geq t)
\geq \frac{1}{C} \exp\left[
- C
\left(
\frac{t^2}{\norm{\mQ}_\mathrm{F}^2}
\wedge
\frac{t}{\norm{\mQ}}
\right)
\right], \quad \forall t > 0.
\end{equation*}
By direct calculation, we obtain
\[
\norm{\mQ}_\mathrm{F} = \frac{1}{\sqrt{2}} \, \norm{\mP_{\vu_\sharp^\perp} \mH \mP_{\vv_\sharp^\perp}}_\mathrm{F}
\]
and
\[
\norm{\mQ} = \frac{1}{2} \, \norm{\mP_{\vu_\sharp^\perp} \mH \mP_{\vv_\sharp^\perp}}.
\]

Now we are ready to derive a lower bound on the probability of the event $\mathcal{E}$ using the aforementioned properties $z_1,z_2,z_3$.  It follows from the definition of the conditional probability that
\begin{align}
\frac{\mathbb{P}(\mathcal{E})}{\mathbb{P}(\mathcal{E}_0)}
& \geq \frac{\mathbb{P}(\mathcal{E} \cap \mathcal{E}_0)}{\mathbb{P}(\mathcal{E}_0)}
= \mathbb{P}(\mathcal{E} | \mathcal{E}_0)
= \mathbb{P}\Big(
r_1 r_2 \, \mathrm{Re}(\vu_\sharp^* \mH \vv_\sharp)
+ r_2 z_1
+ r_1 z_2
+ z_3
\geq \frac{\tau'}{r_1 r_2} ~\Big|~ \mathcal{E}_0
\Big). \label{eq:lb_prob1}
\end{align}
As we choose $\alpha < \beta$ as numerical constants, $\mathbb{P}(\mathcal{E}_0)$ is another numerical constant.  It remains to show that the lower bound in \eqref{eq:lb_prob1} is larger than a numerical constant.  We consider the two complementary scenarios below.

\noindent\textbf{Case 1:} First, we consider the case when $\mH$ satisfies
\begin{equation}
\label{eq:smallperp}
\norm{\mP_{\vu_\sharp^\perp} \mH \mP_{\vv_\sharp^\perp}} \leq \zeta
\end{equation}
for some constant $0 < \zeta < 1$, which we will specify later.

Let $\tau'' > 0$. Then by the inclusion-exclusion principle, the right-hand-side of \eqref{eq:lb_prob1} is lower-bounded by
\begin{align}
&\mathbb{P}\Big(
r_1 r_2 \, \mathrm{Re}(\vu_\sharp^* \mH \vv_\sharp)
+ r_2 z_1
+ r_1 z_2
+ z_3
\geq \frac{\tau'}{r_1 r_2} \,|\, \mathcal{E}_0
\Big) \nonumber \\
&\geq
\mathbb{P}\Big(
r_1 r_2 \, \mathrm{Re}(\vu_\sharp^* \mH \vv_\sharp)
+ r_2 z_1
+ r_1 z_2
\geq \frac{\tau'+\tau''}{r_1 r_2} \,|\, \mathcal{E}_0
\Big)
-
\mathbb{P}\Big(
z_3
< - \frac{\tau''}{r_1 r_2} \,|\, \mathcal{E}_0
\Big). \label{eq:lbp_case1_eq0}
\end{align}
In the sequel we will use the fact that for then the tail probability is
\begin{align}
g\sim\mathcal{N}(0,\varsigma^2) & \implies \mathbb{P}\left(g > t \right)\ \text{is increasing in}\ \varsigma\ \text{and decreasing in}\ t\geq 0\,. \label{eq:Gaussian-fact}
\end{align}
The following cases on the sign of $\mathrm{Re}(\vu_\sharp^* \mH \vv_\sharp)$ have to be distinguished. First suppose that $\mathrm{Re}(\vu_\sharp^* \mH \vv_\sharp) \leq 0$. Conditioned on $r_1$ and $r_2$, the random variable $r_2 z_1 + r_1 z_2$ becomes a Gaussian and invoking \eqref{eq:Gaussian-fact} yields
\begin{align*}
& \mathbb{P}\left(r_2 z_1 + r_1 z_2
\geq \frac{\tau'+\tau''}{r_1 r_2} - r_1 r_2 \, \mathrm{Re}(\vu_\sharp^* \mH \vv_\sharp) \,|\, \mathcal{E}_0,r_1,r_2
\right) \\
&\ge \mathbb{P}\left( \alpha (z_1 + z_2)
\geq \frac{\tau'+\tau''}{\alpha^2} - \beta^2 \, \mathrm{Re}(\vu_\sharp^* \mH \vv_\sharp) \,|\, \mathcal{E}_0,r_1,r_2\right)\,.
\end{align*}
Since the right-hand side of the above inequality is independent of $r_1$ and $r_2$, we can conclude that
\[\mathbb{P}\left(r_2 z_1 + r_1 z_2
\geq \frac{\tau'+\tau''}{r_1 r_2} - r_1 r_2 \, \mathrm{Re}(\vu_\sharp^* \mH \vv_\sharp) \,|\, \mathcal{E}_0\right) \ge \mathbb{P}\left( \alpha (z_1 + z_2)
\geq \frac{\tau'+\tau''}{\alpha^2} - \beta^2 \, \mathrm{Re}(\vu_\sharp^* \mH \vv_\sharp) \right)\,. \]
Furthermore, $z_3$ is symmetric and we obtain an upper estimate of the tail probability of $z_3$ in \eqref{eq:lbp_case1_eq0} given by
\[
\mathbb{P}\Big(
z_3
< - \frac{\tau''}{r_1 r_2} \,|\, \mathcal{E}_0
\Big)
\leq
\mathbb{P}\Big(
z_3
< - \frac{\tau''}{\beta^2}
\Big)
=
\frac{1}{2} \,
\mathbb{P}\Big(
|z_3|
\geq \frac{\tau''}{\beta^2}
\Big).
\]

By combining the above bounds, the lower estimate in \eqref{eq:lbp_case1_eq0} is further bounded from below by
\begin{align*}
&
\mathbb{P}\Big(
r_1 r_2 \, \mathrm{Re}(\vu_\sharp^* \mH \vv_\sharp)
+ r_2 z_1
+ r_1 z_2
\geq \frac{\tau'+\tau''}{r_1 r_2} \,|\, \mathcal{E}_0
\Big)
-
\mathbb{P}\Big(
z_3
< - \frac{\tau''}{r_1 r_2} \,|\, \mathcal{E}_0
\Big) \nonumber \\
&\geq
\mathbb{P}\Big(
z_1 + z_2
\geq \frac{\tau'+\tau''}{\alpha^3}
- \frac{\beta^2 \, \mathrm{Re}(\vu_\sharp^* \mH \vv_\sharp)}{\alpha}
\Big)
-
\frac{1}{2} \,
\mathbb{P}\Big(
|z_3|
\geq \frac{\tau''}{\beta^2}
\Big). %\label{eq:lbp_case1_eq0negative}
\end{align*}
Because $z_1+z_2 \sim \mathcal{N}(0,\norm{\mP_{\vu_\sharp^\perp} \mH \mP_{\vv_\sharp} + \mP_{\vu_\sharp} \mH \mP_{\vv_\sharp^\perp}}_\mathrm{F}^2/2)$, in order to lower-bound the tail probability of $z_1+z_2$, we need to compute a lower estimate of its variance.
It follows from \eqref{eq:A_delta2} that every $\mH \in \A_\delta$ satisfies
\begin{equation}
\label{eq:bndAdelta2}
\norm{\mathcal{P}_{T^\perp}(\mH)}_* \leq \frac{1-\lambda+\delta}{\lambda}.
\end{equation}
By \eqref{eq:smallperp} and \eqref{eq:bndAdelta2} together with H\"older's inequality, we also obtain
\begin{equation}
\label{eq:smallperpH}
\norm{\mP_{\vu_\sharp^\perp} \mH \mP_{\vv_\sharp^\perp}}_\mathrm{F}^2
\leq \norm{\mP_{\vu_\sharp^\perp} \mH \mP_{\vv_\sharp^\perp}} \cdot
\norm{\mP_{\vu_\sharp^\perp} \mH \mP_{\vv_\sharp^\perp}}_*
\leq \frac{(1 - \lambda + \delta)\zeta}{\lambda}.
\end{equation}
Furthermore, by Lemma~\ref{lemma:imagpart}, every $\mH \in \mathcal{R}_\delta$ satisfies
\begin{align*}
\frac{1-\delta^2}{\delta^2} \cdot |\mathrm{Im}(\vu_\sharp^* \mH \vv_\sharp)|^2
&\leq
\norm{\mH - \mP_{\vu_\sharp} \mH \mP_{\vv_\sharp}}_\mathrm{F}^2.
\end{align*}
Then it follows that
\begin{equation}
\label{eq:lbp_case1_eq1}
\begin{aligned}
1 = \norm{\mH}_\mathrm{F}^2
&= \norm{\mP_{\vu_\sharp} \mH \mP_{\vv_\sharp^\perp} + \mP_{\vu_\sharp^\perp} \mH \mP_{\vv_\sharp} + \mP_{\vu_\sharp^\perp} \mH \mP_{\vv_\sharp^\perp}}_\mathrm{F}^2
+ |\mathrm{Im}(\vu_\sharp^* \mH \vv_\sharp)|^2
+ |\mathrm{Re}(\vu_\sharp^* \mH \vv_\sharp)|^2 \\
&\leq \frac{\norm{\mP_{\vu_\sharp} \mH \mP_{\vv_\sharp^\perp} + \mP_{\vu_\sharp^\perp} \mH \mP_{\vv_\sharp} + \mP_{\vu_\sharp^\perp} \mH \mP_{\vv_\sharp^\perp}}_\mathrm{F}^2}{1-\delta^2}
+ |\mathrm{Re}(\vu_\sharp^* \mH \vv_\sharp)|^2.
\end{aligned}
\end{equation}
It also follows from \eqref{eq:A_delta2} that every $\mH \in \mathcal{A}_\delta$ satisfies
\begin{equation}
\label{eq:bndAdelta1}
\mathrm{Re}(\vu_\sharp^* \mH \vv_\sharp) \geq \frac{-\delta}{1-\lambda}.
\end{equation}
The assumption $\lambda + \delta < 1$ implies that the right-hand side of \eqref{eq:bndAdelta1} is strictly larger than $-1$.

By \eqref{eq:bndAdelta1} and $\mathrm{Re}(\vu_\sharp^* \mH \vv_\sharp) \leq 0$, we also have
\begin{equation}
\label{eq:lbp_case1_eq2}
|\mathrm{Re}(\vu_\sharp^* \mH \vv_\sharp)| \leq \frac{\delta}{1-\lambda}.
\end{equation}
Therefore, by applying \eqref{eq:lbp_case1_eq2} to \eqref{eq:lbp_case1_eq1}, after a rearrangement, we obtain
\[
\norm{\mP_{\vu_\sharp} \mH \mP_{\vv_\sharp^\perp} + \mP_{\vu_\sharp^\perp} \mH \mP_{\vv_\sharp} + \mP_{\vu_\sharp^\perp} \mH \mP_{\vv_\sharp^\perp}}_\mathrm{F}^2
\geq (1-\delta^2) \Big( 1 - \frac{\delta^2}{(1-\lambda)^2} \Big).
\]
Then \eqref{eq:smallperpH} implies
\begin{equation}
\label{eq:lbp_case1_eq3}
\norm{\mP_{\vu_\sharp} \mH \mP_{\vv_\sharp^\perp} + \mP_{\vu_\sharp^\perp} \mH \mP_{\vv_\sharp} }_\mathrm{F}^2
\geq (1-\delta^2) \Big( 1 - \frac{\delta^2}{(1-\lambda)^2} \Big) - \frac{(1-\lambda+\delta)\zeta}{\lambda}.
\end{equation}
Now, from \eqref{eq:lbp_case1_eq2} and \eqref{eq:lbp_case1_eq3}, we obtain
\begin{align}
{} & \mathbb{P}\Big(
z_1 + z_2
\geq \frac{\tau'+\tau''}{\alpha^3}
- \frac{\beta^2 \, \mathrm{Re}(\vu_\sharp^* \mH \vv_\sharp)}{\alpha}
\Big)
%\nonumber \\
%{} &
\geq
\mathbb{P}\Big(
z_1 + z_2
\geq \frac{\tau'+\tau''}{\alpha^3} + \frac{\beta^2 \delta}{\alpha (1-\lambda)}
\Big)
%\nonumber \\
%{} &
\geq
\mathbb{P}\Big(
g
\geq \frac{t}{\sigma_\zeta}
\Big) \label{eq:lbp_case1_eq5}
\end{align}
for $g \sim \mathcal{N}(0,1)$, where
\begin{equation*}
%\label{eq:lbp_case1_eq6}
\sigma_\zeta = \sqrt{(1-\delta^2) \Big( 1 - \frac{\delta^2}{(1-\lambda)^2} \Big) - \frac{(1-\lambda+\delta)\zeta}{\lambda}}
\end{equation*}
and
\begin{equation}
\label{eq:lbp_case1_eq7}
t = \frac{\tau'+\tau''}{\alpha^3} + \frac{\beta^2 \delta}{\alpha(1-\lambda)}.
\end{equation}
Moreover, the tail bound of $z_3$ in \eqref{eq:ubtailH-W} implies
\begin{align}
\mathbb{P}\Big(|z_3| \geq \frac{\tau''}{\beta^2}\Big)
&\leq C \exp\Big(
- \frac{2}{C}
\Big[
\frac{\tau''^2/\beta^4}{\norm{\mP_{\vu_\sharp^\perp} \mH \mP_{\vv_\sharp^\perp}}_\mathrm{F}^2}
\wedge
\frac{\tau''/\beta^2}{\norm{\mP_{\vu_\sharp^\perp} \mH \mP_{\vv_\sharp^\perp}}}
\Big]
\Big) \nonumber \\
& \leq C \exp\Big(
- \frac{2}{C}
\Big[
\frac{\lambda \tau''^2}{\beta^4 (1-\lambda+\delta)\zeta}
\wedge
\frac{\tau''}{\beta^2 \zeta}
\Big]
\Big). \label{eq:lbp_case1_eq4}
\end{align}

Note that the tail bound in \eqref{eq:lbp_case1_eq5} is monotone decreasing in $t/\sigma_\zeta$.  Furthermore, for those $\zeta$ that make $\sigma_\zeta$ positive, $t/\sigma_\zeta$ is a monotone increasing in $\zeta$.  (The condition $\delta \leq 0.2$ implies the existence of such $\zeta$.)  Hence the tail bound in \eqref{eq:lbp_case1_eq5} is monotone decreasing in $\zeta$.  On the contrary, the upper bound in \eqref{eq:lbp_case1_eq4} monotonically converges to 0 as $\zeta > 0$ decreases toward 0.  Therefore, there exists small enough $\zeta$ such that the upper bound in \eqref{eq:lbp_case1_eq4} becomes less than half of \eqref{eq:lbp_case1_eq5}.  Then the lower bound \eqref{eq:lb_prob1} is further bounded from below by the half of \eqref{eq:lbp_case1_eq5}.  Note that $\zeta$ is determined independent from all dimension parameters and hence both $\zeta$ and the resulting lower bound for the probability in \eqref{eq:lb_prob1} are numerical constants.

Next we consider the complimentary subcase where $\mathrm{Re}(\vu_\sharp^* \mH \vv_\sharp) > 0$.  Similarly to the previous subcase, since $z_3$ has a symmetric distribution, it follows that
\begin{equation*}
%\label{eq:lbp_case1_eq0positive}
\begin{aligned}
&
\mathbb{P}\Big(
r_1 r_2 \, \mathrm{Re}(\vu_\sharp^* \mH \vv_\sharp)
+ r_2 z_1
+ r_1 z_2
\geq \frac{\tau'+\tau''}{r_1 r_2} \,|\, \mathcal{E}_0
\Big)
-
\mathbb{P}\Big(
z_3
< - \frac{\tau''}{r_1 r_2} \,|\, \mathcal{E}_0
\Big)
\\
& \quad \geq
\mathbb{P}\Big(
\alpha^2 \, \mathrm{Re}(\vu_\sharp^* \mH \vv_\sharp)
+ r_2 z_1
+ r_1 z_2
\geq \frac{\tau'+\tau''}{\alpha^2} \,|\, \mathcal{E}_0
\Big)
-
\frac{1}{2}
\mathbb{P}\Big(
|z_3|
\geq \frac{\tau''}{\beta^2}
\Big).
\end{aligned}
\end{equation*}

If $\mathrm{Re}(\vu_\sharp^* \mH \vv_\sharp) \leq \alpha^{-4} (\tau'+\tau'')$, since $z_1+z_2$ is a zero-mean Gaussian variable, then it follows that
\begin{align*}
\mathbb{P}\Big(
\alpha^2 \, \mathrm{Re}(\vu_\sharp^* \mH \vv_\sharp)
+ r_2 z_1
+ r_1 z_2
\geq \frac{\tau'+\tau''}{\alpha^2} \,|\, \mathcal{E}_0
\Big)
\geq
\mathbb{P}\Big(
z_1 + z_2
\geq \frac{\tau'+\tau''}{\alpha^3}
- \alpha \, \mathrm{Re}(\vu_\sharp^* \mH \vv_\sharp)
\Big).
\end{align*}
Thus by choosing $\tau'+\tau''$ small enough one can satisfy \eqref{eq:lbp_case1_eq2}.  Thus we obtain the desired conclusion as in the previous subcase by repeating the same arguments.

If $\mathrm{Re}(\vu_\sharp^* \mH \vv_\sharp) > \alpha^{-4} (\tau'+\tau'')$  on the other hand, then
\begin{align*}
&
\mathbb{P}\Big(
\alpha^2 \, \mathrm{Re}(\vu_\sharp^* \mH \vv_\sharp)
+ r_2 z_1
+ r_1 z_2
\geq \frac{\tau'+\tau''}{\alpha^2} \,|\, \mathcal{E}_0
\Big)
\geq
\mathbb{P}\Big(
z_1 + z_2
\geq
\underbrace{\frac{\tau'+\tau''}{\alpha^2 \beta}
- \frac{\alpha^2 \, \mathrm{Re}(\vu_\sharp^* \mH \vv_\sharp)}{\beta}}_{< 0}
\Big)
> \frac{1}{2},
\end{align*}
which is larger than the other lower bounds on the tail probability.

\noindent\textbf{Case 2:} Next we consider the complementary case where
\begin{equation}
\label{eq:largeperpH}
\norm{\mP_{\vu_\sharp^\perp} \mH \mP_{\vv_\sharp^\perp}}_\mathrm{F} > \zeta,
\end{equation}
where $\zeta$ is the constant determined in the previous case. In this case, the lower estimate in \eqref{eq:lb_prob1} is further bounded from below by
\begin{equation*}
%\label{eq:lbp_case2_eq0}
\begin{aligned}
&\mathbb{P}\Big(
r_1 r_2 \, \mathrm{Re}(\vu_\sharp^* \mH \vv_\sharp)
+ r_2 z_1
+ r_1 z_2
+ z_3
\geq \frac{\tau'}{r_1 r_2} \,|\, \mathcal{E}_0
\Big) \\
& \quad \geq \mathbb{P}\Big(
- \beta^2 [\mathrm{Re}(\vu_\sharp^* \mH \vv_\sharp)]_-
+ r_2 z_1
+ r_1 z_2
+ z_3
\geq \frac{\tau'}{r_1 r_2} \,|\, \mathcal{E}_0
\Big) \\
& \quad \geq \mathbb{P}\Big(
r_2 z_1
+ r_1 z_2
+ z_3
\geq
\frac{\tau'}{r_1 r_2} + \beta^2 (1-\zeta)
\,|\, \mathcal{E}_0
\Big) \\
& \quad \geq \mathbb{P}\Big(
z_3
\geq
\frac{\tau'+\tau''}{r_1 r_2} + \beta^2 (1-\zeta)
\,|\, \mathcal{E}_0
\Big)
+ \mathbb{P}\Big(
r_2 z_1
+ r_1 z_2
\geq
- \frac{\tau''}{r_1 r_2}
\,|\, \mathcal{E}_0
\Big) - 1 \\
& \quad \geq \mathbb{P}\Big(
z_3
\geq
\frac{\tau'+\tau''}{\alpha^2} + \beta^2 (1-\zeta)
\Big)
- \mathbb{P}\Big(
z_1 + z_2 \geq \frac{\tau''}{\beta^3}
\Big) \\
& \quad \geq \frac{1}{2} \, \mathbb{P}\Big(
|z_3|
\geq
\frac{\tau'+\tau''}{\alpha^2} + \beta^2 (1-\zeta)
\Big)
- \mathbb{P}\Big(
z_1
+ z_2
\geq
\frac{\tau''}{\beta^3}
\Big),
\end{aligned}
\end{equation*}
where the second and third steps follow from \eqref{eq:lbp_case1_eq2} and the inclusion-exclusion principle, respectively.
Then, by \eqref{eq:largeperpH}, the tail bound on $z_3$ is lower-bounded by
\begin{equation}
\label{eq:lbp_case2_eq1}
\begin{aligned}
\mathbb{P}(|z_3| \geq t)
&\geq \frac{1}{C} \exp\Big(
- 2C
\Big[
%\frac{t^2}{\norm{\mP_{\vu_\sharp^\perp} \mH \mP_{\vv_\sharp^\perp}}_\mathrm{F}^2}
\frac{t^2}{\zeta^2}
\wedge
%\frac{t}{\norm{\mP_{\vu_\sharp^\perp} \mH \mP_{\vv_\sharp^\perp}}}
\frac{t}{\zeta}
\Big]
\Big),
\end{aligned}
\end{equation}
where $t$ is given in \eqref{eq:lbp_case1_eq7}.
Since $\norm{\mH}_\mathrm{F} = 1$, the variance of $z_1+z_2$ is no larger than $1/2$.
Thus the tail bound of $z_1 + z_2$ is upper-bounded by
\begin{equation}
\label{eq:lbp_case2_eq2}
\mathbb{P}\Big(
z_1 + z_2 \geq \frac{\tau''}{\beta^3}
\Big)
\leq
\exp\Big(- \frac{2\tau''^2}{\beta^6} \Big).
\end{equation}

Note that $\tau''$ still remains a free parameter.  For every $t > \zeta$ the lower bound in \eqref{eq:lbp_case2_eq1} is an exponential tail while the upper bound in \eqref{eq:lbp_case2_eq2} is a subgaussian tail.  Therefore, as $\tau''$ increases while the other parameters are fixed, by \eqref{eq:lbp_case1_eq7}, $t$ also increases as an affine function of $\tau''$ and the lower bound in \eqref{eq:lbp_case2_eq1} decays slower than the upper bound in \eqref{eq:lbp_case2_eq2}.  We may choose $\tau''$ so that the lower bound in \eqref{eq:lbp_case2_eq1} is larger then four times the upper bound in \eqref{eq:lbp_case2_eq2}.  Then the lower bound \eqref{eq:lb_prob1} is further bounded below by the resulting value of \eqref{eq:lbp_case2_eq2}.  Again, this lower bound is a numerical constant independent of scaling of all dimension parameters.
%\end{proof} %[Proof of Lemma~\ref{lemma:lbp_rank1}]

\section{Proof of Lemma~\ref{lemma:ubrc_rank1}}
%\begin{proof}[Proof of Lemma~\ref{lemma:ubrc_rank1}]
\label{sec:proof:lemma:ubrc_rank1}

Without loss of generality, we may assume that $\norm{\mH}_\mathrm{F} = \norm{\mX_\sharp}_\mathrm{F} = 1$. Then $\mX_\sharp$ is written as $\vu_\sharp \vv_\sharp^*$ where $\vu_\sharp \in \mathbb{C}^{d_1}$ and $\vv_\sharp \in \mathbb{C}^{d_2}$ satisfy $\norm{\vu_\sharp}_2 = \norm{\vv_\sharp}_2 = 1$.  With this expression of $\mX_\sharp$, the Rademacher complexity $\mathfrak{C}_M(\mathcal{A}_\delta)$ is written as
\begin{align*}
\mathfrak{C}_M(\mathcal{A}_\delta)
&=
\E
\sup_{\mH \in \mathcal{A}_\delta} \frac{1}{\sqrt{M}}
\sum_{m=1}^M \epsilon_m
\mathrm{Re}(\vb_m^* \vv_\sharp \vu_\sharp^* \va_m \va_m^* \mH \vb_m) \\
&=
\E
\sup_{\mH \in \mathcal{A}_\delta} \frac{1}{\sqrt{M}}
\sum_{m=1}^M
\epsilon_m
\mathrm{Re}\,\langle \va_m \va_m^* \vu_\sharp \vv_\sharp^* \vb_m \vb_m^*, \mH\rangle \\
&\leq \E
\sup_{\mH \in \mathcal{A}_\delta} \frac{1}{\sqrt{M}}
\sum_{m=1}^M \epsilon_m
\mathrm{Re}\,\langle \mathcal{P}_T(\va_m \va_m^* \vu_\sharp \vv_\sharp^* \vb_m \vb_m^*), \mathcal{P}_T(\mH)\rangle \\
&+ \E
\sup_{\mH \in \mathcal{A}_\delta} \frac{1}{\sqrt{M}}
\sum_{m=1}^M \epsilon_m
\mathrm{Re}\,\langle \mathcal{P}_{T^\perp}(\va_m \va_m^* \vu_\sharp \vv_\sharp^* \vb_m \vb_m^*), \mathcal{P}_{T^\perp}(\mH)\rangle \\
&\leq \E
\Big\|
\frac{1}{\sqrt{M}}
\sum_{m=1}^M \epsilon_m
\mathcal{P}_T(\va_m \va_m^* \vu_\sharp \vv_\sharp^* \vb_m \vb_m^*)
\Big\|_\mathrm{F}
\cdot \sup_{\mH \in \mathcal{A}_\delta} \norm{\mathcal{P}_T(\mH)}_\mathrm{F} \\
&+ \E
\Big\|
\frac{1}{\sqrt{M}}
\sum_{m=1}^M \epsilon_m
\mathcal{P}_{T^\perp}(\va_m \va_m^* \vu_\sharp \vv_\sharp^* \vb_m \vb_m^*)
\Big\|
\cdot \sup_{\mH \in \mathcal{A}_\delta} \norm{\mathcal{P}_{T^\perp}(\mH)}_*,
\end{align*}
where the first inequality is obtained by taking the supremum of each summand after applying $\mH = \mathcal{P}_T(\mH) + \mathcal{P}_{T^\perp}(\mH)$ and the second inequality holds by H\"older's inequality.

Since $\mathcal{P}_T$ is an orthogonal projection onto a subspace, we have $\norm{\mathcal{P}_T(\mH)}_\mathrm{F} \leq \norm{\mH}_\mathrm{F} = 1$.  Furthermore, for all $\mH \in \mathcal{A}_\delta$, $\norm{\mathcal{P}_{T^\perp}(\mH)}_*$ is upper-bounded by \eqref{eq:bndAdelta2}.
Therefore, we obtain
\begin{equation}
\label{eq:ubrc_eq1}
\begin{aligned}
\mathfrak{C}_M(\mathcal{A}_\delta)
&\leq \E
\Big\|
\frac{1}{\sqrt{M}}
\sum_{m=1}^M \epsilon_m
\mathcal{P}_T(\va_m \va_m^* \vu_\sharp \vv_\sharp^* \vb_m \vb_m^*)
\Big\|_\mathrm{F} \\
&+ \E
\Big\|
\frac{1}{\sqrt{M}}
\sum_{m=1}^M \epsilon_m
\mathcal{P}_{T^\perp}(\va_m \va_m^* \vu_\sharp \vv_\sharp^* \vb_m \vb_m^*)
\Big\|
\cdot \Big(
\frac{1-\lambda+\delta}{\lambda}
\Big).
\end{aligned}
\end{equation}

It remains to compute upper estimates of the expectation terms in \eqref{eq:ubrc_eq1}.  Since $(\epsilon_m)_{m=1}^M$ is a Rademacher sequence, we have
\begin{align*}
& \E \, \Big\|
\frac{1}{\sqrt{M}}
\sum_{m=1}^M \epsilon_m
\mathcal{P}_T(\va_m \va_m^* \vu_\sharp \vv_\sharp^* \vb_m \vb_m^*)
\Big\|_\mathrm{F}
%\\
%& \quad
\leq
\sqrt{
\E \, \Big\|
\frac{1}{\sqrt{M}}
\sum_{m=1}^M \epsilon_m
\mathcal{P}_T(\va_m \va_m^* \vu_\sharp \vv_\sharp^* \vb_m \vb_m^*)
\Big\|_\mathrm{F}^2 } \\
& \quad =
\sqrt{
\E \, \frac{1}{M}
\sum_{m=1}^M \norm{\mathcal{P}_T(\va_m \va_m^* \vu_\sharp \vv_\sharp^* \vb_m \vb_m^*)}_\mathrm{F}^2
}
%\\
%& \quad
=
\sqrt{
\E
\norm{\mathcal{P}_T(\va \va^* \vu_\sharp \vv_\sharp^* \vb \vb^*)}_\mathrm{F}^2
},
\end{align*}
where the first step follows from Jensen's inequality and the last step follows since $\va_1,\dots\va_M$ (resp. $\vb_1,\dots,\vb_M$) are independent copies of $\va$ (resp. $\vb$). 

Note that $\mathcal{P}_T(\va \va^* \vu_\sharp \vv_\sharp^* \vb \vb^*)$ is written as
\begin{align*}
\mathcal{P}_T(\va \va^* \vu_\sharp \vv_\sharp^* \vb \vb^*)
&= \va^* \vu_\sharp \vv_\sharp^* \vb \cdot \mathcal{P}_T(\va \vb^*) \\
&= \va^* \vu_\sharp \vv_\sharp^* \vb \cdot
(
\mP_{\vu_\sharp} \va \vb^* \mP_{\vv_\sharp}
+ \mP_{\vu_\sharp^\perp} \va \vb^* \mP_{\vv_\sharp}
+ \mP_{\vu_\sharp} \va \vb^* \mP_{\vv_\sharp^\perp}
),
\end{align*}
where $\mP_{\vu_\sharp} \va \vb^* \mP_{\vv_\sharp}$, $\mP_{\vu_\sharp^\perp} \va \vb^* \mP_{\vv_\sharp}$, and $\mP_{\vu_\sharp} \va \vb^* \mP_{\vv_\sharp^\perp}$ are mutually orthogonal matrices in the Hilbert space $S_2$.  Thus the Pythagorean identity implies
\begin{align*}
\norm{\mathcal{P}_T(\va \va^* \vu_\sharp \vv_\sharp^* \vb \vb^*)}_\mathrm{F}^2
&= |\va^* \vu_\sharp|^2 |\vb^* \vv_\sharp|^2 (
\norm{\mP_{\vu_\sharp} \va \vb^* \mP_{\vv_\sharp}}_\mathrm{F}^2
+ \norm{\mP_{\vu_\sharp^\perp} \va \vb^* \mP_{\vv_\sharp}}_\mathrm{F}^2
+ \norm{\mP_{\vu_\sharp} \va \vb^* \mP_{\vv_\sharp^\perp}}_\mathrm{F}^2
) \\
&= |\va^* \vu_\sharp|^4 |\vb^* \vv_\sharp|^4
+ |\va^* \vu_\sharp|^2 |\vb^* \vv_\sharp|^4 \norm{\mP_{\vu_\sharp^\perp} \va}_2^2
+ |\va^* \vu_\sharp|^4 |\vb^* \vv_\sharp|^2 \norm{\mP_{\vv_\sharp^\perp} \vb}_2^2.
\end{align*}
Since $\va \sim \mathcal{CN}(\vzero,\mId_{d_1})$ and $\vb \sim \mathcal{CN}(\vzero,\mId_{d_2})$ are independent, $\va^* \vu_\sharp$, $\vb^* \vv_\sharp$, $\mP_{\vu_\sharp^\perp} \va$, and $\mP_{\vv_\sharp^\perp} \vb$ are all mutually independent.  Therefore, exploiting this independence, one can show that the expectation is upper-bounded by
\begin{align*}
\E \norm{\mathcal{P}_T(\va \va^* \vu_\sharp \vv_\sharp^* \vb \vb^*)}_\mathrm{F}^2
& \leq 2 \norm{\vu_\sharp}_2^2 \norm{\vv_\sharp}_2^2 (2 + d_1 + d_2).
\end{align*}
By Jensen's inequality, the second expectation in \eqref{eq:ubrc_eq1} is upper-bounded by
\begin{equation}
\label{eq:ubrc_eq2}
\E
\Big\|
\frac{1}{\sqrt{M}}
\sum_{m=1}^M \epsilon_m
\mathcal{P}_{T^\perp}(\va_m \va_m^* \vu_\sharp \vv_\sharp^* \vb_m \vb_m^*)
\Big\|
\leq
\Big(
\E
\Big\|
\frac{1}{\sqrt{M}}
\sum_{m=1}^M \epsilon_m
\mathcal{P}_{T^\perp}(\va_m \va_m^* \vu_\sharp \vv_\sharp^* \vb_m \vb_m^*)
\Big\|^p
\Big)^{1/p}
\end{equation}
for all $p \in 2\mathbb{N}$.
To upper bound the right-hand side of \eqref{eq:ubrc_eq2}, we apply Theorem~\ref{thm:ncrosenthal} for
\[
\mY_m = \epsilon_m \mathcal{P}_{T^\perp}(\va_m \va_m^* \vu_\sharp \vv_\sharp^* \vb_m \vb_m^*)
= \epsilon_m \va_m^* \vu_\sharp \vv_\sharp^* \vb_m \mP_{\vu_\sharp^\perp} \va_m \vb_m^* \mP_{\vv_\sharp^\perp}, \quad m=1,\dots,M,
\]
with some $p \in \mathbb{N}$ that satisfies $p \geq 2$.  Note that $\E \mY_m = \vzero$ for all $m=1,\dots,M$.
By direct computation, we obtain
\[
\E \mY_m \mY_m^* = \norm{\vu_\sharp}_2^2 \norm{\vv_\sharp}_2^2 \mathrm{tr}(\mP_{\vv_\sharp^\perp}) \mP_{\vu_\sharp^\perp}
\quad \text{and} \quad
\E \mY_m^* \mY_m = \norm{\vu_\sharp}_2^2 \norm{\vv_\sharp}_2^2 \mathrm{tr}(\mP_{\vu_\sharp^\perp}) \mP_{\vv_\sharp^\perp},
\quad m=1,\dots,M.
\]
Therefore, 
\[
\Big\| \sum_{m=1}^M \E \mY_m \mY_m^* \Big\|^{1/2}
\vee
\Big\| \sum_{m=1}^M \E \mY_m^* \mY_m \Big\|^{1/2}
\leq \norm{\vu_\sharp}_2 \norm{\vv_\sharp}_2 \sqrt{M(d_1+d_2)}.
\]
Since the spectral norm of $\mY_m$ is upper-bounded by
\[
\norm{\mY_m} = |\va_m^* \vu_\sharp| |\vb_m^* \vv_\sharp| \norm{\mP_{\vu_\sharp^\perp} \va_m}_2 \norm{\mP_{\vv_\sharp^\perp} \vb_m}_2
\leq 2 |\va_m^* \vu_\sharp| |\vb_m^* \vv_\sharp| (\norm{\mP_{\vu_\sharp^\perp} \va_m}_2^2 + \norm{\mP_{\vv_\sharp^\perp} \vb_m}_2^2),
\]
it follows that
\begin{align*}
(\E \norm{\mY_m}^p)^{1/p}
&\leq 2 (\E |\va_m^* \vu_\sharp|^p)^{1/p} \cdot (\E |\vb_m^* \vv_\sharp|^2)^{1/p} \cdot [\E (\norm{\mP_{\vu_\sharp^\perp} \va_m}_2^2 + \norm{\mP_{\vv_\sharp^\perp} \vb_m}_2^2)^p]^{1/p} \\
&\leq 2 (\E |\va_m^* \vu_\sharp|^p)^{1/p} \cdot (\E |\vb_m^* \vv_\sharp|^2)^{1/p} \cdot [\E (\norm{\va_m}_2^2 + \norm{\vb_m}_2^2)^p]^{1/p}.
\end{align*}
Since $\va_m^* \vu_\sharp \sim \mathcal{CN}(0,1)$ and $\vb_m^* \vv_\sharp \sim \mathcal{CN}(0,1)$, we have
\[
(\E |\va_m^* \vu_\sharp|^p)^{1/p} = (\E |\vb_m^* \vv_\sharp|^p)^{1/p} \leq C_1 \sqrt{p}.
\]
for a numerical constant $C_1$.
Since $2(\norm{\va_m}_2^2 + \norm{\vb_m}_2^2)$ is a chi-square random variable of the degree-of-freedom $2(d_1+d_2)$, it follows that for $p \geq 2$ we have
\[
(\E (\norm{\va_m}_2^2 + \norm{\vb_m}_2^2)^p)^{1/p}
\leq 2 d_1 + 2 d_2 + C_2 p
\]
for a numerical constant $C_2$.
By collecting these estimates, we obtain
\[
p \Big( \sum_{m=1}^M \E \norm{\mY_m}^p \Big)^{1/p}
\leq C_3 M^{1/p} p^2(d_1+d_2+p).
\]

Applying the above estimates to Theorem~\ref{thm:ncrosenthal} together with \eqref{eq:ubrc_eq2} provides
\begin{equation}
\label{eq:mbnd:ubrc_rank1}
\begin{aligned}
\E \, \Big\|
\frac{1}{\sqrt{M}}
\sum_{m=1}^M \epsilon_m
\mathcal{P}_{T^\perp}(\va_m \va_m^* \vu_\sharp \vv_\sharp^* \vb_m \vb_m^*)
\Big\|
\leq C_4 \left(\sqrt{p(d_1 + d_2)} + M^{1/p-1/2} p^2(d_1+d_2+p)\right).
\end{aligned}
\end{equation}
As we set $p = \log M$, \eqref{eq:mbnd:ubrc_rank1} implies
\begin{align*}
\E \, \Big\|
\frac{1}{\sqrt{M}}
\sum_{m=1}^M \epsilon_m
\mathcal{P}_{T^\perp}(\va_m \va_m^* \vu_\sharp \vv_\sharp^* \vb_m \vb_m^*)
\Big\|
\leq C_5 \sqrt{(d_1 + d_2) \log M} \, \cdot \left( 1 + \frac{\sqrt{d_1+d_2} \, \log^{3/2} M}{\sqrt{M}} \right).
\end{align*}
Then \eqref{eq:sampl_comp_rank1} implies that the right-hand side is further upper bounded by $C_6 \sqrt{d_1+d_2} \log M$.
Finally, \eqref{eq:res:lemma:ubrc_rank1} is obtained by plugging in these upper estimates to \eqref{eq:ubrc_eq1}.

%\end{proof} %[Proof of Lemma~\ref{lemma:ubrc_rank1}]

%\bibliographystyle{IEEEtran}
%\bibliographystyle{imaiai}
%\bibliography{IEEEabrv,myrefs,jrom-refs}

\ifx\undefined\BySame
\newcommand{\BySame}{\leavevmode\rule[.5ex]{3em}{.5pt}\ }
\fi
\ifx\undefined\textsc
\newcommand{\textsc}[1]{{\sc #1}}
\newcommand{\emph}[1]{{\em #1\/}}
\let\tmpsmall\small
\renewcommand{\small}{\tmpsmall\sc}
\fi

\end{document}